\pdfoutput=1
\documentclass[runningheads]{llncs}

\usepackage{wrapfig}
\usepackage{float}

\usepackage{graphicx}
\usepackage{tikz}
\usepackage{subfigure}
\usepackage{makeidx}
\usepackage{pgf}
\usepackage{gastex}
\usepackage{adjustbox}

\usetikzlibrary{arrows,
	automata,
	calc,
	decorations,
	decorations.pathreplacing,
	positioning,
	shapes}

\usepackage[utf8]{inputenc}
\usepackage{amsmath}
\usepackage{amsfonts}
\usepackage{amssymb}
\usepackage{mathptmx}
\usepackage{eulervm}
\usepackage{eufrak}

\makeatletter
\renewcommand{\@Opargbegintheorem}[4]{%
  #4\trivlist\item[\hskip\labelsep{#3#2\@thmcounterend}]}
\makeatother

\newcommand{\MEC}[1]{\mathit{MEC}(#1)}

\newcommand{\effuncov}{\mathsf{eff}_{\mathsf{unc}}}
\newcommand{\effunc}{\effuncov}
\newcommand{\effcov}{\mathsf{eff}_{\mathsf{cov}}}

\newcommand{\freq}[2]{\mathit{freq}_{#1}(#2)}

\newcommand{\Nat}{\mathbb{N}}
\newcommand{\Real}{\mathbb{R}}

\newcommand{\Rat}{\mathbb{Q}}
\newcommand{\Rational}{\Rat}

\def\Pr{\mathrm{Pr}}

\newcommand{\eqdef}{\stackrel{\text{\rm \tiny def}}{=}}

\newcommand{\init}{\mathsf{init}}

\newcommand{\StAct}{\mathit{StAct}}
\newcommand{\SA}{\StAct}

\newcommand{\Act}{\mathit{Act}}

\newcommand{\residual}[2]{\mathit{res}(#1,#2)}

\newcommand{\yes}{\mathit{yes}}
\newcommand{\no}{\mathit{no}}
\newcommand{\choice}{\mathit{choice}}

\newcommand{\sched}{\mathfrak{S}}

\newcommand{\tsched}{\mathfrak{T}}
\newcommand{\usched}{\mathfrak{U}}
\newcommand{\vsched}{\mathfrak{V}}

\newcommand{\Ssched}{\Sigma}

\newcommand{\Usched}{\Upsilon}

\newcommand{\last}{\mathit{last}}

\newcommand{\Cause}{\mathsf{Cause}}

\newcommand{\Effect}{\mathsf{Eff}}
\newcommand{\Eff}{\Effect}
\newcommand{\effect}{\mathsf{eff}}
\newcommand{\eff}{\effect}
\newcommand{\noeff}{\mathsf{noeff}}
\newcommand{\CanCause}{\mathsf{CanCause}}
\newcommand{\CanCau}{\CanCause}

\newcommand{\noeffc}{\noeff_{\mathsf{fp}}}
\newcommand{\noefffp}{\noeffc}
\newcommand{\noeffbot}{\noeff_{\mathsf{tn}}}
\newcommand{\noefftn}{\noeffbot}

\newcommand{\precision}{\operatorname{\mathit{precision}}}
\newcommand{\recall}{\operatorname{\mathit{recall}}}
\newcommand{\fscore}{\mathit{fscore}}

\newcommand{\relcov}{\recall}
\newcommand{\ratiocov}{\mathit{covrat}}
\newcommand{\covratio}{\ratiocov}
\newcommand{\covrat}{\ratiocov}

\newcommand{\ratio}[3]{\mathit{ratio}^{#1}_{#2}(#3)}

\newcommand{\wminMDP}[2]{#1_{[#2]}}
\newcommand{\wminMDPmax}[2]{#1^{\max}_{[#2]}}

\newcommand{\Until}{ \, \mathrm{U}\, }
\newcommand{\until}{\Until}

\newcommand{\PTIME}{\mathrm{P}}

\newcommand{\NP}{\mathrm{NP}}
\newcommand{\coNP}{\mathrm{coNP}}

\newcommand{\PSPACE}{\mathrm{PSPACE}}
\newcommand{\NPSPACE}{\mathrm{NPSPACE}}

\def\cC{\mathcal{C}}

\def\cE{\mathcal{E}}

\def\cG{\mathcal{G}}

\def\cK{\mathcal{K}}

\def\cM{\mathcal{M}}
\def\cN{\mathcal{N}}

\newcommand{\Ende}{\hfill $\lhd$}

\spnewtheorem{notation}[theorem]{Notation}{\bfseries}{\upshape}
\spnewtheorem{algorithm}[theorem]{Algorithm}{\bfseries}{\upshape}
\usepackage{thmtools}
\newenvironment{proofsketch}{\proof}{\endproof}

\usepackage[paperheight=235mm,paperwidth=155mm,textwidth=12.2cm,textheight=19.3cm]{geometry}
\usepackage{bbding}
\makeatletter
\RequirePackage[bookmarks,unicode,colorlinks=true]{hyperref}%
\def\@citecolor{blue}%
\def\@urlcolor{blue}%
\def\@linkcolor{blue}%

\def\orcidID#1{\smash{\href{http://orcid.org/#1}{\protect\raisebox{-1.25pt}{\protect\includegraphics{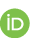}}}}}
\makeatother

\usepackage{xcolor}

\begin{document}

\title{On probability-raising causality
		  \\ in Markov decision processes
	  \thanks{This work was funded by DFG grant 389792660 as part of TRR~248, the Cluster of Excellence EXC 2050/1 (CeTI, project ID 390696704, as part of Germany’s Excellence Strategy), DFG-projects BA-1679/11-1 and BA-1679/12-1,and the RTG QuantLA (GRK 1763).}
  }
\author{Christel Baier(\Envelope) \orcidID{0000-0002-5321-9343} \and Florian Funke \orcidID{0000-0001-7301-1550} \and Jakob Piribauer(\Envelope) \orcidID{0000-0003-4829-0476} \and  Robin Ziemek(\Envelope) \orcidID{0000-0002-8490-1433}}
\authorrunning{Baier et al.}
\institute{Technische Universität Dresden\\
	\email{\{christel.baier, florian.funke, jakob.piribauer,robin.ziemek\}@tu-dresden.de}
}

\maketitle

\begin{abstract}
	The purpose of this paper is to introduce a notion of causality in Markov decision processes based on the probability-raising principle and to analyze its algorithmic properties.
	The latter includes algorithms for checking cause-effect relationships and the existence of probability-raising causes for given effect scenarios.
	Inspired by concepts of statistical analysis, we study quality measures 
	(recall, coverage ratio and f-score)
	for causes and develop algorithms for their computation.
	Finally, the computational complexity for finding optimal causes with respect to these measures is analyzed.
\end{abstract}

\noindent \textbf{Related version:} This is the extended version of the conference version accepted for publication at FoSSaCS 2022.

\section{Introduction}

As modern software systems control more and more aspects of our everyday lives, they grow increasingly complex.
Even small changes to a system might cause undesired or even disastrous behavior.
Therefore, the goal of modern computer science does not only lie in the development of powerful and versatile systems, but also in providing comprehensive techniques to understand these systems.
In the area of formal verification,
counterexamples, invariants and related certificates are often used to provide a verifiable justification that a system does or does not behave according to a specification (see e.g., \cite{MaPn95,CGP99,Namjoshi01}).
These, however, provide only elementary insights on the system behavior. Thus, there is a growing demand for a deeper understanding on \emph{why} a system satisfies or violates a specification and \emph{how} different components influence the performance.
The analysis of causal relations between events occurring during the execution of a system can lead to such understanding.
The majority of prior work in this direction relies on causality notions based on Lewis' counterfactual principle \cite{Lewis1973} stating the effect would not have occurred if the cause would not have happened.
A prominent formalization of the counterfactual principle is given by Halpern and Pearl \cite{HalpernP2001} via structural equation models.
This inspired formal definitions of causality and related notions of blameworthiness and responsibility in Kripke and game structures 
(see, e.g., \cite{ChocklerHK2008,BeerBCOT12,Chockler16,YazdanpanahDas16,FriedenbergHalpern19,YazdanpanahDasJamAleLog19,IJCAI21}).

In this work, we approach the concept of causality in a probabilistic setting, where we focus on the widely accepted \emph{probability-raising principle} which has its roots in philosophy \cite{Reichenbach56,Suppes70,Eells91,Hitchcock-Handbook} 
and has been refined by Pearl \cite{Pearl09} for causal and probabilistic reasoning in intelligent systems.
The different notions of probability-raising cause-effect relations 
discussed in the literature share the following two main principles:
\vspace{-4pt}
\begin{description}
\item [(C1)] Causes raise the probabilities for their effects,
     informally expressed by the requirement
    ``$\Pr( \, \text{effect} \, | \, \text{cause} \, ) >
        \Pr( \, \text{effect} \, )$''.
\item [(C2)] Causes must happen before their effects.
\end{description}
\vspace{-4pt}
Despite the huge amount of work on probabilistic causation in other disciplines, research on probability-raising causes in the context of formal methods is comparably rare and has concentrated on Markov chains (see, e.g., \cite{KleinbergM2009,Kleinberg2012,ATVA21} and the discussion of related work in Section \ref{sec:related-work}).
To the best of our knowledge,
probabilistic causation for probabilistic operational models with nondeterminism has not been studied before.

We formalize the principles (C1) and (C2) for Markov decision processes (MDPs), a standard operational model combining probabilistic and non-deterministic behavior, and  concentrate on reachability properties where both cause and effect are given as sets of states.
Condition (C1) can be interpreted in two natural ways in this setting:
On one hand, the probability-raising property can be locally required for each element of the cause.
Such causes are called \emph{strict probability-raising (SPR) causes} in our framework. 
This interpretation is especially suited when the task is to identify system states that have to be avoided for lowering the effect probability.
On the other hand, one might want to treat the cause set globally as a unit in (C1) leading to the notion of \emph{global probability-raising (GPR) cause}.
Considering the cause set as a whole is better suited when further constraints are imposed on the candidates for cause set.
This might apply, e.g., when the set of non-terminal states of the given MDP is partitioned into sets of states $S_i$  under the control of an agent $i$, $1\leq i \leq k$. For the task  to identify which agent's decisions cause the effect only the subsets of $S_1,\ldots,S_k$ are candidates for causes.
Furthermore, global causes are more appropriate when causes are used for monitoring purposes under partial observability constraints as then the cause candidates are sets of indistinguishable states.

%%\subsection*{Tabelle}

\newcommand{\polytime}{poly-time}
\newcommand{\polyspace}{poly-space}

\begin{table}[t]
 \begin{center} 
 	\caption{Complexity results for MDPs and Markov chains (MC) with fixed effect set}
 \begin{adjustbox}{max width=\textwidth}
  \begin{tabular}{c||c|c|c|c|c|c}
    \multicolumn{1}{c}{\phantom{ GPR }} &
    \multicolumn{1}{c}{\phantom{ in $\PSPACE$ }} &
    \multicolumn{1}{c}{\phantom{covratio}} &
    \multicolumn{1}{c}{\phantom{ covratio }} &
    \multicolumn{1}{c}{\phantom{covratio}} &
    \multicolumn{1}{c}{\phantom{ recall = covratio }} &
    \multicolumn{1}{c}{\phantom{covratio}}   
    \\[-2.5ex]
    &
    \multicolumn{4}{c|}{for fixed set $\Cause$}
    &
    \multicolumn{2}{c}{find optimal cause}
    \\
    %\cline{2-5} \cline{6-7}
    &
    %\begin{tabular} check PR \\[-0.3ex] condition \end{tabular}
    &
    \multicolumn{3}{c|}{compute quality values}
    \\[-2ex]
    &
    \multicolumn{1}{c}{%
      \begin{tabular}{c} \\[-2.5ex]
          check PR \\[-0.3ex] condition \\[0.5ex] \end{tabular}}
    &
    \multicolumn{3}{c}{}
    &
    \multicolumn{1}{c}{%
      \begin{tabular}{r} \\[-2.5ex]
          covratio-optimal \\[-0.3ex] = recall-optimal \\[0.5ex] \end{tabular}}
    &
    \multicolumn{1}{c}{}
    \\[-3.8ex]
    &
    &
    \multicolumn{3}{c|}{(recall, covratio, f-score)} 
    &
    \phantom{recall = covratio} & f-score-optimal 
    \\[1ex]
    \hline
    \hline
    SPR &
    $\in \PTIME$ &
    \multicolumn{3}{c|}{\polytime} &
    \polytime
    &
    \begin{tabular}{c}
      \\[-2ex]
      \polyspace  \\
      \polytime \ for MC \\
      threshold problem $\in \NP\cap \coNP$ 
      %%and $\in \PTIME$ for MC
      \\[0.5ex]
    \end{tabular}
    \\
    \hline
    GPR &
    \begin{tabular}{c}
      \\[-2ex]
      $\in \PSPACE$  \\
      and $\in \PTIME$ for MC
      \\[0.5ex]
    \end{tabular}
    &
    \multicolumn{3}{c|}{\polytime} &
    \multicolumn{2}{c}{%
      \begin{tabular}{c}
         \\[-2ex]
         \polyspace \\
         threshold 
         problems $\in \PSPACE$ and $\NP$-hard 
         \\
         and 
         $\NP$-complete for MC
      \\[0ex]
      \end{tabular}}  
  \end{tabular}
  \end{adjustbox}
 \end{center} 
 \vspace{-18pt}
 \label{tabelle:summary-of-results}
\end{table}

Different causes for an effect according to our definition can differ substantially regarding how well they predict the effect and how well the executions exhibiting the cause cover the executions showing the effect.
Taking inspiration from measures used in statistical analysis (see, e.g., \cite{Powers-fscore}), we introduce quality measures that allow us to compare causes and to look for optimal causes: 
The \emph{recall} captures the probability that the effect is indeed preceded by the cause.
The \emph{coverage-ratio} quantifies the fraction of the probability that cause and effect are observed and the probability that the effect but not the cause is observed.
Finally, the \emph{f-score}, a widely used quality measure for binary classifiers, is the harmonic mean of recall and precision, i.e., the probability that the cause is followed by the effect.

\vspace{0.5ex}\noindent\textbf{Contributions.}
The goal of this work are the mathematical and algorithmic foundations
of probabilistic causation in MDPs based on (C1) and (C2).
We introduce strict and global probability-raising causes in MDPs (Section \ref{sec:SPR-GPR}). Algorithms are provided to check whether given cause and effect sets satisfy (one of) the probability-raising conditions (Section \ref{sec:SPR_check} and \ref{sec:check-GPR}) and  to check the existence of causes for a given effect (Section \ref{sec:SPR_check}).
In order to evaluate the coverage properties of a cause, we subsequently introduce the above-mentioned quality measures (Section \ref{sec:acc-measures}).
We give algorithms for computing these values for given cause-effect relations (Section \ref{sec:comp-acc-measures-fixed-cause}) and characterize the computational complexity of finding optimal causes wrt. to the different measures (Section \ref{sec:opt-PR-causes}).
Table \ref{tabelle:summary-of-results} summarizes our complexity results.
Omitted proofs can be found in the appendix.

%%%%%%%%%%%%%%%%%%%%%%%%%%%%%%%%%%%%%%%%%%%%%%%%%%%%%%%%%%%%%%%%%%%%%%%%%%

\section{Preliminaries}

\label{sec:prelim}
Throughout the paper, we will assume some familiarity with basic concepts of
Markov decision processes. Here, we only present a brief summary of the notations used in the paper. For more details, we refer to \cite{Puterman,BaierK2008,Kallenberg20}.

A \emph{Markov decision process (MDP)} is a tuple $\cM=(S, \Act,P,\init)$ where
  $S$ is a finite set of states,
$\Act$  a finite set of actions,
$\init \in S$  the initial state and
$P : S \times \Act \times S \to [0,1]$  the transition probability function such that $\sum_{t\in S}P(s, \alpha,t) \in \{0,1 \}$ for all states $s \in S$ and actions $\alpha \in \Act$.
An action $\alpha$ is \emph{enabled} in state $s \in S$ if $\sum_{t\in S}P(s, \alpha,t)=1$. We define $\Act(s) = \{\alpha \mid \alpha \text{ is enabled in } s\}$.
A state $t$ is \emph{terminal} if $\Act(t) = \emptyset$.
A Markov chain (MC) is a special case of an MDP where $\Act$ is a singleton
(we then write $P(s,u)$ rather than $P(s,\alpha,u)$).
A \emph{path} in an MDP $\cM$ is a (finite or infinite) alternating sequence 
$\pi=s_0 \, \alpha_0 \, s_1 \, \alpha_1 \, s_2 \dots \in (S \times \Act)^* \cup (S \times \Act)^\omega$ such that  $P(s_i, \alpha_{i},s_{i+1})>0$ for all indices $i$. A path is called maximal if it is infinite or finite and ends in a terminal state.
An MDP can be interpreted as a
Kripke structure in which transitions go from states to probability distributions over states.

A \emph{(randomized) scheduler} $\sched$ is a function that maps each finite non-maximal path $s_0 \alpha_0  \dots \alpha_{n-1} s_n$ to a distribution over $\Act(s_n)$. $\sched$ is called deterministic if $\sched(\pi)$ is a Dirac distribution for all finite non-maximal paths $\pi$.
If the chosen action only depends on the last state of the path, 
 $\sched$ is called \emph{memoryless}. We write MR for the class of memoryless (randomized)  and MD for the class of memoryless deterministic schedulers. \emph{Finite-memory} schedulers are those that are representable by a finite-state automaton.

The scheduler $\sched$ of $\cM$ induces a (possibly infinite) Markov chain.
We write $\Pr^{\sched}_{\cM,s}$ for the standard probability measure on measurable sets of maximal paths in the Markov chain induced by $\sched$ with initial state $s$. If $\varphi$ is a measurable set of maximal paths, then $\Pr^{\max}_{\cM,s}(\varphi)$ and $\Pr^{\min}_{\cM,s}(\varphi)$ denote the supremum resp. infimum of the probabilities for $\varphi$ under all schedulers. We use the abbreviation $\Pr^{\sched}_{\cM}=\Pr^{\sched}_{\cM,\init}$ and notations $\Pr^{\max}_{\cM}$ and $\Pr^{\min}_{\cM}$ for extremal probabilities. Analogous notations will be used for expectations. So, if $f$ is a random variable, then, e.g., $\mathrm{E}^{\sched}_{\cM}(f)$ denotes the expectation of $f$ under $\sched$ and $\mathrm{E}^{\max}_{\cM}(f)$ its supremum over all schedulers.
We use LTL-like temporal modalities such as $\Diamond$ (eventually) and $\Until$ (until) to denote path properties.
For $X,T \subseteq S$ the formula $X \Until T$ is satisfied by paths $\pi = s_0 s_1 \dots $ such that there exists $j \geq 0$ such that for all $i<j: s_i \in X$ and $s_j \in T$
and $\Diamond T = S \Until T$.
It is well-known that $\Pr^{\min}_{\cM}(X \Until T)$ and
$\Pr^{\max}_{\cM}(X \Until T)$ and corresponding optimal MD-schedulers are computable in polynomial time. 

If $s\in S$ and $\alpha\in \Act(s)$, then $(s,\alpha)$ is said to be a state-action pair of $\cM$. 
An \emph{end component} (EC) of an MDP $\cM$ is a strongly connected sub-MDP containing at least one state-action pair. ECs will be often identified with the set of their state-action pairs. An EC $\cE$ is called maximal (abbreviated MEC) if there is no proper superset $\cE'$ of (the set of state-action pairs of) $\cE$ which is an EC.

\section{Strict and global probability-raising causes}

\label{sec:SPR-GPR}

We now provide formal definitions for cause-effect relations
in MDPs which rely on the probability-raising (PR) principle
as stated by (C1) and (C2) in the introduction.
We focus on the case where both causes and effects
are state properties, i.e., sets of states.

In the sequel, let
$\cM = (S,\Act,P,\init)$ be an MDP
and $\Effect \subseteq S \setminus \{\init\}$ 
a nonempty set of terminal states.
(Dealing with a fixed effect set, the assumption that all effect states are terminal is justified by (C2).)
Furthermore, we may assume that every state $s\in S$ is reachable from $\init$.
Proofs for the results of this section are provided in Appendix \ref{app:section_3}.

We consider here two variants of the probability-raising condition: the global setting treats the set $\Cause$ as a unit, while the strict view requires the probability-raising condition for all states in $\Cause$ individually.

\begin{definition}[Global and strict probability-raising cause (GPR/SPR cause)]
	\label{def:GPR} 
        \label{def-PR-causes} 
	        Let $\cM$ and $\Effect$ be as above and
        $\Cause$ a nonempty subset of $S \setminus \Effect$.
	Then, $\Cause$ is said to be a
        \emph{GPR cause} for
	$\Effect$ iff the following two conditions (G) and (M) hold:
	\begin{enumerate}
		\item [(G)]
		For each scheduler $\sched$ where
		$\Pr^{\sched}_{\cM}( \Diamond \Cause) >0$:
		\begin{equation}
			\label{GPR}  
			\Pr^{\sched}_{\cM}(\ \Diamond \Effect \ | \ \Diamond \Cause \ )
			\ > \ \Pr^{\sched}_{\cM}(\Diamond \Effect).
			\tag{\text{GPR}}
		\end{equation}  
		\item [(M)]
		For each $c\in \Cause$, there is a scheduler $\sched$
		with $\Pr^{\sched}_{\cM}( (\neg \Cause) \Until c ) >0$.
	\end{enumerate}
        	$\Cause$ is called
        an \emph{SPR cause} for
	$\Effect$ iff (M) and 
        the following condition (S) hold:
	\begin{enumerate}
		\item [(S)]
		For each state $c\in \Cause$ and each
		scheduler $\sched$ where $\Pr^{\sched}_{\cM}( (\neg \Cause) \Until c ) >0$:
		\begin{equation}
			\label{SPR}  
			\Pr^{\sched}_{\cM}(\ \Diamond \Effect \ | \
			(\neg \Cause) \Until c \ )
			\ > \ \Pr^{\sched}_{\cM}(\Diamond \Effect).
			\tag{\text{SPR}}
		\end{equation}   
	\end{enumerate}
\end{definition}

Condition (M) can be seen as a minimality requirement as states $c\in \Cause$ that are not accessible from $\init$ without traversing other states in $\Cause$ could be omitted without affecting the true positives (events where an effect state is reached after visiting a cause state, ``covered effects'') or false negatives (events where an effect state is reached without visiting a cause state before, ``uncovered effect'').
More concretely, whenever a set $C \subseteq S \setminus \Effect$ satisfies conditions (G) or (S) then the set $\Cause$ of states $c\in C$ where $\cM$ has a path from $\init$ satisfying $(\neg C)\Until c$ is a GPR resp. an SPR cause.

%%%%%%%%%%%%%%%%%%%%%%%%%%%%%%%%%%%%%%%%%%%%%%%%%%%%%%%%%%%%%%%%%%%%%%%%

\subsection{Examples and simple properties of probability-raising causes}

We first observe that SPR/GPR causes cannot contain the initial state 
$\init$, since otherwise an equality instead of an inequality would hold in \eqref{GPR}
and \eqref{SPR}.
Furthermore as a direct consequence of the definitions and using the equivalence of
the LTL formulas
$\Diamond \Cause$ and
$(\neg \Cause) \until \Cause$
we obtain:

\begin{lemma}[Singleton PR causes]
	\label{global-vs-strict-causes}
	If $\Cause$ is a singleton then
	$\Cause$ is a SPR cause for $\Effect$
	if and only if
	$\Cause$ is a GPR cause for $\Effect$.
\end{lemma}

As the event $\Diamond \Cause$ is a disjoint union of all events $(\neg \Cause) \Until c$ with $c\in \Cause$,
the probability for covered effects
$\Pr^{\sched}_{\cM}(\ \Diamond \Effect \ | \ \Diamond \Cause \ )$
is a weighted average of the
probabilities
$\Pr^{\sched}_{\cM}(\ \Diamond \Effect \ | \ (\neg \Cause)\Until c \ )$
for $c\in \Cause$.
This yields:

\begin{restatable}[Strict implies global]{lemma}{strictimpliesglobal}
	\label{lemma:strict-implies-global} 
	Every SPR cause for $\Effect$ is a GPR cause for $\Effect$.
\end{restatable}

\begin{example}[Non-strict GPR cause]
	\label{ex:non-strict-global-causes-in-MC}  
	{\rm
	  Consider the Markov chain $\cM$ depicted below
	  where the nodes represent states and the directed edges represent transitions labeled with their respective probabilities.
		Let $\Effect=\{\effect\}$.
%		
%		%
		Then,
                $\Pr_{\cM}(\Diamond \Effect) = \frac{1}{3} \, + \, \frac{1}{3}\cdot \frac{1}{4} \, + \, \frac{1}{12} 
                = \frac{1}{2}$,
                $\Pr_{\cM}( \Diamond \Effect  |  \Diamond c_1  ) = \Pr_{\cM,c_1}(\Diamond \effect) =1$ and
                $\Pr_{\cM}( \Diamond \Effect  |  \Diamond c_2  ) = \Pr_{\cM,c_2}(\Diamond \effect) =\frac{1}{4}$.
		Thus, $\{c_1\}$ is both an SPR and a GPR cause for $\Effect$, while $\{c_2\}$
		is not.
		The set $\Cause=\{c_1,c_2\}$ is a
                non-strict GPR cause for $\Effect$
		as:
		\begin{center}
				$\Pr_{\cM}(\ \Diamond \Effect \ | \ \Diamond \Cause \ ) 
				=
				(\frac{1}{3}+\frac{1}{3}\cdot \frac{1}{4}) / (\frac{1}{3}+\frac{1}{3})
				=
				(\frac{5}{12})/(\frac{2}{3})
				=
				\frac{5}{8}
				>
				\frac{1}{2}
				=
				\Pr_{\cM}(\Diamond \Eff)$.
		\end{center}

		The second condition (M) is obviously fulfilled.
		Non-strictness follows from the fact that
		the SPR condition does not hold for
		state $c_2$.\Ende
		\begin{center}			
			\resizebox{0.55\textwidth}{!}{%Non-strict global PR cause

\begin{tikzpicture}[->,>=stealth',shorten >=1pt,auto ,node distance=0.5cm, thick]
	\node[scale=1, state] (s0) {$\init$};
	\node[scale=1, state] (c1) [below left = 1 of s0] {$c_1$};
	\node[scale=1, state] (c2) [right =2 of c1] {$c_2$};
	\node[scale=1, state] (e) [left =2 of c1] {$\effect$};
	\node[scale=1, state] (f) [right = 2 of c2] {$\noeff$};
	
	\draw[<-] (s0) --++(-0.55,0.55);
	\draw (s0) -- (c1) node[below=0.2, pos=0.3,scale=1] {$1/3$};
	\draw (s0) -- (c2) node[below=0.2, pos=0.3,scale=1] {$1/3$};
	\draw (s0) -- (e) node[above, pos=0.5,scale=1] {$1/12$};
	\draw (s0) -- (f) node[pos=0.5,scale=1] {$1/4$};
	\draw (c1) -- (e) node[pos=0.3,above,scale=1] {$1$};
	\draw (c2) -- (f) node[below, pos=0.5,scale=1] {$3/4$};
	\draw (c2) to [out=200, in=340] (e) node[below right = .5, scale=1] {$1/4$};
\end{tikzpicture}}
		\end{center}
	}
\end{example}

\begin{example}[Probability-raising causes might not exist]
	\label{ex:no-prob-raising-cause}  
	{\rm  
                PR causes might not
		exist, even if $\cM$ is a Markov chain.
		This applies, e.g., to the Markov chain $\cM$ with two states
		$\init$ and $\effect$ where
		$P(\init,\effect)=1$ and the effect set $\Effect=\{\effect\}$.
		The only cause candidate is the singleton $\{\init\}$.
		However, the strict inequality in \eqref{GPR} or \eqref{SPR} does
		not hold for $\Cause =\{\init\}$.
                The same phenomenon occurs if all non-terminal states of a Markov chain
                reach the effect states with the same probability.
                In such cases, however, the non-existence of PR causes
                is well justified as the events $\Diamond \Effect$ and
                $\Diamond \Cause$
                are stochastically independent for every set
                $\Cause \subseteq S \setminus \Effect$.
		\Ende
	}
\end{example}

%%%%%%%%%%%%%%%%%%%%%%%%%%%%%%%%%%%%%%%%%%%%%%%%%%%%%%%%%%%%%%%%%%%%%%%%%%%%%  

\begin{remark}[Memory needed for refuting PR condition]
  \label{rem:memory_necessary}
	\label{rem:memory-needed}
  Let $\cM$ be the MDP in Figure~\ref{fig:memory-needed}, where the notation is similar to Example \ref{ex:no-prob-raising-cause} with the addition of actions $\alpha, \beta$ and $\gamma$.
  Let $\Cause = \{c\}$
  and $\Effect=\{\eff\}$.
  Only state $s$ has a nondeterministic choice.
  $\Cause$ is not an PR cause. To see this,
  regard the deterministic scheduler $\tsched$ that
  schedules $\beta$ only for the first visit of $s$ and $\alpha$ for
  the second visit of $s$.
  Then:
  \begin{center}
  $
    \Pr^{\tsched}_{\cM}(\Diamond \eff)
    \ = \ \frac{1}{2} \cdot \frac{1}{2} \, + \, 
          \frac{1}{2} \cdot \frac{1}{2} \cdot 1 \cdot \frac{1}{4}
    \ = \ \frac{5}{16}
    \ > \ \frac{1}{4} \ = \
    \Pr^{\tsched}_{\cM}(\Diamond \eff |\Diamond c)
    $
  \end{center}  
  Denote the MR schedulers reaching $c$ with positive probability as $\sched_{\lambda}$ with $\sched_{\lambda}(s)(\alpha)$ $=$ $\lambda$ and $\sched_{\lambda}(s)(\beta)=1{-}\lambda$  for some $\lambda \in \, [0,1[$. Then, $\Pr^{\sched_{\lambda}}_{\cM,s}(\Diamond \eff)  >0$ and:
\begin{center}
  	$
   \Pr^{\sched_{\lambda}}_{\cM}(\Diamond \eff)
   \ = \
   \frac{1}{2} \cdot \Pr^{\sched_{\lambda}}_{\cM,s}(\Diamond \eff)
   \ < \ 
   \Pr^{\sched_{\lambda}}_{\cM,s}(\Diamond \eff)
   \ = \
   \Pr^{\sched_{\lambda}}_{\cM,c}(\Diamond \eff)
   \ = \
   \Pr^{\sched_{\lambda}}_{\cM}(\Diamond \eff |\Diamond c)
    $
 \end{center}
  Thus, the SPR/GPR condition holds for $\Cause$
  and $\Effect$ 
  under all memoryless schedulers reaching $\Cause$ with positive probability,
  although $\Cause$ is not an PR cause.
  \Ende
\end{remark}  

%%%%%%%%%%%%%%%%%%%%%%%%%%%%%%%%%%%%%%%%%%%%%%%%%%%%%%%%%%%%%%%%%%%%%%%%%%%%%

\begin{figure}[t]
	\centering
	\begin{minipage}{0.41\textwidth}
		\centering
		\resizebox{\textwidth}{!}{
			%Memory needed

\begin{tikzpicture}[scale=1,->,>=stealth',auto ,node distance=0.5cm, thick]
	\tikzstyle{round}=[thin,draw=black,circle]
	
	\node[scale=1, state] (init) {$\init$};
	\node[scale=1, state, below=1.5 of init] (noeff) {$\noeff$};
	\node[scale=1, state, right=1.5 of init] (s) {$s$};
	\node[scale=1, state, below = 1.5 of s] (eff) {$\eff$};
	\node[scale=1, state, right=1.5 of eff] (c) {$c$};
	
	\draw[<-] (init) --++(-0.55,0.55);
	\draw[color=black ,->] (init) edge  node [pos=0.5,above] {$\gamma \mid 1/2$} (s) ;
	\draw[color=black ,->] (init)  edge  node [pos=0.5, left] {$\gamma \mid 1/2$} (noeff) ;
	\draw[color=black ,->] (s) edge  node [anchor=center] (n1) {} node [pos=0.8,right] {$3/4$} (noeff) ;
	\draw[color=black,->] (s) edge[out=260, in=100] node [anchor=center] (m1) {} node [pos=0.85,left] {$1/4$} (eff) ;
	\draw[color=black,->] (s) edge[out=280, in=80] node [anchor=center] (m2) {} node [pos=0.85,right] {$1/2$} (eff) ;
	\draw[color=black,->] (s) edge node [anchor=center] (n2) {} node [pos=0.8,left] {$1/2$} (c) ;
	\draw[color=black,->] (c) edge[out=90, in=0] node [pos=0.5, right] {$\gamma \mid 1$} (s) ;
	\draw[color=black , very thick, -] (n1.center) edge [bend right=45] node [pos=0.3] {$\alpha$} (m1.center);
	\draw[color=black, very thick, -] (m2.center) edge [bend right=45] node [pos=0.3] {$\beta$} (n2.center);
	
\end{tikzpicture}
		}
		\caption{MDP $\cM$ from Remark \ref{rem:memory-needed}}
		\label{fig:memory-needed}
	\end{minipage}\hspace{25pt}
	\begin{minipage}{0.41\textwidth}
		\centering
		\resizebox{\textwidth}{!}{
			%Randomization needed

\begin{tikzpicture}[scale=1,->,>=stealth',auto ,node distance=0.5cm, thick]
	\tikzstyle{round}=[thin,draw=black,circle]
	
	\node[scale=1, state] (init) {$\init$};
	\node[scale=1, state, below=1.25 of init] (pre) {$\effuncov$};
	\node[scale=1, state, right=3.75 of init] (c) {$c$};
	\node[scale=1, state, right=1.25 of pre] (t) {$\noeff$};
	\node[scale=1, state, below=1.25 of c] (post) {$\effcov$};
	
	\draw[<-] (init) --++(-0.55,0.55);
	\draw[color=black ,->] (init) edge  node [very near start, anchor=center] (n5) {} node [pos=0.5,above] {$1/2$} (c) ;
	\draw[color=black ,->] (init)  edge  node [very near start, anchor=center] (n0) {} node [pos=0.5, left] {$\alpha \mid 1$} (pre) ;
	\draw[color=black ,->] (c) edge  node [near start, anchor=center] (m5) {} node [pos=0.5,right] {$1/2$} (post) ;
	\draw[color=black,->] (init) edge node [near start, anchor=center] (n6) {} node [pos=0.5,right] {$1/2$} (t) ;
	\draw[color=black , very thick, -] (n6.center) edge [bend right=45] node [pos=0.3] {$\beta$} (n5.center);
	
	\draw[color=black ,->] (c)  edge  node [near start, anchor=center] (m6) {} node [pos=0.5,left] {$1/2$} (t) ;
	\draw[color=black, very thick, -] (m6.center) edge [bend right=45] node [pos=0.3] {$\tau$} (m5.center);
	
\end{tikzpicture}
		}
		\caption{MDP $\cM$ from Remark \ref{rem:randomization-needed}}
		\label{fig:randomization-needed}
	\end{minipage}
\end{figure}

%%%%%%%%%%%%%%%%%%%%%%%%%%%%%%%%%%%%%%%%%%%%%%%%%%%%%%%%%%%%%%%%%%%%%%%%%%%%%  

\begin{remark}[Randomization needed for refuting PR condition]
	\label{rem:randomization-needed}
	{\rm
	  Consider the MDP $\cM$ of Figure \ref{fig:randomization-needed}.
		Let $\Effect = \{\effuncov,\effcov\}$ and $\Cause = \{c\}$.
		The two MD-schedulers $\sched_{\alpha}$ and $\sched_{\beta}$ 
		that select $\alpha$ resp. $\beta$ for the initial
		state $\init$ are the only deterministic schedulers.
		As $\sched_{\alpha}$ does not reach $c$, it is irrelevant for
		the 
                SPR or GPR condition.
		$\sched_{\beta}$ satisfies \eqref{SPR} and \eqref{GPR} as
                $\Pr^{\sched_{\beta}}_{\cM}(\Diamond \Effect|\Diamond c)
                = \frac{1}{2} > \frac{1}{4}=
                 \Pr^{\sched_{\beta}}_{\cM}(\Diamond \Effect)$.
				The MR scheduler $\tsched$ which selects $\alpha$ and $\beta$
		with probability $\frac{1}{2}$ in $\init$
		reaches $c$ with positive probability and violates
		\eqref{SPR} and \eqref{GPR} as 
                $\Pr^{\tsched}_{\cM}(\Diamond \Effect|\Diamond c)
                = \frac{1}{2} < \frac{5}{8}
                = \frac{1}{2} + \frac{1}{2} \cdot \frac{1}{2} \cdot \frac{1}{2}
                = \Pr^{\tsched}_{\cM}(\Diamond \Effect)$.
		\Ende
	}
\end{remark}

\begin{remark}[Cause-effect relations for regular classes of schedulers]
\label{remark:regular-classes-of-schedulers}
  The definitions of PR causes in MDPs impose constraints for all
  schedulers reaching a cause state.
  This condition is fairly strong and can often lead to the phenomenon
  that no PR cause exists.
  Replacing $\cM$ with an MDP resulting from the
  synchronous parallel composition of $\cM$ with a
  deterministic finite automaton
  representing a regular constraint on the scheduled state-action sequences
    (e.g., ``alternate between actions $\alpha$ and $\beta$ in state $s$''
  or ``take $\alpha$ on every third visit to state $s$ and actions $\beta$ or $\gamma$ otherwise'')
  leads to a weaker notion of PR causality.
  This can be useful to obtain more detailed
  information on cause-effect relationships in special scenarios,
  be it at design time where multiple scenarios (regular classes of schedulers)
  are considered or
  for a post-hoc analysis where one seeks for the causes
  of an occurred effect and where information about the scheduled actions
  is extractable from log files or the information gathered by a monitor.
  \Ende
\end{remark}

\begin{remark}[Action causality and other forms of PR causality]
  \label{remark:action-causality}
  Our notions of PR causes are purely state-based with PR conditions that compare probabilities
  under the same scheduler.
  However, in combination with model transformations,
  the proposed notions of PR causes are
  also applicable for reasoning about other forms of PR causality.
  
  Suppose, 
  the task is to check whether taking action $\alpha$ in state $s$ raises the effect probabilities compared to never scheduling
  $\alpha$ in state $s$. Let $\cM_0$ and $\cM_1$ be copies of $\cM$ with the following modifications: In $\cM_0$, the only enabled
  action of state $s$ is $\alpha$, while in $\cM_1$ the enabled actions of state $s$ are the elements of
  $\Act_{\cM}(s)\setminus \{\alpha\}$. Let now $\cN$ be the MDP whose initial state has a single enabled action and moves
  with probability $1/2$ to 
  $\cM_0$ and $\cM_1$.
  Then, action $\alpha$ raises the effect probability in $\cM$ iff the initial state of $\cM_0$ consitutes an SPR cause in $\cN$.
  This idea can be generalized to check whether scheduler classes satisfying
  a regular constraint
  have higher effect probability compared to all other schedulers.
  In this case, we can deal with an MDP $\cN$ as above where
  $\cM_0$ and $\cM_1$ are defined as the synchronous product of deterministic finite
  automata and $\cM$.
  \Ende
\end{remark}

%%%%%%%%%%%%%%%%%%%%%%%%%%%%%%%%%%%%%%%%%%%%%%%%%%%%%%%%%%%%%%%%%%%%%%%%

\subsection{Related work}

\label{sec:related-work}

Previous work in the direction of probabilistic causation in stochastic operational models has mainly concentrated on Markov chains. Kleinberg \cite{KleinbergM2009,Kleinberg2012} introduced \emph{prima facie causes} in finite Markov chains where both causes and effects are formalized as PCTL state formulae, and thus they can be seen as sets of states as in our approach.
The correspondence of Kleinberg's PCTL constraints for prima facie causes and the strict probability-raising condition formalized using conditional probabilities has been worked out in the survey article \cite{ICALP21}.
Our notion of SPR causes corresponds to Kleinberg's prima facie causes, except for the minimality condition (M).
\'{A}brah\'{a}m et al \cite{Abraham-QEST2018} introduces a hyperlogic for Markov chains and gives a formalization of probabilistic causation in Markov chains as a hyperproperty, which is consistent with Kleinberg's prima facie causes, and with SPR causes up to minimality.
 Cause-effect relations in Markov chains where effects are $\omega$-regular properties 
 has been introduced in \cite{ATVA21}. It relies on strict probability-raising condition, but requires completeness in the sense that every path where the effect occurs has a prefix in the cause set. The paper \cite{ATVA21} permits a non-strict inequality in the SPR condition with the consequence that causes always exist, which is not the case for our notions.

The survey article \cite{ICALP21} introduces notions of global probability-raising causes for Markov chains where 
causes and effects can be path properties.
\cite{ICALP21}'s notion of \emph{reachability causes} in Markov chains directly corresponds to our notion GPR causes, the only difference being that \cite{ICALP21} deals with a relaxed minimality condition and requires that the cause set is reachable without visiting an effect state before. 
The latter is inherent in our approach as we suppose that all states are reachable and the effect states are terminal.

To the best of our knowledge, probabilistic causation in MDPs has not been studied before. The only work in this direction we are aware of is the recent paper by Dimitrova et al \cite{DimFinkbeinerTorfah-ATVA2020} on a hyperlogic, called PHL, for MDPs.
While the paper focuses on the foundation of PHL, it contains an example illustrating how action causality can be formalized as a PHL formula.
Roughly, the presented formula expresses that taking a specific action $\alpha$ increases the probability for reaching effect states.
Thus, it also relies on the probability-raising principle, but compares the ``effect probabilities'' under different schedulers (which either schedule $\alpha$ or not) rather than comparing probabilities under the same scheduler as in our PR condition.
However, as Remark \ref{remark:action-causality} argues, to some extent our notions of PR causes can reason about action causality as well.

There has also been work on causality-based explanations of counterexamples in probabilistic models \cite{LeitnerFischer-Leue-SAFECOMP11,LeitnerFischer-PhD15}. The underlying causality notion of this work, however, relies on the non-probabilistic counterfactual principle rather than the probability-raising condition. The same applies to the notions of forward and backward responsibility in stochastic games in extensive form introduced in the recent work \cite{IJCAI21}.

\section{Checking the existence of PR causes and the PR conditions}

\label{sec:check}

We now turn to algorithms for checking whether a given set $\Cause$ is an SPR or GPR cause for $\Effect$. As condition (M) of SPR and GPR causes is verifiable by standard model checking techniques in polynomial time, we concentrate on checking the probability-raising conditions (SPR) and (GPR). For Markov chains, both (SPR) and (GPR) can be checked in polynomial time by computing the corresponding probabilities. So, the interesting case is checking the PR conditions in MDPs. 
In case of SPR causality, this is closely related to the existence of PR causes and solvable in polynomial time (Section \ref{sec:check-SPR}), while checking the GPR condition is more complex and polynomially reducible to (the non-solvability of) a quadratic constraint system (Section \ref{sec:check-GPR}).
All proofs and omitted details to this section can be found in Appendix \ref{app:section_4}.
  
We start by stating that for the SPR and GPR condition, it suffices to consider schedulers minimizing the probability to reach an effect state from every cause state.

\begin{notation}
[MDP with minimal effect probabilities from cause candidates]
\label{notation:MDP-mit-min-prob-ab-cause-candidates}  
If $C \subseteq S$ then we write $\wminMDP{\cM}{C}$ for the MDP resulting from $\cM$ by 
removing all enabled actions of the states in $C$. Instead, $\wminMDP{\cM}{C}$ has a new action $\gamma$ that is enabled exactly in the states $s\in C$ with the transition probabilities $P_{\wminMDP{\cM}{C}}(s,\gamma,\eff)=\Pr^{\min}_{\cM,s}(\Diamond \Effect)$ and $P_{\wminMDP{\cM}{C}}(s,\gamma,\noeff)=1{-}\Pr^{\min}_{\cM,s}(\Diamond \Effect)$. Here, $\eff$ is a fixed state in $\Effect$ and $\noeff$ a (possibly fresh) terminal state not in $\Effect$. 
We write $\wminMDP{\cM}{c}$ if $C=\{c\}$ is a singleton.
\end{notation}

\begin{restatable}{lemma}{wmincriterionPRcauses}
  \label{lemma:wmin-criterion-PR-causes}    
 Let $\cM=(S,\Act,P,\init)$ be an MDP and $\Effect\subseteq S$ a set of terminal states.
  Let $\Cause \subseteq S\setminus \Effect$. Then, $\Cause$ is an SPR cause (resp. a GPR cause) for $\Effect$ in $\cM$ if and only if $\Cause$ is an SPR cause (resp. a GPR cause) for $\Effect$ in $\wminMDP{\cM}{\Cause}$.  
\end{restatable}

\subsection{Checking the strict probability-raising condition and the existence of causes}

\label{sec:check-SPR}
\label{sec:SPR_check}
\label{sec:check-SPR-condition}

The basis of both checking the existence of PR causes or checking the SPR condition for a given cause candidate is the following polynomial time algorithm to check whether the SPR condition
holds in a given state $c$ of $\cM$
for all schedulers $\sched$ with $\Pr^{\sched}_{\cM}(\Diamond c)>0$:

\begin{algorithm}
  \label{alg:SPR-check}
  Input: state $c \in S$, set of terminal states $\Eff \subseteq S$;
  Task: Decide whether \eqref{SPR} holds in $c$ for all schedulers $\sched$.
  
  Compute $w_c =\Pr^{\min}_{\cM,c}(\Diamond \Eff)$ and $q_s =\Pr^{\max}_{\wminMDP{\cM}{c},s}(\Diamond \Effect)$ for each state $s$ in $\wminMDP{\cM}{c}$ .
	\begin{enumerate}
		\item [1.] If $q_{\init} < w_c$, then return
		``yes, \eqref{SPR} holds for $c$''.
		
		\item [2.] If $q_{\init} > w_c$, then  return
		``no, \eqref{SPR} does not hold for $c$''.
		
		\item [3.]
		Suppose $q_{\init} = w_c$. Let
		$A(s) = \{\alpha \in \Act_{\wminMDP{\cM}{c}}(s) \mid q_s = \sum_{t\in \wminMDP{S}{c}} P_{\wminMDP{\cM}{c}}(s,\alpha,t)\cdot q_t\}$
		for each non-terminal state $s$.
		Let $\wminMDPmax{\cM}{c}$
                denote the sub-MDP of $\wminMDP{\cM}{c}$
		induced by the state-action pairs $(s,\alpha)$ where
		$\alpha \in A(s)$.
		\begin{enumerate}
			\item [3.1]
			  If $c$ is reachable from $\init$ in
                          $\wminMDPmax{\cM}{c}$,
			  then return
                          \mbox{``no,  \eqref{SPR} does not hold for $c$''.}

			\item [3.2]
			  If $c$ is not reachable from $\init$
                          in $\wminMDPmax{\cM}{c}$,
			then return ``yes,  \eqref{SPR} holds for $c$''.
		\end{enumerate}
	\end{enumerate}  
\end{algorithm}

\begin{restatable}{lemma}{soundnessSPRalgo} \label{soundness-SPR-algo}
	Algorithm \ref{alg:SPR-check} is sound and runs in polynomial time.
\end{restatable}

\begin{proofsketch}[Soundness]
  Let $\cN=\wminMDP{\cM}{c}$.
  Soundness is obvious in case 1.
  For case 2, consider a real number $\lambda$ with $1 > \lambda > \frac{w_c}{q_\init}$
  and MD-schedulers $\tsched$ and $\sched$
  realizing $\Pr^{\tsched}_{\cN,s}(\Diamond \Effect) = q_s$ and $\Pr^{\sched}_{\cN}(\Diamond c)>0$ for all states $s$.
  We can combine $\tsched$ and $\sched$ to a new MR-scheduler
  $\usched$ with the property that
  $\Pr^{\usched}_{\cN}(\Diamond t)=\lambda \Pr^{\tsched}_{\cN}(\Diamond t)+(1{-}\lambda) \Pr^{\sched}_{\cN}(\Diamond t)$ for all terminal states $t$ and for $t=c$.
  Then, $\usched$ witnesses a violation of \eqref{SPR}.
  For case 3.1 consider an MD-scheduler $\sched$ of $\wminMDPmax{\cM}{c}$
  where
	$c$ is reachable from $\init$ via a $\sched$-path and
        $\Pr^{\sched}_{\cN,s}(\Diamond \Effect)=q_s$ for all states $s$.
 	Then, \eqref{SPR} does not hold for $c$ in the scheduler $\sched$.
	In case 3.2 we have
	$\Pr^{\sched}_{\cN}(\Diamond c)=0$ for all schedulers $\sched$ for $\cN$
	with $\Pr^{\sched}_{\cN}(\Diamond \Effect)=q_{\init}=w_c$.
	But then $\Pr^{\sched}_{\cN}(\Diamond c)>0$ implies
	$\Pr^{\sched}_{\cN}(\Diamond \Effect) < w_c$ as required in \eqref{SPR}.
	For more details on the soundness see Appendix \ref{app:proofs_SPR_check}.    
        \qed
\end{proofsketch}        
  
%%%%%%%%%%%%%%%%%%%%%%%%%%%%%%%%%%%%%%%%%%%%%%%%%%%%%%%%%%%%%%%%%%%%%%%%

By applying Algorithm \ref{alg:SPR-check} to all states $c \in \Cause$ and standard algorithms to check the existence of a path satisfying $(\neg \Cause) \Until c$ for every state $c\in \Cause$, we obtain:

\begin{theorem}[Checking SPR causes]
	The problem ``given $\cM$, $\Cause$ and $\Effect$, check whether
	$\Cause$ is a SPR cause for $\Effect$ in $\cM$''
	is solvable in polynomial-time.
\end{theorem}  

%%%%%%%%%%%%%%%%%%%%%%%%%%%%%%%%%%%%%%%%%%%%%%%%%%%%%%%%%%%%%%%%%%%%%%%%

\begin{remark}[Memory requirements for refuting the SPR property]
	\label{MR-sufficient-SRP}  
	{\rm
	  As the soundness proof for Algorithm \ref{alg:SPR-check}
          shows:
	If $\Cause$ does not satisfy the SPR condition,
	then there is an MR-scheduler $\sched$ for
        $\wminMDP{\cM}{\Cause}$ witnessing the violation of \eqref{SPR}.
          Scheduler $\sched$ 
          corresponds to a finite-memory (randomized) scheduler $\tsched$
          with two memory cells for $\cM$:
``before $\Cause$'' (where $\tsched$ behaves as $\sched$)
and ``after $\Cause$'' (where $\tsched$ behaves as an MD-scheduler minimizing
the effect probability form every state).
		\Ende
	}
\end{remark}

%%%%%%%%%%%%%%%%%%%%%%%%%%%%%%%%%%%%%%%%%%%%%%%%%%%%%%%%%%%%%%%%%%%%%%%%

\begin{restatable}[Criterion for the existence of probability-raising causes]{lemma}{existencePRcauses}
	\label{lem:existence-of-prob-raising-causes}
	Let $\cM$ be an MDP and $\Effect$ a nonempty set of states. Then 
	$\Effect$ has an SPR cause in $\cM$ iff
	$\Effect$ has a GPR cause in $\cM$ iff
	there is a state $c_0\in S \setminus \Effect$ such that
	the singleton $\{c_0\}$ is an SPR cause (and therefore a GRP cause) for $\Effect$ in $\cM$.
    In particular, the existence of SPR/GPR causes can be checked
    with Algorithm \ref{alg:SPR-check} in polynomial time.
\end{restatable}  
The lemma can be derived from
Lemmata \ref{global-vs-strict-causes}, \ref{lemma:strict-implies-global}
and \ref{lemma:wmin-criterion-PR-causes} together with the implication ``(b) $\Longrightarrow$ (c)'' shown in Appendix \ref{app:proofs_SPR_check}.

%%%%%%%%%%%%%%%%%%%%%%%%%%%%%%%%%%%%%%%%%%%%%%%%%%%%%%%%%%%%%%%%%%%%%%%%%%%%%%%%

\subsection{Checking the global probability-raising condition}

\label{sec:check-GPR-condition}
\label{sec:check-GPR}

Throughout this section, we suppose that both the effect set $\Effect$ and the cause candidate $\Cause$ are fixed disjoint subsets of the state space of the MDP $\cM=(S,\Act,P,\init)$, and address the task to check whether $\Cause$ is a strict resp. global probability-raising cause for $\Effect$ in $\cM$. As the minimality condition (M) can be checked in polynomial time using a standard graph algorithm, we will concentrate on an algorithm to check the probability-raising condition (GPR).
We start by stating the main results of this section.

\begin{theorem}
  \label{thm:checking-GPR-in-poly-space}
	The problem ``given $\cM$, $\Cause$ and $\Effect$, check whether
	$\Cause$ is a GPR cause for $\Effect$ in $\cM$''
	is solvable in polynomial space.
\end{theorem}

In order to provide an algorithm, we perform a model transformation after which the violation of \eqref{GPR} by a scheduler $\sched$ can be expressed solely in terms of the expected frequencies of the state-action pairs of the transformed MDP under $\sched$. 
This allows us to express the existence of a scheduler witnessing the non-causality of $\Cause$ in terms of the satisfiability of a quadratic constraint system.
Thus, we can restrict the quantification in (G) to MR-schedulers in the transformed model.
We trace back the memory requirements to 
$\wminMDP{\cM}{\Cause}$ and to the original MDP $\cM$ yielding the second main result.
Still, memory can be necessary to witness non-causality (Remark \ref{rem:memory_necessary}). 

\begin{theorem}
  \label{thm:MR-sufficient-GPR}
        Let $\cM$ be an MDP with effect set $\Effect$ as before and
        $\Cause$ a set of non-effect states such that condition (M)
        holds.
       If $\Cause$ is not a GPR cause for $\Effect$, then
        there is an MR-scheduler for $\wminMDP{\cM}{\Cause}$ refuting the GPR
        condition for $\Cause$
        in $\wminMDP{\cM}{\Cause}$
        and a finite-memory scheduler for $\cM$ with two memory cells
        refuting the GPR
        condition for $\Cause$
        in $\cM$.
\end{theorem}

%%%%%%%%%%%%%%%%%%%%%%%%%%%%%%%%%%%%%%%%%%%%%%%%%%%%%%%%%%%%%%%%%%%%%%%%

The remainder of this section is concerned with the proofs of
Theorem \ref{thm:checking-GPR-in-poly-space} and
Theorem \ref{thm:MR-sufficient-GPR}.
We suppose that both the effect set $\Effect$ and the cause candidate $\Cause$ are fixed disjoint subsets of the state space of the MDP $\cM$
and that $\Cause$ satisfies (M).

\paragraph*{\bf Checking the GPR condition (Proof of Theorem \ref{thm:checking-GPR-in-poly-space}). }
The first step is a polynomial-time model transformation which permits to make the following assumptions when checking the GPR condition of $\Cause$ for $\Effect$.
\begin{description}
       \item [(A1)]
          $\Effect=\{\effuncov,\effcov\}$ consists of two terminal states.
          
	\item [(A2)] %
          For every state $c\in \Cause$, there is only a single enabled
          action, say $\Act(c)=\{\gamma\}$, and
          there exists $w_c\in [0,1]\cap \Rational$ such that
	  $P(c,\gamma,\effcov)=w_c$ and
	  $P(c,\gamma,\noeffc)=1{-}w_c$
          where $\noeffc$ is a terminal non-effect state
          and $\noeffc$ and $\effcov$ are only accessible via the
          $\gamma$-transition from the states $c\in \Cause$.
          
        \item [(A3)]
          $\cM$ has no end components and
          there is a further terminal state $\noeffbot$
          and an action $\tau$ such that 
          $\tau \in \Act(s)$ implies $P(s,\tau,\noeffbot)=1$.
\end{description}
Intuitively, $\effcov$ stands for covered effects (``$\Effect$ after $\Cause$'')
and can be seen as a true positive,
while $\effuncov$ represents the uncovered effects (``$\Effect$ without preceding $\Cause$'')
and corresponds to a false negative.
Let $\sched$ be a scheduler in $\cM$. Note that 
$\Pr^{\sched}_{\cM}((\neg \Cause)\until \Effect)=
 \Pr^{\sched}_{\cM}(\Diamond \effuncov)$
and
$\Pr^{\sched}_{\cM}( \Diamond (\Cause \wedge \Diamond \Effect))=
 \Pr^{\sched}_{\cM}(\Diamond \effcov)$.
 As the cause states can not reach each other we also have
 $\Pr^{\sched}_{\cM}((\neg \Cause)\Until c)= \Pr^{\sched}_{\cM}(\Diamond c)$
 for each $c\in \Cause$.
 The intuitive meaning of $\noeffc$ is a false positive (``no effect after $\Cause$''), while $\noeffbot$ stands for true negatives where neither the effect nor the cause is observed.
 Note that
 $\Pr^{\sched}_{\cM}( \Diamond (\Cause \wedge \neg \Diamond \Effect))= \Pr^{\sched}_{\cM}(\Diamond \noeffc)$
 and
 $\Pr^{\sched}_{\cM}( \neg \Diamond \Cause \wedge \neg \Diamond \Effect))= \Pr^{\sched}_{\cM}(\Diamond \noefftn)$.

\paragraph*{Justification of assumptions (A1)-(A3):}
We justify the assumptions as we can transform $\cM$ into a new MDP of the same asymptotic size satisfying the above assumptions.
Thanks to Lemma \ref{lemma:wmin-criterion-PR-causes},
we may suppose that $\cM=\wminMDP{\cM}{\Cause}$
(see Notation \ref{notation:MDP-mit-min-prob-ab-cause-candidates})
without changing the satisfaction of the GPR condition.
We then may rename the effect state $\eff$ and the non-effect state $\noeff$ reachable from $\Cause$ into $\effcov$ and $\noeffc$, respectively.
Furthermore, we collapse all other effect states into a single state $\effunc$ and all true negative states into $\noefftn$.
Similarly, by renaming and possibly duplicating terminal states we also suppose that $\noeffc$ has no other incoming transitions than the $\gamma$-transitions from the states in $\Cause$.
This ensures (A1) and (A2).
For (A3) consider the set $T$ of terminal states in the MDP obtained so far.
We remove all end components by switching to the MEC-quotient \cite{deAlfaro1997}, i.e., we collapse all states that belong to the same MEC $\cE$ into a single state $s_{\cE}$ while ignoring the actions inside $\cE$.
Additionally, we add a fresh $\tau$-transition from the states $s_{\cE}$ to $\noeffbot$
(i.e., $P(s_{\cE},\tau,\noeffbot)=1$).
The $\tau$-transitions from states $s_{\cE}$  to $\noeffbot$ mimic cases where schedulers of the original MDP eventually enter an end component and stay there forever with positive probability.
The soundness of the transition to the MEC-quotient is shown in Lemma \ref{lem:probabilities_MEC-quotient} and Corollary~\ref{cor:GPR_MEC}.

Note, however, that the transformation changes the memory-requirements of schedulers witnessing that $\Cause$ is not a GPR cause for $\Effect$. We will address the memory requirements in the original MDP later.

With assumptions (A1)-(A3), the GPR condition can be reformulated as follows:

\begin{lemma}
  \label{lem:GPR-poly-constraint}
  Under assumptions (A1)-(A3), $\Cause$ satisfies the GPR condition
  if and only if for
  each scheduler $\sched$ with $\Pr_{\cM}^{\sched}(\Diamond \Cause) >0$
  the following condition holds:
		\begin{equation}
			\label{GPR-1}
			\Pr_{\cM}^{\sched}(\Diamond \Cause) \cdot \Pr^{\sched}_{\cM}( \Diamond \effuncov)
			 \ < \ 
			 \bigl(1{-}\Pr_{\cM}^{\sched}(\Diamond \Cause)\bigr)
                         \cdot 
			\!\!\!\!\!\!\sum_{c\in \Cause} \!\!\!\!\!\!
			\Pr^{\sched}_{\cM}(\Diamond c) \cdot w_c
			\tag{GPR-1}
		\end{equation}
 \end{lemma}

With assumptions (A1)-(A3),  a terminal state of $\cM$ is reached almost surely under any scheduler after finitely many steps in expectation.
Given a scheduler $\sched$ for $\cM$, the expected frequencies
(i.e., expected number of occurrences in maximal paths) of 
state action-pairs $(s, \alpha)$, states $s\in S$ and state-sets $T \subseteq S$
under $\sched$ are defined by:
\begin{align*}
  \freq{\sched}{s,\alpha} & \ \ \eqdef \ \
  \mathrm{E}^{\sched}_{\cM}(\text{number of visits to $s$ in which $\alpha$ is taken})\\
  \freq{\sched}{s} &\ \ \eqdef
        \sum\nolimits_{\alpha \in \Act(s)} \freq{\sched}{s,\alpha},
	\qquad	\freq{\sched}{T} \eqdef \sum\nolimits_{s \in T} \freq{\sched}{s}.
\end{align*}
Let $T$ be one of the sets $\{\effcov\}$, $\{\effunc\}$, $\Cause$, or a singleton $\{c\}$ with $c\in \Cause$.
As $T$ is visited at most once during each run of $\cM$ (assumptions (A1) and (A2)), we have 
$\Pr^{\sched}_{\cN}(\Diamond T) = \freq{\sched}{T}$
for each scheduler $\sched$.
This allows us to express the violation of the GPR condition in terms of a quadratic constraint system over variables for the expected frequencies of state-action pairs in the following way:

Let $\SA$ denote the set of state-action pairs in $\cM$. We consider the following constraint system over the variables $x_{s,\alpha}$ for each $(s,\alpha)\in \SA$ where we use the short form notation
$x_s = \sum_{\alpha \in \Act(s)}x_{s,\alpha}$:
\begin{align*}
  x_{s,\alpha} & \ \geqslant 0 \qquad \text{for all $(s,\alpha) \in \SA$}
  \tag{1}
  \\
  x_{\init} & \ = \
  1+ \!\!\!\!\! \sum_{(t,\alpha) \in \SA} \!\!\!\!\!
      x_{t,\alpha}\cdot P(t,\alpha,\init)
  \tag{2}\\
  x_{s} & \ =  \sum_{(t,\alpha) \in \SA} \!\!\!\!\! x_{t,\alpha}\cdot P(t,\alpha,s)
  \qquad \text{for all $s\in S\setminus\{\init\}$}
  \tag{3}
\end{align*}
Using well-known results for MDPs without ECs (see, e.g., \cite[Theorem 9.16]{Kallenberg20}), given a vector $x\in \Real^{\SA}$, then $x$ is a solution to (1) and the balance equations (2) and (3)
if and only if there is a (possibly history-dependent) scheduler $\sched$ for $\cM$ with $x_{s,\alpha}=\freq{\sched}{s,\alpha}$ for all $(s,\alpha)\in \SA$
if and only if there is an MR-scheduler $\sched$ for $\cM$ with $x_{s,\alpha}=\freq{\sched}{s,\alpha}$ for all $(s,\alpha)\in \SA$.

The violation of \eqref{GPR-1} in Lemma \ref{lem:GPR-poly-constraint}
  and the condition $\Pr^{\sched}_{\cM}(\Diamond \Cause)>0$ can be reformulated in terms of the frequency-variables as follows where $x_{\Cause}$ is an abbreviation for $\sum_{c\in \Cause} x_c$:
\begin{align*}
  \label{non-GPR}
      &
      x_{\Cause} \cdot x_{\effuncov}
        \ \ \geqslant \ \
  \bigl(1-  x_{\Cause} \bigl) \cdot \!\!\!
  \sum_{c \in \Cause} \!\!\!\!\!\! x_{c} \cdot w_c
	\tag{4}
        \\[-1ex]
        &
        x_{\Cause}>0 \tag{5}      
\end{align*}

\begin{lemma}
	\label{prop:quadradic-system}
	Under assumptions (A1)-(A3), the set $\Cause$ is not a GPR cause for $\Effect$ in $\cM$ iff the constructed quadratic system of inequalities (1)-(5) has a solution.
\end{lemma}

This now puts us in the position to prove Theorem \ref{thm:checking-GPR-in-poly-space}.

\begin{proof}[Proof of Theorem \ref{thm:checking-GPR-in-poly-space}]
	The existence of a solution to the quadratic system of inequalities (Lemma \ref{prop:quadradic-system}) can straight-forwardly be formulated as a sentence in the language of the existential theory of the reals.
	The system of inequalities can be constructed from $\cM$, $\Cause$,  and $\Effect$ in polynomial time. Its solvability is decidable in polynomial space
 as the decision problem of the existential theory of the reals is in PSPACE \cite{ETH-Canny88}.
	\qed
\end{proof}

%%%%%%%%%%%%%%%%%%%%%%%%%%%%%%%%%%%%%%%%%%%%%%%%%%%%%%%%%%%%%%%%%%%%%%%%%%%%%  

\noindent{\bf Memory requirements of schedulers in the original MDP (Proof of Theorem \ref{thm:MR-sufficient-GPR}). }
As stated above, every solution to the
linear system of inequalities (1), (2), and (3) corresponds to the
expected frequencies of state-action pairs of an MR-scheduler in the
transformed model satisfying (A1)-(A3). Hence:

\begin{corollary}\label{cor:MR-scheduler}
Under assumptions (A1)-(A3),
		$\Cause$ is no GPR cause for
	$\Effect$ iff there exists an MR-scheduler $\tsched$
	with $\Pr^{\tsched}_{\cM}(\Diamond \Cause) >0$
        violating the GPR condition.
\end{corollary}

The model transformation we used for assumptions (A1)-(A3), however, does affect the memory requirements of scheduler.
We may further restrict the MR-schedulers necessary to witness non-causality under assumptions (A1)-(A3).
For the following lemma, recall that $\tau$ is the action of the MEC quotient used for the extra transition from states representing MECs to a new trap state (see also assumption (A3)).

\begin{restatable}{lemma}{MRschedulersufficient}
  \label{lem:MR-sufficient-global-cause}
	Assume (A1)-(A3).
    Given an MR-scheduler $\usched$ with $\Pr^{\usched}_{\cM}(\Diamond \Cause) >0$ that violates \eqref{GPR},
    an MR-scheduler $\tsched$ with $\tsched(s)(\tau)\in \{0,1\}$ for each state $s$ with $\tau \in \Act(s)$
    that satifies $\Pr^{\tsched}_{\cM}(\Diamond \Cause) >0$ and violates \eqref{GPR}
    is computable in polynomial time.
\end{restatable}

For the proof, see Appendix \ref{app:GPR_check}.
The condition that $\tau$ only has to be scheduled  with probability $0$ or $1$ in each state is the key to transfer the sufficiency of MR-schedulers to the MDP $\wminMDP{\cM}{\Cause}$. %
This fact is of general interest as well and stated in the following theorem where $\tau$ again is the action added to move from a state $s_{\cE}$ to the new trap state  in the MEC-quotient.

\begin{restatable}{theorem}{MRschedulerlift}
 \label{thm:MR-schedulers-MEC-quotient}\label{thm:MRscheduler_lift}
 Let $\cM$ be an MDP with pairwise disjoint action sets for all states. 
 Then, for each MR-scheduler $\sched$ for the MEC-quotient of $\cM$ with
 $\sched(s_{\cE})(\tau)\in \{0,1\}$ for each MEC $\cE$ of $\cM$ there is an MR-scheduler $\tsched$ for $\cM$ such that every action  $\alpha$ of $\cM$ that does not belong to an MEC of $\cM$, has the same expected frequency under $\sched$ and $\tsched$.
\end{restatable}

\begin{proofsketch}   
      The crux are cases where $\sched(s_{\cE})(\tau)=0$, which requires
      to traverse the MEC $\cE$ of $\cM$ in a memoryless way such that all actions leaving $\cE$ have the same expected frequency under $\tsched$ and $\sched$. First, we construct a finite-memory scheduler $\tsched^\prime$ that always leaves each such end component according to the distribution  given by $\sched(s_{\cE})$. By \cite[Theorem 9.16]{Kallenberg20},  we then conclude that there is an MR-scheduler $\tsched$ under which the  expected frequencies of all state-action pairs are the same as under $\tsched^\prime$.
	\qed
\end{proofsketch}

\begin{proof}[Proof of Theorem \ref{thm:MR-sufficient-GPR}]
The model transformation establishing assumptions (A1)-(A3) results in the MEC-quotient of $\wminMDP{\cM}{\Cause}$ up to the renaming and collapsing of terminal states. By Corollary \ref{cor:MR-scheduler} and Theorem  \ref{thm:MRscheduler_lift}, we conclude that
$\Cause$ is not a GPR cause for $\Effect$ in $\cM$ if and only if there is a MR-scheduler $\sched$ for $\wminMDP{\cM}{\Cause}$ with $\Pr^{\sched}_{\wminMDP{\cM}{\Cause}}(\Diamond \Cause)>0$ that violates \eqref{GPR}.
As in Remark \ref{MR-sufficient-SRP}, 
$\sched$ can be extended to
a finite-memory randomized scheduler $\tsched$ for $\cM$ with two memory cells.
\qed
\end{proof}

%%%%%%%%%%%%%%%%%%%%%%%%%%%%%%%%%%%%%%%%%%%%%%%%%%%%%%%%%%%%%%%%%%%%%%%%

\begin{remark}[On lower bounds on GPR checking]
Solving systems of quadratic inequalities with linear side constraints is NP-hard in general (see, e.g., \cite{GareyJ79}).
For convex problems, in which the associated symmetric matrix occurring in the quadratic inequality has only non-negative eigenvalues, the problem is, however, solvable in polynomial time \cite{kozlov1980polynomial}.
Unfortunately, the quadratic constraint system given by (1)-(5) is not of this form.
We observe that even if $\Cause$ is a singleton $\{c\}$ and the variable $x_{\effuncov}$ is forced to take a constant value $y$ by (1)-(3), i.e., by the structure of the MDP, the inequality (4) takes the form:
\begin{enumerate}
\item []
  \hspace*{3cm}
  $x_c\cdot w_c - x_c^2 \cdot (w_c+y) \leq 0$
  \hfill (*)
\end{enumerate}
Here, the $1\times 1$-matrix $({-}w_c{-}y)$ has a negative eigenvalue.
Although it is not ruled out that
  (1)-(5) belongs to another class of
  efficiently solvable constraint systems,
  the NP-hardness result in \cite{pardalos1991quadratic}
  for the solvability of
  quadratic inequalities of the form (*) 
  with linear side constraints 
  might be an indication for the computational difficulty.
 \Ende
\end{remark}

\section{Quality and optimality of causes}

\label{sec:criteria}

The goal of this section is to identify notions that measure how ``good'' causes are and to present algorithms to determine good causes according to the proposed quality measures.
We have seen so far that small (singleton) causes are easy to determine (see Section \ref{sec:check-SPR-condition}).
Moreover, it is easy to see that the proposed existence-checking algorithm can be formulated in such a way that the algorithm returns a singleton (strict or global) probability-raising cause $\{c_0\}$ with maximal \emph{precision}, i.e., a state $c_0$ where 
$\inf_{\sched} \Pr^{\sched}_{\cM}(\Diamond \Effect |\Diamond c_0)
 = \Pr^{\min}_{\cM,c_0}(\Diamond \Effect)$ is maximal.
On the other hand, 
singleton or small cause sets might have poor coverage in the sense that the probability for paths that reach an effect state without visiting a cause state before (``uncovered effects'') can be large. 
This motivates the consideration of quality notions for causes that incorporate how well 
 effect scenarios are covered. 
We take inspiration of quality measures that are considered in statistical analysis (see e.g. \cite{Powers-fscore}). 
This includes the \emph{recall} as a measure for the relative coverage (proportion of covered effects among all effect scenarios), the \emph{coverage ratio} (quotient of covered and uncovered effects) as well as the \emph{f-score}. 
The f-score is a standard measure for classifiers  defined by the harmonic mean of precision and recall.
It can be seen as a compromise to achieve both good precision and good recall.

Throughout this section, we assume as before an MDP $\cM = (S,\Act,P,\init)$ and a set
$\Effect \subseteq S$ are given where all effect states are terminal. Furthermore, we suppose that all
states $s\in S$ are reachable from $\init$. 
Detailed proofs can be found in Appendix \ref{app:section_5}.

%%%%%%%%%%%%%%%%%%%%%%%%%%%%%%%%%%%%%%%%%%%%%%%%%%%%%%%%%%%%%%%%%%%%%%%%

\subsection{Quality measures for causes}

\label{sec:acc-measures}

In statistical analysis, the precision of a classifier with binary outcomes (``positive'' or ``negative'') is defined as the ratio of all true positives among all positively classified elements, while its recall is defined as the ratio of all true positives among all actual positive elements.
Translated to our setting, we consider classifiers induced by a given cause set $\Cause$ that return ``positive'' for sample paths in case that a cause state is visited 
and ``negative'' otherwise. 
The intuitive meaning of true positives and false negatives is as explained after Definition \ref{def-PR-causes}.
The meaning of true negatives and false positives is analogous.
We use $\mathsf{tp}^\sched$ for the probability for true positives under $\sched$. 
The notations $\mathsf{fp}^\sched$, $\mathsf{fn}^\sched$, $\mathsf{tn}^\sched$ have analogous meanings.

With this interpretation of causes as binary classifiers in mind, the recall and precision
and coverage ratio 
of a cause set $\Cause$ \emph{under a scheduler} $\sched$ is defined as follows
(assuming $\Pr^{\sched}_{\cM}(\Diamond \Effect)>0$
resp. $\Pr^{\sched}_{\cM}(\Diamond \Cause)>0$ resp. $\Pr^{\sched}_{\cM}\bigl( (\neg \Cause) \until \Effect \bigr)>0$):
$$
   \begin{array}{rclcl}
           \precision^{\sched}(\Cause) & \ = \ &
           \Pr^{\sched}_{\cM}(\ \Diamond \Effect \ | \ \Diamond \Cause \ )
           & = &
           \frac{\mathsf{tp}^{\sched}}
                {\mathsf{tp}^{\sched} + \mathsf{fp}^{\sched}}
           \\[1ex] 
           \recall^{\sched}(\Cause) & = &
           \Pr^{\sched}_{\cM}(\ \Diamond \Cause  \ | \ \Diamond \Effect \ )
           & = &
           \frac{\mathsf{tp}^{\sched}}{\mathsf{tp}^{\sched}+\mathsf{fn}^{\sched}}
   \end{array}
$$
$$
	\begin{array}{rclcl}
		 \ratiocov^{\sched}(\Cause) & = &
		\frac{\displaystyle
			\Pr^{\sched}_{\cM}
			\bigl(\Diamond (\Cause \wedge \Diamond \Effect) \bigr)}
		{\displaystyle
			\Pr^{\sched}_{\cM}\bigl((\neg \Cause) \until \Effect \bigr)}
		& = &
		\frac{\mathsf{tp}^{\sched}}
		{\mathsf{fn}^{\sched}}.
	\end{array}
$$
For the coverage ratio, if $\Pr^{\sched}_{\cM}\bigl( (\neg \Cause) \until \Effect \bigr)=0$ and $\Pr^\sched_{\cM}(\Diamond \Cause)>0$ we define $\covrat^\sched(\Cause)=+\infty$.
Finally, the f-score of $\Cause$ \emph{under a scheduler} $\sched$
is defined as the harmonic mean of the precision and recall
(assuming $\Pr^{\sched}_{\cM}(\Diamond \Cause)>0$, which implies
$\Pr^{\sched}_{\cM}(\Diamond \Effect)>0$ as $\Cause$ is a PR cause):
\begin{equation*}
 \label{f-score}
            \fscore^{\sched}(\Cause) \ \ \eqdef \ \
	     2 \cdot
             \frac{\precision^{\sched}(\Cause)\cdot \recall^{\sched}(\Cause)}
                  {\precision^{\sched}(\Cause) + \recall^{\sched}(\Cause)}
			\end{equation*}
If, however,
$\Pr_{\cM}^\sched(\Diamond \Eff)>0$ and $\Pr_{\cM}^{\sched}(\Diamond \Cause)=0$
we define $\fscore^\sched(\Cause)=0$.

\paragraph*{\bf Quality measures for cause sets.}   
 Let $\Cause$ be a PR cause.
 The recall of $\Cause$ measures the relative coverage 
 in terms of the worst-case conditional
 probability for covered effects (true positives)
 among all scenarios where the effect occurs.
\begin{center}
  $\relcov(\Cause)  \ = \ 
   \inf_{\sched} \ \relcov^{\sched}(\Cause)
   \ = \
   \Pr^{\min}_{\cM}(\ \Diamond \Cause \ | \ \Diamond \Effect \ )$
\end{center}
when ranging over all schedulers $\sched$
with $\Pr^{\sched}_{\cM}(\Diamond \Effect)>0$.   
 Likewise, the coverage ratio and f-score of $\Cause$
 are defined by the worst-case coverage ratio resp. f-score
 (when ranging over schedulers for which $\covratio^{\sched}(\Cause)$ 
  resp. $\fscore^{\sched}(\Cause)$ is defined):
 \begin{center}
   $\covratio(\Cause)  \ = \ 
   \inf_{\sched} \ \covratio^{\sched}(\Cause),
   \quad
   \fscore(\Cause)  \ = \ 
   \inf_{\sched} \ \fscore^{\sched}(\Cause)$
\end{center}

\subsection{Computation schemes for the quality measures for fixed cause set}

\label{sec:comp-acc-measures-fixed-cause}

For this section, we assume  a fixed PR cause $\Cause$ is given and address the problem to compute its quality values.
Since all quality measures are preserved by the switch from $\cM$ to $\wminMDP{\cM}{\Cause}$ 
as well as the transformations of $\wminMDP{\cM}{\Cause}$ to an MDP that satisfies conditions (A1)-(A3) of Section \ref{sec:check-GPR},
we may assume that $\cM$ satisfies (A1)-(A3).

%%%%%%%%%%%%%%%%%%%%%%%%%%%%%%%%%%%%%%%%%%%%%%%%%%%%%%%%%%%%%%%%%%%%%%%%%%%%%%%

While efficient computation methods for $\recall(\Cause)$
are known from literature (see \cite{TACAS14-condprob,Maercker-PhD20}
for poly-time algorithms to compute
conditional reachability probabilities), we are not aware of
known concepts that are applicable 
for computing the coverage ratio or the f-score.
Indeed, both are efficiently computable:

\begin{restatable}{theorem}{computecovratiofscore}
  \label{ratiocov-fscore-in-PTIME}
  The values $\ratiocov(\Cause)$ and $\fscore(\Cause)$ and corresponding worst-case schedulers are computable in polynomial time.
\end{restatable}

%%%%%%%%%%%%%%%%%%%%%%%%%%%%%%%%%%%%%%%%%%%%%%%%%%%%%%%%%%%%%%%%%%%%%%%%%%%%%%%

The remainder of this subsection is devoted to the proof of Theorem \ref{ratiocov-fscore-in-PTIME}.
By definition, the value $\covratio(\Cause)$ is the infimum over a quotient of reachability probabilities for disjoint sets of terminal states.
While this is not the case for the \mbox{f-score}, we can express $\fscore(\Cause)$ in terms of the supremum of such a quotient.
More precisely, under assumptions (A1)-(A3) and assuming $\fscore(\Cause)>0$,
we have:
\begin{center}
	$
  \fscore(\Cause) = \frac{2}{X+2}
  \quad \text{where} \quad
  X \ = \ %
    \sup_\sched \, \frac{\Pr^{\sched}_{\cM}(\Diamond \noefffp)
                    + \Pr^{\sched}_{\cM}(\Diamond \effuncov)}
                     {\Pr^{\sched}_{\cM}(\Diamond \effcov)}
$
\end{center}
where $\sched$ ranges over all schedulers with
      $\Pr_{\cM}^\sched (\Diamond \effcov)>0$.
Moreover, $\fscore(\Cause)=0$ iff $\recall(\Cause)=0$ iff there exists a scheduler $\sched$ satisfying $\Pr_{\cM}^\sched(\Diamond \Eff)>0$ and $\Pr_{\cM}^{\sched}(\Diamond \Cause)=0$. %%

So, the remaining task
to prove Theorem \ref{ratiocov-fscore-in-PTIME} is a generally applicable technique for computing
extremal ratios of reachability probabilities in MDPs without ECs.

\paragraph*{\bf Max/min ratios of reachability probabilities for disjoint sets of terminal states.}
\label{sec:comp-quotient}
Suppose we are given an MDP $\cM= (S,\Act,P,\init)$ without ECs and disjoint subsets $U,V\subseteq S$  of terminal states. 
Given a scheduler $\sched$ with $\Pr_\cM^\sched(\Diamond V)>0$ we define:
\begin{center}
  $\ratio{\sched}{\cM}{U,V} \ = \
  {\Pr_{\cM}^{\sched}(\Diamond U)}\, / \, {\Pr_{\cM}^{\sched} (\Diamond V)}$
\end{center}
The goal is to provide an algorithm for computing the extremal values:
$\ratio{\min}{\cM}{U,V} = \inf_{\sched} \ratio{\sched}{\cM}{U,V}$
and
$\ratio{\max}{\cM}{U,V} = \sup_{\sched} \ratio{\sched}{\cM}{U,V}$
where $\sched$ ranges over all schedulers with $\Pr_\cM^\sched(\Diamond V)>0$. 
To compute these, we rely on a polynomial reduction to the classical \emph{stochastic shortest 
	 path problem} \cite{BT91}. 
For this, consider the MDP $\cN$ arising from $\cM$
by adding reset transitions
from all terminal states $t \in S \backslash V$ to $\init$.
Thus, exactly the $V$-states are terminal in $\cN$.
$\cN$ might contain ECs, which, however, do not intersect with $V$.
We equip $\cN$ with the weight function that assigns $1$ to all states in $U$ and $0$ to all other states. 
For a scheduler $\tsched$ with $\Pr^{\tsched}_{\cN}(\Diamond V)=1$, 
let $\mathrm{E}^{\tsched}_{\cN}(\boxplus V)$ be the expected accumulated weight until reaching $V$ under $\tsched$.
Let $\mathrm{E}^{\min}_{\cN}(\boxplus V) =  \inf_{\tsched} \mathrm{E}^{\tsched}_{\cN}(\boxplus V)$ and  $\mathrm{E}^{\max}_{\cN}(\boxplus V) =  \sup_{\tsched} \mathrm{E}^{\tsched}_{\cN}(\boxplus V)$, 
where $\tsched$ ranges over all schedulers 
 with
\mbox{$\Pr^{\tsched}_{\cN}(\Diamond V)=1$.}
We can rely on known results \cite{BT91,Alfaro-CONCUR99,LICS18-SSP} to obtain
that both
$\mathrm{E}^{\min}_{\cN}(\boxplus V)$ and $\mathrm{E}^{\max}_{\cN}(\boxplus V)$
are computable in polynomial time.
As $\cN$ has only non-negative weights,
$\mathrm{E}^{\min}_{\cN}(\boxplus V)$ is finite
and a corresponding MD-scheduler with minimal expectation exists.
If $\cN$ has an EC containing at least one $U$-state,
which is the case iff $\cM$ has a scheduler $\sched$
with $\Pr^{\sched}_{\cM}(\Diamond U)>0$ and $\Pr^{\sched}_{\cM}(\Diamond V)=0$,
then
$\mathrm{E}^{\max}_{\cN}(\boxplus V) = +\infty$.
Otherwise, $\mathrm{E}^{\max}_{\cN}(\boxplus V)$
is finite and the maximum is achieved by an MD-scheduler as well.

\begin{restatable}{theorem}{compQ}\label{thm:comp-Q}
  Let $\cM$ be an MDP without ECs and
  $U,V$ disjoint sets of terminal states in $\cM$, and let $\cN$ be as before.
  Then,
  $\ratio{\min}{\cM}{U,V}=\mathrm{E}^{\min}_{\cN}(\boxplus V)$ and
  $\ratio{\max}{\cM}{U,V}=\mathrm{E}^{\max}_{\cN}(\boxplus V)$.
  Thus, both values are computable in polynomial time,
  and there is an MD-scheduler minimizing $\ratio{\sched}{\cM}{U,V}$,
  and an MD-scheduler maximizing $\ratio{\sched}{\cM}{U,V}$ if 
  $\ratio{\max}{\cM}{U,V}$ is finite.
\end{restatable}

\begin{proofsketch}[Proof of Theorem \ref{ratiocov-fscore-in-PTIME}]
Using assumptions (A1)-(A3), we obtain that 
$\covratio(\Cause)=\ratio{\min}{\cM}{U,V}$ where
$U=\{\effcov\}$, $V=\{\effunc\}$.
Similarly, with $U=\{\noefffp,\effunc\}$, $V=\{\effcov\}$, we get
$\fscore(\Cause)=0$ if $\ratio{\max}{\cM}{U,V}=+\infty$
and
$\fscore(\Cause)=2/(\ratio{\max}{\cM}{U,V}+2)$ otherwise.
Thus, the claim follows from Theorem \ref{thm:comp-Q}.
\qed
\end{proofsketch}

%%%%%%%%%%%%%%%%%%%%%%%%%%%%%%%%%%%%%%%%%%%%%%%%%%%%%%%%%%%%%%%%%%%%%%%%%%%%%%%%%%%%%%%%%%%%%

\subsection{Quality-optimal probability-raising causes}	

\label{sec:opt-PR-causes}

An SPR cause $\Cause$ is called \emph{recall-optimal}
if
$\relcov(\Cause) = \max_C \relcov(C)$ where $C$ ranges over all
SPR causes.
Likewise, \emph{ratio-optimality} resp. \emph{f-score-optimality} of
$\Cause$ means maximality of
$\ratiocov(\Cause)$ resp. 
$\fscore(\Cause)$ among all SPR causes.
Recall-, ratio- and f-score-optimality for
GPR causes are defined accordingly.

%

%%%%%%%%%%%%%%%%%%%%%%%%%%%%%%%%%%%%%%%%%%%%%%%%%%%%%%%%%%%%%%%%%%%%%%%%%%%%%%%%%%%%%%%%%%%%%

\begin{restatable}{lemma}{recalloptimalityequalsratiooptimality}
  \label{lemma:recall-opt=ratio-opt}
  Let $\Cause$ be an SPR or a GPR cause.
  Then, $\Cause$ is recall-optimal if and only if $\Cause$ is ratio-optimal.
\end{restatable}

%%%%%%%%%%%%%%%%%%%%%%%%%%%%%%%%%%%%%%%%%%%%%%%%%%%%%%%%%%%%%%%%%%%%%%%%%%

\paragraph*{\bf Recall-  and ratio-optimal SPR causes.}
\label{sec:opt-SPR-causes}
The techniques of Section \ref{sec:check-SPR-condition}
yield an algorithm for generating a canonical SPR
cause with optimal recall and ratio.
To see this, let $\cC$ denote the set of states that constitute a singleton SPR cause. The canonical cause $\CanCause$ is defined as the set of states $c\in \cC$ such that there is a scheduler $\sched$ with $\Pr_{\cM}^\sched((\neg \cC) \Until c)>0$. Obviously, $\cC$ and $\CanCau$ are computable in polynomial time.

\begin{restatable}{theorem}{ThmCanCau}%
  \label{thm:optimality-of-canonical-SPR-cause}
  If $\cC\not= \varnothing$ then
  $\CanCause$ is a ratio- and recall-optimal SPR cause.
\end{restatable}

%%%%%%%%%%%%%%%%%%%%%%%%%%%%%%%%%%%%%%%%%%%%%%%%%%%%%%%%%%%%%%%%%%%%%%%%%

\setlength{\intextsep}{1pt}
\begin{wrapfigure}[6]{R}{0.3\textwidth}
	\resizebox{0.3\textwidth}{!}{
		%CanCau is not f-score optimal

\begin{tikzpicture}[->,>=stealth',shorten >=1pt,auto ,node distance=0.5cm, thick]
	\node[scale=1, state] (s0) {$\init$};
	\node[scale=1, state] (eff) [right = 1 of s0] {$\eff$};
	\node[scale=1, state] (noeff) [below =0.9 of s0] {$\noeff$};
	\node[scale=1, state] (s1) [below =1 of eff] {$s_1$};
	\node[scale=1, state] (s2) [right = 1 of s1] {$s_2$};
	
	\draw[<-] (s0) --++(-0.55,0.55);
	\draw (s0) -- (eff) node[above, pos=0.5,scale=1] {$1/4$};
	\draw (s0) -- (noeff) node[left , pos=0.5,scale=1] {$1/4$};
	\draw (s0) -- (s1) node[right, pos=0.5,scale=1] {$1/2$};
	\draw (s1) -- (noeff) node[pos=0.5,scale=1] {$1/4$};
	\draw (s1) -- (s2) node[below, pos=0.5,scale=1] {$3/4$};
	\draw (s2) -- (eff) node[right, pos=0.5,scale=1] {$1$};
\end{tikzpicture}}
\end{wrapfigure}
	\label{ex:CanCau-not-optimal}
		This is not true for the f-score.
		To see this, 
		Consider the Markov chain 
		on the right hand side.
		We have
		$\CanCau = \{s_1\}$, 
		which has 
		$\precision(\CanCau) = \frac{3}{4}$ and $\recall(\CanCau) = \frac{3}{8}/(\frac{1}{4}+\frac{3}{8}) = \frac{3}{5}.$
		But the SPR cause $\{s_2\}$ has better f-score as its precision is $1$ and it has the same recall as $\CanCau$.

%%%%%%%%%%%%%%%%%%%%%%%%%%%%%%%%%%%%%%%%%%%%%%%%%%%%%%%%%%%%%%%%%%%%%%%%%

\paragraph*{\bf F-score-optimal SPR cause.}
From Section \ref{sec:comp-acc-measures-fixed-cause}, we see that f-score-optimal SPR causes in MDPs can be computed 
in polynomial space by computing the f-score for all potential SPR causes one by one in polynomial time (Theorem \ref{ratiocov-fscore-in-PTIME}).
As the space can be reused after each computation, this results in polynomial space.
For Markov chains, we can do better
 and compute an f-score-optimal SPR cause in polynomial time.
 via a polynomial reduction to the stochastic shortest path problem:

\begin{restatable}{theorem}{fscoreMC}
  \label{thm:fscore-opt-MC}
  In Markov chains that have SPR causes, an f-score-optimal SPR cause can be computed in polynomial time.
\end{restatable}
  
\begin{proof}
  We regard the given Markov chain 
  $\cM$ as an MDP with a singleton action set $\Act=\{\alpha\}$.
  As $\cM$ has SPR causes, the set $\cC$
  of states that constitute a singleton SPR cause
  is nonempty.
  We may assume that $\cM$ has no non-trivial (i.e., cyclic)
  bottom strongly connected components as we may collapse them.
  Let $w_c$ $=$ $\Pr_{\cM,c}(\Diamond \Effect)$.  
  We switch from $\cM$ to a new MDP $\cK$ with state space
  $S_{\cK}=S \cup \{\effcov,\noefffp\}$ with 
  fresh states $\noefffp$ and $\effcov$
  and the action set
  $\Act_{\cK}=\{\alpha,\gamma \}$.
  The MDP $\cK$ arises from $\cM$
  by adding (i) for each state $c\in \cC$ a fresh state-action pair
  $(c,\gamma)$ with $P_{\cK}(c,\gamma,\effcov)=w_c$
  and $P_{\cK}(c,\gamma,\noefffp)=1{-}w_c$
  and (ii) 
  reset transitions to $\init$ with action label $\alpha$
  from the new state $\noefffp$ and
  all terminal states of $\cM$,
  i.e.,
  $P_{\cK}(\noefffp,\alpha,\init)=1$ and $P_{\cK}(s,\alpha,\init)=1$ 
  for $s \in \Effect$ or if $s$ is a terminal
  non-effect state of $\cM$.
  So, exactly $\effcov$ is terminal in $\cK$, 
  and $\Act_{\cK}(c)=\{\alpha,\gamma\}$ for $c\in \cC$, while
  $\Act_{\cK}(s)=\{\alpha\}$ 
  for all other states $s$.
  Intuitively, taking action $\gamma$ in state $c \in \cC$ selects
  $c$ to be a cause state.
   The states in $\Effect$ represent uncovered effects in $\cK$,
   while $\effcov$ stands for covered effects.
  
  We assign weight $1$ to all states 
  in $U= \Effect \cup \{\noefffp\}$ 
  and weight $0$ to all other states of~$\cK$. 
  Let $V=\{\effcov\}$.
  Then, $f= \mathrm{E}^{\min}_{\cK}(\boxplus V)$
  and  an MD-scheduler $\sched$ for $\cK$ such that
  $\mathrm{E}^{\sched}_{\cK}(\boxplus V)=f$ are computable in
  polynomial time.
  Let $\cC_{\gamma}$ 
  denote the set of states $c\in \cC$ where
  $\sched(c)=\gamma$ 
  and let $\Cause$ be the set of states $c\in \cC_{\gamma}$
  where $\cM$ has a path satisfying
  $(\neg \cC_{\gamma}) \Until c$.
  Then, $\Cause$ is an SPR cause of $\cM$.
  With arguments as in Section \ref{sec:comp-acc-measures-fixed-cause}
  we obtain
  $\fscore(\Cause)=2/(f{+}2)$. 
  It remains to show that $\Cause$ is f-score-optimal.
  Let $C$ be an arbitrary SPR cause.
  Then, $C \subseteq \cC$. Let $\tsched$ be the MD-scheduler
  for $\cK$ that schedules $\gamma$ in $C$ and
  $\alpha$ for all other states of $\cK$.
  Then, $\fscore(C)=2/(f^{\tsched}{+}2)$
  where $f^{\tsched}=\mathrm{E}^{\tsched}_{\cK}(\boxplus V)$.
  Hence, $f \leqslant f^{\tsched}$, which yields 
  $\fscore(\Cause) \geqslant \fscore(C)$.
	\qed
\end{proof}

The na\"ive adaption of 
the construction presented in the proof of Theorem \ref{thm:fscore-opt-MC}
for MDPs would yield a stochastic game structure where the objective of
one player is to minimize the expected accumulated weight until reaching
a target state.
Although algorithms for \emph{stochastic shortest path (SSP) games} are known
\cite{patek1999stochastic}, 
they rely on assumptions on the game structure which would not
be satisfied here.
However, for the threshold problem \emph{SPR-f-score} where inputs are an MDP $\cM$, $\Effect$ and $\vartheta \in \Rat_{\geq 0}$ and the task is to decide the existence of an SPR cause whose f-score exceeds $\vartheta$, we can establish a polynomial reduction to SSP games, which yields an $\NP \cap \coNP$ upper bound:

\begin{restatable}{theorem}{fscoreNPcoNP}
	\label{fscore-threshold-poblem-via-stochMPgames}
	The decision problem SPR-f-score is in $\NP\cap \coNP$.
\end{restatable}

\begin{proofsketch}
	Given an MDP $\cM$, $\Effect$, and $\vartheta$, 
	we construct an SSP game %%%stochastic shortest path game
       \cite{patek1999stochastic} 
after a series of model transformations ensuring 
(i)	that 	terminal states are reached almost surely and
(ii)	 that $\Effect$ is reached with positive probability under all schedulers. 
Condition (i) is established by a standard  MEC-quotient construction. To establish condition (ii),  we provide a construction that forces schedulers to leave an initial sub-MDP in which the minimal probability to reach $\Effect$ is $0$. This construction -- unlike  the MEC-quotient -- affects the possible combinations of probability values with which terminal states and potential cause states can be reached, but the existence of an SPR cause satisfying the f-score-threshold condition is not affected.

The underlying idea of the construction of the game shares similarities with the MDP constructed in the proof of Theorem \ref{thm:fscore-opt-MC}:  Player $0$ takes the role to select potential cause states while   player $1$ takes the role of a scheduler in the transformed MDP. 
	Using the  observation  that for each cause $C$, $\fscore(C)>\vartheta$ iff
\begin{center}
  \hspace*{.3cm}
  $2(1{-}\vartheta)\Pr^{\sched}_{\cM}(\Diamond C \land \Diamond \Effect) - \vartheta \Pr^{\sched}_{\cM}(\neg \Diamond C \land \Diamond \Effect)
	-\vartheta \Pr^{\sched}_{\cM}(\Diamond C \land \neg \Diamond \Effect)  >  0$
  \hfill ($\times$)
\end{center}
 for all schedulers $\sched$ for $\cM$ with $\Pr^{\sched}_{\cM}(\Diamond \Effect)>0$, weights are assigned to $\Effect$-states and other terminal states depending on whether player $0$ has chosen to include a state to the cause beforehand. 
  In the resulting SSP game, both players have 
optimal MD-strategies \cite{patek1999stochastic}. Given such strategies $\zeta$ for player $0$ and $\sched$ for player $1$, the resulting expected accumulated weight agrees with the left-hand side of ($\times$) 
  when considering $\sched$ as a scheduler for the transformed MDP and the cause $C$ induced by the states that $\zeta$  chooses to belong to the cause.
  So, player $0$ wins the constructed game iff an SPR cause with f-score above the threshold $\vartheta$ exists.
  The existence of optimal 
 MD-strategies for both players  allows us to decide this threshold problem in NP and coNP. 
 \qed 
 \end{proofsketch}

%%%%%%%%%%%%%%%%%%%%%%%%%%%%%%%%%%%%%%%%%%%%%%%%%%%%%%%%%%%%%%%%%%%%%%%%%%%%%%%

\paragraph*{\bf Optimality and threshold constraints for GPR causes.}
Computing optimal GPR causes for either quality measure can be done in polynomial space by considering all cause candidates, checking the GPR condition in polynomial space (Theorem \ref{thm:checking-GPR-in-poly-space}) and computing the corresponding quality measure in polynomial time (Section \ref{sec:comp-acc-measures-fixed-cause}).
However, we show that no polynomial-time algorithms can be expected as the corresponding threshold problems 
are $\NP$-hard.
Let GPR-covratio (resp. GPR-recall, GPR-f-score) denote the decision
problems: Given $\cM,\Effect$ and $\vartheta \in \Rat$, decide
whether there exists a GPR cause with
coverage ratio (resp. recall, f-score) at least $\vartheta$.

\iffalse
\noindent GPR-covratio:
\begin{itemize}
\item
  input: $\cM,\Effect$ and $\vartheta \in \Rat$
\item
  question:
  does there exists a GPR cause $\Cause$
  with coverage ratio at least $\vartheta$ ?
\end{itemize}  
GPR-recall:
\begin{itemize}
\item
  input: $\cM,\Effect$ and $\vartheta \in \Rat$
\item
  question:
  does there exists a GPR cause $\Cause$ with recall at least $\vartheta$ ?
\end{itemize}  
GPR-f-score:
\begin{itemize}
\item
  input: $\cM,\Effect$ and $\vartheta \in \Rat$
\item
  question:
  does there exists a GPR cause $\Cause$
  with f-score at least $\vartheta$ ?
\end{itemize}
\fi

\begin{restatable}{theorem}{measureNPhardness}
  \label{GPR-recall-NPhard-and-in-PSPACE}
  The problems GPR-covratio, GPR-recall and GPR-f-score
  are   NP-hard and belong to PSPACE.
  For Markov chains, all three problems are NP-complete.
  NP-hardness even holds
  for tree-like Markov chains.
\end{restatable}  

\begin{proofsketch}
  NP-hardness is established via a polynomial reduction from
  the knapsack problem.
  Membership to NP for Markov chains resp. to $\PSPACE=\NPSPACE$
  for MDPs is obvious as we can guess nondeterministically
  a cause candidate and then check (i) the GPR condition in
  polynomial time (Markov chains) resp. polynomial space (MDPs)
  and (ii) the threshold condition in polynomial time
  (see Section \ref{sec:comp-acc-measures-fixed-cause}).
\qed  
\end{proofsketch}

%%%%%%%%%%%%%%%%%%%%%%%%%%%%%%%%%%%%%%%%%%%%%%%%%%%%%%%%%%%%%%%%%%%%%%%%%%%%%%%%%%%%%%%%%%%%%

\section{Conclusion}

\label{sec:conclusion}

The goal of the paper was to formalize the PR principle in MDPs and related quality notions for PR causes and to study fundamental algorithmic problems for them. We considered the strict (local) and the global view. Our results indicate that GPR causes are more general and leave more flexibility to achieve better accuracy, while algorithmic reasoning about SPR causes is simpler.

\vspace*{0.2ex}
\noindent
{\it Existential definition of SPR/GPR causes.}
The proposed definition of PR causes relies on a universal quantification
over all relevant schedulers.
However, another approach could be via existential quantification, i.e. there is a scheduler $\sched$ such that \eqref{GPR} or resp. \eqref{SPR} hold.
The resulting notion of causality yields fairly the same results (up to $\Pr^{\max}_{\cM,c}(\Diamond \Effect)$ instead of $\Pr^{\min}_{\cM,c}(\Diamond \Effect)$ etc).
A canonical existential SPR cause can be defined in analogy to the universal case and shown to be recall- and ratio-optimal (cf.~Theorem \ref{thm:optimality-of-canonical-SPR-cause}).
The problem to find an existential f-score-optimal SPR cause is even simpler and solvable in polynomial time as the construction presented in the proof of Theorem \ref{thm:fscore-opt-MC}
can be adapted for MDPs (thanks to the simpler nature of $\max_C \sup_{\sched} \fscore^{\sched}(C)$ compared to
$\max_C \inf_{\sched} \fscore^{\sched}(C)$).
However, NP-hardness 
for the existence of GPR causes with threshold constraints for the quality carries over to the existential definition (as NP-hardness holds for Markov chains, Theorem \ref{GPR-recall-NPhard-and-in-PSPACE}).

\vspace*{0.2ex}
\noindent
{\it Non-strict inequality in the PR conditions.}
Our notions of PR causes are in line with the classical approach
of probability-raising causality 
in literature
with strict inequality in the PR condition, with the consequence that
causes might not exist (see Example \ref{ex:no-prob-raising-cause}). 
The switch to a relaxed
definition of PR causes with non-strict inequality
seems to be a minor change that identifies more sets as 
causes.
Indeed, the proposed algorithms for checking the SPR and GPR condition
(Section \ref{sec:check})
can easily be modified for the relaxed definition.
While the relaxed definition leads to a questionable
notion of causality
(e.g., $\{\init\}$ would always be a 
recall- and ratio-optimal
SPR cause under the relaxed definition), it could be useful
in combination with other side constraints.
E.g., requiring the relaxed PR condition for all
schedulers that reach a cause state with positive probability and
the existence of a scheduler
where the PR condition with strict inequality holds
might be a useful alternative definition
that agrees with Def.~\ref{def:GPR} for Markov chains.

\vspace*{0.2ex}
\noindent
{\it Relaxing the minimality condition (M).}
As many causality notions of the literature include some minimality constraint, we included condition (M). However, (M) could be dropped without affecting the algorithmic results presented here. This can be useful when the task is to identify components or agents that are responsible for the occurrences of undesired effects. In these cases the cause candidates are fixed (e.g., for each agent $i$, the set of states controlled by agent $i$), but some of them might violate (M).

\vspace*{0.5ex}

\noindent
{\it Future directions}
include  PR causality when causes and effects are path properties and the investigation of other quality measures for PR causes inspired by other indices for binary classifiers used in machine learning or customized for applications of cause-effect reasoning in MDPs. 
More sophisticated notions of probabilistic backward causality and considerations on PR causality with external interventions as in Pearl's do-calculus \cite{Pearl09} are left for future work.

\vspace*{0.5ex}

\noindent\textbf{Acknowledgments}
We would like to thank Simon Jantsch and Clemens Dubslaff for their helpful comments and feedback on the topic of causality in MDPs.

\bibliographystyle{splncs04}
\bibliography{lit}

\newpage

\begin{appendix}

\section{Notation and preliminary results used in the appendix}\label{app:notation}

\subsection{Basic notations}

Let $\cM=(S,\Act,P,\init)$ be an MDP. 
For $\alpha\in \Act$ and $U\subseteq S$, 
$P(s,\alpha,U)$ is a shortform notation for $\sum_{u\in U} P(s,\alpha,u)$.
If $\pi$
is a finite path in $\cM$ then $\last(\pi)$ denotes the last state of $\pi$.
That is, if $\pi = s_0\, \alpha_0 \, s_1 \, \alpha_1 \ldots \alpha_{n-1}\, s_n$
then $\last(\pi)=s_n$.
If $\sched$ is a scheduler then $\pi$ is said to be a $\sched$-path if
$\sched(s_0\, \alpha_0 \ldots \alpha_{i-1} \, s_i)(\alpha_i) >0$
for each $i\in \{0,\ldots,n{-}1\}$.

When dealing with model transformations, we often attach the name of the MDP
as a subscript for the state space, action set, transition probability
function
and the initial state. That is, we then write $S_{\cM}$ for $S$,
$\Act_{\cM}$ for $\Act$, $P_{\cM}$ for $P$ and $\init_{\cM}$ for $\init$.

\begin{notation}[Residual scheduler]
Given an MDP $\cM=(S,\Act,P,\init)$, a scheduler $\sched$, and a path $\pi=s_0\, \alpha_0 \, \dots \, \alpha_{n-1} \, s_n$, the \emph{residual scheduler}
$\residual{\sched}{\pi}$  of $\sched$ after $\pi$ is defined by
\[
\residual{\sched}{\pi}(\zeta) = \sched (\pi\circ \zeta)
\]
for all finite paths $\zeta$ starting in $s_n$. Here, $\pi\circ \zeta$ denotes the concatenation of the paths $\pi$ and $\zeta$.
\Ende
\end{notation}
Intuitively speaking, $\residual{\sched}{\pi}$ behaves like $\sched$ after $\pi$ has already been seen.

\subsection{MR-scheduler in MDPs without ECs}
The following preliminary lemma is folklore (see, e.g., \cite[Theorem 9.16]{Kallenberg20}) and used in the proof of Lemma \ref{lem:MR-sufficient-global-cause} in the following form.

\begin{lemma}[From general schedulers to MR-schedulers in MDPs without ECs]
	\label{lem:from-general-to-MR-schedulers}
	Let $\cM=(S,\Act,P,\init)$ be an MDP without end components.
	Then, for each scheduler  $\sched$ for $\cM$,
	there exists an MR-scheduler $\tsched$ such that:
	\begin{center}
		$\Pr^{\sched}_{\cM}(\Diamond t) \ = \
		\Pr^{\tsched}_{\cM}(\Diamond t)$
		\ \ for each terminal state $t$.
	\end{center}
\end{lemma}

%%%%%%%%%%%%%%%%%%%%%%%%%%%%%%%%%%%%%%%%%%%%%%%%%%%%%%%%%%%%%%%%%%%%%%%%%%%%%%%%

\begin{lemma}[Convex combination of MR-schedulers]
	\label{convex-comb-MR-schedulers}
	Let $\cM$ be an MDP without end components and let
	$\sched$ and $\tsched$ be schedulers for $\cM$ and
	$\lambda$ a real number in the open interval $]0,1[$.
	Then, there exists an MR-scheduler $\usched$
	such that:
	\begin{center}
		$\Pr^{\usched}_{\cM}(\Diamond t) \ = \
		\lambda \cdot \Pr^{\sched}_{\cM}(\Diamond t) \ + \
		(1{-}\lambda) \cdot \Pr^{\tsched}_{\cM}(\Diamond t)$ 
	\end{center}
	for each terminal state $t$.
\end{lemma}

\begin{proof}
	Thanks to Lemma \ref{lem:from-general-to-MR-schedulers}
	we may suppose that $\sched$ and $\tsched$ are MR-schedulers.
	Let
	$$
	f_* \ \ = \ \
	\lambda \cdot \freq{\sched}{*} \ + \ (1{-}\lambda)\cdot \freq{\tsched}{*}
	$$
	where $*$ stands for a state or a state-action pair in $\cM$.
	Let $\usched$ be an MR-scheduler
	defined by
	$$
	\usched(s)(\alpha)
	\ \ = \ \ \frac{f_{s,\alpha}}{f_s}
	$$
	for each non-terminal state $s$ where $f_s>0$ and
	each action $\alpha \in \Act(s)$.
	If $f_s=0$ then $\usched$ selects an arbitrary distribution
	over $\Act(s)$.
	
	Using 
	Lemma \ref{lem:from-general-to-MR-schedulers}
	we then obtain
	$f_* = \freq{\usched}{*}$ where $*$ ranges over all states
	and state-action pairs in $\cM$.
	But this yields:
	\begin{eqnarray*}
		\Pr^{\usched}_{\cM}(\Diamond t) & = & f_t \ \ = \ \
		\lambda \cdot \freq{\sched}{t} \ + \
		(1{-}\lambda) \cdot \freq{\tsched}{t}
		\\[1ex]
		& = &   
		\lambda \cdot \Pr^{\sched}_{\cM}(\Diamond t) \ + \
		(1{-}\lambda) \cdot \Pr^{\tsched}_{\cM}(\Diamond t)
	\end{eqnarray*}
	for each terminal state $t$.
	\qed
\end{proof}

\begin{notation}[Convex combination of schedulers]
	\label{notation:convex-comb-schedulers}  
	{\rm  
		Let $\cM,\sched,\tsched,\lambda$ be as in
		Lemma \ref{convex-comb-MR-schedulers}.
		Then, the notation $\lambda \sched \oplus (1{-}\lambda)\tsched$
		will be used to denote any MR-scheduler $\usched$ as in
		Lemma \ref{convex-comb-MR-schedulers}.
		\Ende
	}
\end{notation}

\subsection{MEC-quotient}\label{app:MEC-quotient}

We now recall the definition of the MEC-quotient, which is a standard concept for the analysis of MDPs \cite{deAlfaro1997}. More concretely, we use a modified version with an additional trap state as in \cite{LICS18-SSP} that serves to mimic behaviors inside an end component of the original MDP.

\begin{definition}[MEC-quotient of an MDP]
	\label{def:MEC-contraction} 
		Let $\cM = (S,\Act,P,\init)$ be an MDP with end components.  
		Let $\cE_1,\ldots,\cE_k$ be the MECs of $\cM$.
		We may suppose without loss of generality that
		the enabled actions of the states 
		are pairwise disjoint, i.e., whenever $s_1,s_2$ are states
		in $\cM$ with $s_1\not= s_2$ then
		$\Act_{\cM}(s_1)\cap \Act_{\cM}(s_2)=\varnothing$.
		This permits to consider $\cE_i$ as a subset of $\Act$.
		Let $U_i$ denote the set of states that belong to $\cE_i$
		and let $U=U_1\cup \ldots \cup U_k$.
		
		The MEC-quotient of $\cM$ is the MDP $\cN = (S',\Act',P',\init')$
		and the function $\iota : S \to S'$ are defined as follows.
		\begin{itemize}
			\item
			The state space $S'$ is $S \setminus U \cup \{s_{\cE_1},\ldots,s_{\cE_k},\bot\}$
			where $s_{\cE_1},\ldots,s_{\cE_k},\bot$ are pairwise distinct fresh states.
			\item
			The function $\iota$ is given by $\iota(s)=s$ if $s\in S \setminus U$
			and $\iota(u)=s_{\cE_i}$ if $u\in U_i$.
			\item
			The initial state of $\cN$ is $\init' = \iota(\init)$.
			\item
			The action set $\Act'$ is $\Act \cup \{\tau\}$ where
			$\tau$ is a fresh action symbol.
			\item
			The set of actions enabled in state $s\in S'$ of $\cN$
			and the transition probabilities are defined as follows:
			\begin{itemize}
				\item
				If $s$ is a state of $\cM$ that does not belong to an MEC of $\cM$
				(i.e., $s\in S \cap S'$) then 
				then $\Act_{\cN}(s)=\Act_{\cM}(s)$ and
				$P'(s,\alpha,s')=P(s,\alpha,\iota^{-1}(s'))$
				for all $s'\in S'$ and $\alpha \in \Act_{\cM}(s)$.
				\item
				If $s=s_{\cE_i}$ is a state representing MEC $\cE_i$ of $\cM$
				then (recall that we may view $\cE_i$ as a set of actions):   
				$$
				\Act_{\cN}\bigl(s_{\cE_i}\bigr) \ = \
				\bigcup_{u\in U_i} (\Act_{\cM}(u) \setminus \cE_i)
				\cup \{\tau\}
				$$
				The $\tau$-action stands for the deterministic transition to
				the fresh state $\bot$, i.e.:
				$$
				P'(s_{\cE_i},\tau,\bot)=1
				$$
				Suppose now that $u\in U_i$ and $\alpha \in \Act_{\cM}(u) \setminus \cE_i$. 
				Then, we set $P'(s_{\cE_i},\alpha,s')=P(u,\alpha,\iota^{-1}(s'))$
				for all $s'\in S'$.   
				\item
				The state $\bot$ is terminal, i.e.,
				$\Act_{\cN}(\bot)=\varnothing$.
			\end{itemize}
		\end{itemize}
\end{definition}
		
Thus, each terminal state of $\cM$ is terminal in its MEC-quotient $\cN$ too.
		Vice versa, every terminal state of $\cN$ is either a terminal state of
		$\cM$ or $\bot$.
		Moreover, $\cN$ has no end components, which implies that
		under every scheduler
		$\tsched$ for $\cN$, a terminal state will be reached with probability 1.

In the main paper (Section \ref{sec:check-GPR}), 
we use the notation $\noeffbot$ rather than $\bot$.

%%%%%%%%%%%%%%%%%%%%%%%%%%%%%%%%%%%%%%%%%%%%%%%%%%%%%%%%%%%%%%%%%

\begin{lemma}[Correspondence of an MDP and its MEC-quotient]
	\label{lem:MEC-contraction-terminal-states-final} 
	Let $\cM$ be an MDP and $\cN$ its MEC-quotient.
	Then, for each scheduler $\sched$ for $\cM$ there is a scheduler
	$\tsched$ for $\cN$ such that
	\begin{equation}
		\label{contraction-property}
		\text{$\Pr^{\sched}_{\cM}(\Diamond t) \ = \
			\Pr^{\tsched}_{\cN}(\Diamond t)$ \
			for each terminal state $t$ of $\cM$}
		\tag{$\dagger$}  
	\end{equation}  
	and vice versa.
	Moreover, if \eqref{contraction-property} holds then
	$\Pr^{\tsched}_{\cN}(\Diamond \bot)$ equals the probability
	for $\sched$ to generate an infinite path in $\cM$ that
	eventually enters and stays forever in an end component.
\end{lemma}  

\begin{proof}
	Given a scheduler $\tsched$ for $\cN$, we pick an
	MD-scheduler $\usched$ such that
	such that $\usched(u) \in \cE_i$ for each $u\in U_i$.
	Then, the corresponding scheduler $\sched$ for $\cM$ behaves as
	$\tsched$ as long as $\tsched$ does not choose the
	$\tau$-transition to $\bot$. As soon as $\tsched$ schedules $\tau$
	then $\sched$ behaves as $\usched$ from this moment on.
	
	Vice versa, if we are given a scheduler $\sched$ for $\cM$ then
	a corresponding scheduler $\tsched$ for $\cN$ mimics
	$\sched$ as long as $\sched$ has not visited a state
	belong to an end component $\cE_i$ of $\cM$.
	Scheduler $\tsched$ ignores $\sched$'s transitions inside
	an MEC $\cE_i$ and takes
	$\beta \in \bigcup_{u\in U_i} (\Act_{\cM}(u) \setminus \cE_i)$
	with the same probability
	as $\sched$ leaves $\cE_i$.
	With the remaining probability mass, $\sched$ stays forever inside $\cE_i$,
	which is mimicked by $\tsched$ by taking the $\tau$-transition to $\bot$.
	
	For the formal definition of $\tsched$, we use the following notation.
	For simplicity, let us assume that $\init \notin U_1 \cup \ldots \cup U_k$.
	This yields $\init= \init'$.
	Given a finite path
	$$
	\pi= s_0\, \alpha_0 \, s_1 \, \alpha_1 \ldots \alpha_{m-1} \, s_m
	$$
	in $\cM$ with $s_0=\init$, let $\pi_{\cN}$ the path in $\cN$ resulting from
	by replacing each maximal path fragment
	$s_h \alpha_h \ldots  \alpha_{j-1} s_j$
	consisting of actions inside an $\cE_i$ with state $s_{\cE_i}$.
	(Here, maximality means if $h > 0$ then $\alpha_{h-1}\notin \cE_i$
	and if $j<m$ then $\alpha_{j+1}\notin \cE_i$.)
	Furthermore, let $p_{\pi}^{\sched}$ denote the probability
	for $\sched$ to generate the path $\pi$ when starting in the
	first state of $\pi$.

	Let $\rho$ be a finite path in $\cN$ with first state $\init$ (recall that we suppose that $\cM$'s initial state
	does not belong to an MEC, which yields $\init=\init'$)
	and $\last(\rho)\not= \bot$.
	Then, $\Pi_{\rho}$ denotes the set of finite paths
	$\pi= s_0\, \alpha_0 \, s_1 \, \alpha_1 \ldots \alpha_{m-1} \, s_m$
	in $\cM$ such that
	(i) $\pi_{\cN}=\rho$ and
	(ii) if $s_m \in U_i$ then $\alpha_{m-1}\notin \cE_i$.
	The formal definition of scheduler $\tsched$ is now as follows.
	Let $\rho$ be a finite path in $\cN$ where the last state $s$ of $\rho$ is
	non-terminal.
	If $s$ is a state of $\cM$ that does not belong to an MEC of $\cM$
	and $\beta \in \Act_{\cM}(s)$
	then:
	$$
	\tsched(\rho)(\beta)
	\ \ = \ \ \sum_{\pi\in \Pi_{\rho}} p_{\pi}^{\sched} \cdot \sched(\pi)(\beta)
	$$
	If $s = s_{\cE_i}$ and $\beta \in \Act_{\cN}\bigl(s_{\cE_i}\bigr)\setminus \{\tau\}$
	then
	$$
	\tsched(\rho)(\beta)
	\ \ = \ \
	\sum_{\pi\in \Pi_{\rho}} p_{\pi}^{\sched} \cdot
	\Pr^{\residual{\sched}{\pi}}_{\cM,\last(\pi)}\bigl( \
	\text{``leave $\cE_i$ via action $\beta$''} \ \bigr) 
	$$
	where ``leave $\cE_i$ via action $\beta$'' means the existence of a prefix
	whose action sequence consists of actions inside $\cE_i$ followed by action
	$\beta$. The last state of this prefix, however,
	could be a state of $U_i$. (Note $\beta \in \Act_{\cN}(s_{\cE_i})$ means that
	$\beta$ could have reached a state outside $U_i$, but
	there might be states inside $U_i$ that are accessible via $\beta$.)
	Similarly,  
	$$
	\tsched(\rho)(\tau)
	\ \ = \ \
	\sum_{\pi\in \Pi_{\rho}} p_{\pi}^{\sched} \cdot
	\Pr^{\residual{\sched}{\pi}}_{\cM,\last(\pi)}\bigl( \
	\text{``stay forever in $\cE_i$''} \ \bigr)
	$$
	where ``stay forever in $\cE_i$'' means that only actions
	inside $\cE_i$ are performed.
	By induction on the length of $\rho$ we obtain: 
	$$
	p_{\rho}^{\tsched} \ \ = \ \ \sum_{\pi\in \Pi_{\rho}} p_{\pi}^{\sched}
	$$   
	But this yields 
	$\Pr^{\sched}_{\cM}(\Diamond t)$ = $\Pr^{\tsched}_{\cN}(\Diamond t)$
	for each terminal state $t$ of $\cM$.
	Moreover,
	$$\Pr^{\sched}_{\cM}(\text{``eventually enter and stay forever in $\cE_i$''})
	$$
	equals the probability for $\tsched$ to
	reach the terminal state $\bot$ via a path of the form
	$\rho \, \tau \, \bot$ where $\last(\rho)=s_{\cE_i}$.
	\qed
\end{proof}

\section{Omitted Proofs and Details of Section \ref{sec:SPR-GPR}}\label{app:section_3}

\strictimpliesglobal*

\begin{proof}
Assume that $\Cause$ is a SPR cause for $\Effect$ in $\cM$ and let $\sched$ be a scheduler that reaches $\Cause$ with positive probability.
Further, let 
\[
C_{\sched}\eqdef \{c\in \Cause\ \mid \ \Pr^{\sched}_{\cM}((\neg \Cause )\Until c)>0\}
\]
and 
\[
m\eqdef \min_{c\in C_{\sched}} \Pr^{\sched}_{\cM}(\ \Diamond \Effect \ | \
			(\neg \Cause) \Until c \ ).
\]
As $\Cause$ is a SPR cause, $ m> \Pr^{\sched}_{\cM}(\Diamond \Effect)$.
The set of $\sched$-paths satisfying $\Diamond \Cause$ is the disjoint union of the sets of $\sched$-paths satisfying $(\neg \Cause) \Until c$ with $c \in C_{\sched}$.
Hence,
\[
\Pr^{\sched}_{\cM}(\Diamond \Effect \mid  \Diamond \Cause)=\frac{\sum_{c\in C_{\sched}}\Pr^{\sched}_{\cM}(\Diamond \Effect  \mid
			(\neg \Cause) \Until c  )\cdot \Pr^{\sched}_{\cM}((\neg \Cause) \Until c  )}{\sum_{c\in C_{\sched}}\Pr^{\sched}_{\cM}((\neg \Cause) \Until c) }\geq m.
\]
As $m>\Pr^{\sched}_{\cM}(\Diamond \Effect)$, the GPR condition (GPR) is satisfied under $\sched$. \qed
\end{proof}

\section{Omitted Proofs and Details of Section \ref{sec:check}}\label{app:section_4}

\subsection{Switch to the MDP $\wminMDP{\cM}{\Cause}$}\label{app:section_4_Mcause}

\wmincriterionPRcauses*

Obviously, condition (M) holds for $\Cause$ in $\cM$ if and only if
condition (M) holds for $\Cause$ in $\wminMDP{\cM}{\Cause}$.  
Furthermore, it is clear all SPR resp. GPR causes of $\cM$
 are SPR resp. GPR causes in $\wminMDP{\cM}{\Cause}$. 
So, it remains to prove the converse direction.
This will be done
in Lemma \ref{app:lem:criterion-strict-prob-raising} for SPR causes and
in Lemma \ref{app:lem:criterion-global-prob-raising} for GPR causes.

\begin{lemma}%
	[Criterion for strict probability-raising causes]
	\label{app:lem:criterion-strict-prob-raising}
	Suppose $\Cause$ is
	an SPR cause for $\Effect$ in $\wminMDP{\cM}{\Cause}$.
	Then, $\Cause$ is an SPR cause for $\Effect$ in $\cM$.
\end{lemma}

\begin{proof}
	We fix a state $c\in \Cause$. Recall also that we assume the states in $\Effect$ to be terminal.
	Let $\psi_c=(\neg \Cause) \Until c$, $w_c = \Pr_{\cM, c}^{\min}(\Diamond \Eff)$ and
	let $\Usched_c$ denote the set of all schedulers $\usched$
	for $\cM$ such that
	\begin{itemize}
		\item
		$\Pr_{\cM}^{\usched}(\psi_c)>0$ and
		\item
		$\Pr^{\residual{\usched}{\pi}}_{\cM,c}(\Diamond \Effect)=w_c$
		for each finite $\usched$-path $\pi$ from $\init$ to $c$.
	\end{itemize}
	Clearly, $\Pr^{\usched}_{\cM}( \Diamond c \wedge \Diamond \Effect)
	= \Pr_\cM^\usched(\Diamond c)\cdot w_c$
	for $\usched \in \Usched_c$.
	
	As $\Cause$ is an SPR cause in $\wminMDP{\cM}{\Cause}$
	we have:
	\begin{equation}
		\label{SPR-1} 
		w_c \ > \ \Pr^{\usched}_{\cM}(\Diamond \Effect)
		\quad \text{for all schedulers $\usched \in \Usched_c$}
		\tag{SPR-1}    
	\end{equation}
	The task is to prove that the SPR condition 
	holds for $c$ and all schedulers of $\cM$ with
	$\Pr^{\sched}_{\cM}(\psi_c) >0$.
	
	Suppose $\sched$ is a scheduler for $\cM$ with
	$\Pr^{\sched}_{\cM}(\psi_c) >0$. Then:
	\[\Pr^{\sched}_{\cM}(\psi_c \wedge \Diamond \Eff)
	\ \ \geqslant  \ \
	\Pr^{\sched}_{\cM}(\psi_c)\cdot w_c
	\]
	Moreover, there
	exists a scheduler $\usched =\usched_{\sched} \in \Usched_c$
	with
	\begin{center}
		$\Pr^{\sched}_{\cM}(\psi_c) = \Pr^{\usched}_{\cM}(\psi_c)$
		\ \ and \ \
		$\Pr^{\sched}_{\cM}((\neg \psi_c) \wedge \Diamond \Effect)=
		\Pr_\cM^\usched((\neg \psi_c) \wedge \Diamond \Effect)$.
	\end{center}
	To see this, consider the scheduler $\usched$ that
	behaves as $\sched$ as long as $c$ is not reached.
	As soon as $\usched$ has reached $c$, scheduler $\usched$
	switches mode and
	behaves
	as an MD-scheduler minimizing the
	probability to reach an effect state.
	
	The SPR condition
	holds for $c$ and $\sched$ if and only if
	\begin{equation}
		\label{SPR-2}
		\frac{\Pr^{\sched}_{\cM}( \psi_c \wedge \Diamond \Effect)}
		{\Pr^{\sched}_{\cM}(\psi_c)}
		\ \ > \ \
		\Pr^{\sched}_{\cM}(\Diamond \Effect)
		\tag{$\dagger$}
	\end{equation}
	As
	\begin{eqnarray*}
		\Pr^{\sched}_{\cM}(\Diamond \Effect)
		& = &
		\Pr^{\sched}_{\cM}( \psi_c \wedge \Diamond \Effect)
		\ + \
		\Pr^{\sched}_{\cM}(
		(\neg \psi_c) \wedge \Diamond \Effect )
	\end{eqnarray*}
	we can equivalently convert condition \eqref{SPR-2} 
	for $c$ and $\sched$ to
	\begin{align*}
		\label{SPR-3}
		\Pr^{\sched}_{\cM}(\psi_c \wedge \Diamond \Effect)
		\cdot
		\frac{1-\Pr^{\sched}_{\cM}(\psi_c)}{\Pr^{\sched}_{\cM}(\psi_c)}
		\ \ > \ \
		\Pr^{\sched}_{\cM}((\neg \psi_c) \wedge \Diamond \Effect )
		\tag{$\ddagger$}
	\end{align*}
	So, the remaining
	task is now to derive \eqref{SPR-3} from \eqref{SPR-1}.
	
	\eqref{SPR-1} applied to scheduler $\usched=\usched_{\sched}$ yields:
	$$
	\begin{array}{lcl}         
		w_c & \ \ > \ \ &
		\Pr^{\usched}_{\cM}(\psi_c \wedge  \Diamond \Effect)
		\ + \
		\Pr^{\usched}_{\cM}((\neg \psi_c) \wedge  \Diamond \Effect)
		\\
		\\[0ex]
		& = &
		\Pr^{\sched}_{\cM}(\psi_c)\cdot w_c
		\ + \
		\Pr^{\sched}_{\cM}((\neg \psi_c) \wedge  \Diamond \Effect)
	\end{array}   
	$$
	We conclude:
	\begin{eqnarray*}
		\Pr^{\sched}_{\cM}(\psi_c \wedge \Diamond \Effect)
		\cdot
		\frac{1-\Pr^{\sched}_{\cM}(\psi_c)}{\Pr^{\sched}_{\cM}(\psi_c)}
		& \ \geqslant \ &
		\Pr^{\sched}_{\cM}(\psi_c) \cdot w_c
		\cdot
		\frac{1-\Pr^{\sched}_{\cM}(\psi_c)}{\Pr^{\sched}_{\cM}(\psi_c)}
		\\
		\\[0ex]
		& = &
		\bigl(1-\Pr^{\sched}_{\cM}(\psi_c)\bigr)\cdot w_c
		\\[1ex]
		& > &
		\Pr^{\sched}_{\cM}((\neg \psi_c) \wedge  \Diamond \Effect)
	\end{eqnarray*}
	Thus, \eqref{SPR-3} holds for $c$ and $\sched$.
	\qed
\end{proof}

\begin{lemma}[Criterion for GPR causes]
	\label{app:lem:criterion-global-prob-raising}
	Suppose $\Cause$ is
	an GPR cause for $\Effect$ in $\wminMDP{\cM}{\Cause}$.
	Then, $\Cause$ is an GPR cause for $\Effect$ in $\cM$.
\end{lemma}

\begin{proof}
From the assumption that $\Cause$ is
	an GPR cause for $\Effect$ in $\wminMDP{\cM}{\Cause}$, we can conclude that 
	the GPR condition \eqref{GPR} holds
		for all schedulers $\sched$
		that satisfy \[\Pr_{\cM}^{\sched}(\Diamond \Cause)>0\] and
		\[
		\Pr^{\residual{\sched}{\pi}}_{\cM,c}(\Diamond \Effect)
		\ \ = \ \ \Pr^{\min}_{\cM,c}(\Diamond \Effect)
		\]
		for each finite $\sched$-path from the initial state
		$\init$ to a state $c \in \Cause$.

	To prove that the GPR condition \eqref{GPR}
		holds
		for all schedulers $\sched$ that satisfy $\Pr_{\cM}^{\sched}(\Diamond \Cause)>0$, we introduce the following notation:
		We write
	\begin{itemize}
		\item
		$\Ssched_{>0}$ for the set of all schedulers $\sched$ such that
		$\Pr_{\cM}^{\sched}(\Diamond \Cause)>0$,
		\item
		$\Ssched_{>0,\min}$ for the set of all schedulers
		with $\Pr_{\cM}^{\sched}(\Diamond \Cause)>0$ such that
		$$
		\Pr^{\residual{\sched}{\pi}}_{\cM,c}(\Diamond \Effect)
		\ \ = \ \ \Pr^{\min}_{\cM,c}(\Diamond \Effect)
		$$
		for each finite $\sched$-path from the initial state
		$\init$ to a state $c \in \Cause$.
	\end{itemize} 
	It now suffices to show
	that for each scheduler $\sched \in \Ssched_{>0}$ there exists a scheduler
	$\sched' \in \Ssched_{>0,\min}$ such that
	if \eqref{GPR} holds $\sched'$ then \eqref{GPR} holds for $\sched$.
	So, let $\sched\in\Ssched_{>0}$. 
	
	For $c\in \Cause$, let $\Pi_c$ denote the set of
	finite paths
	$\pi=s_0 \, \alpha_0\,  s_1 \, \alpha_1 \ldots \alpha_{n-1}\, s_n$
	with $s_0=\init$, $s_n=c$ and
	$\{s_0,\ldots,s_{n-1} \}\cap (\Cause \cup \Effect) =\varnothing$.
	Let
	$$
	w_{\pi}^{\sched} \ = \ \Pr^{\residual{\sched}{\pi}}_{\cM,c}(\Diamond \Effect)
	$$
	Furthermore, let $p_{\pi}^{\sched}$ denote the probability for
	(the cylinder set of) $\pi$ under scheduler $\sched$.
	Then
	$$
	\Pr^{\sched}_{\cM}((\neg \Cause)\Until c)
           \ = \ \sum_{\pi \in \Pi_c} p_{\pi}^{\sched}
	$$
	Moreover:
	$$
	\Pr^{\sched}_{\cM}(\Diamond \Effect)
	\ \ = \ \ 
	\Pr^{\sched}_{\cM}(\neg \Cause \Until \Effect) \ + 
	\sum_{c\in \Cause} \sum_{\pi \in \Pi_c}
	\!\!\! p_{\pi}^{\sched} \cdot w_{\pi}^{\sched}
	$$
	and,
	$$
	\Pr^{\sched}_{\cM}(\ \Diamond \Effect \ | \ \Diamond \Cause \ )
	\ \ = \ \
	\frac{1}{\Pr_{\cM}^{\sched}(\Diamond \Cause)} \cdot \!\!\!
	\sum_{c\in \Cause} \
	\sum_{\pi\in \Pi_c} p_{\pi}^{\sched} \cdot w_{\pi}^{\sched}      
	$$
	Thus, the condition \eqref{GPR}
	holds for the scheduler
	$\sched \in \Ssched_{>0}$ if and only if
	\begin{eqnarray*}
		\Pr^{\sched}_{\cM}(\neg \Cause \Until \Effect)
		+
		\sum_{c\in \Cause}\sum_{\pi \in \Pi_c} p_{\pi}^{\sched}\cdot w_{\pi}^{\sched}
		&  <  &
		\frac{1}{\Pr_{\cM}^{\sched}(\Diamond \Cause)} \cdot \!\!\!
		\sum_{c\in \Cause} \sum_{\pi\in \Pi_c} p_{\pi}^{\sched} \cdot w_{\pi}^{\sched}
	\end{eqnarray*}
	The latter is equivalent to:
	\begin{align*}
	&\Pr_{\cM}^{\sched}(\Diamond \Cause) \cdot \Pr^{\sched}_{\cM}(\neg \Cause \Until \Effect)
	\ \ + \ \
	\Pr_{\cM}^{\sched}(\Diamond \Cause) \cdot \!\!\!
	\sum_{c\in \Cause} \ \sum_{\pi \in \Pi_c} p_{\pi}^{\sched}\cdot w_{\pi}^{\sched} \\
	< &
	\sum_{c\in \Cause}\ \sum_{\pi \in \Pi_c} p_{\pi}^{\sched}\cdot w_{\pi}^{\sched}
	\end{align*}
	which again is equivalent to:
	\begin{align*}
		\label{GPR-2}
		&\Pr_{\cM}^{\sched}(\Diamond \Cause) \cdot \Pr^{\sched}_{\cM}(\neg \Cause \Until \Effect) \\
		< \ \ &
		\bigl(1-\Pr_{\cM}^{\sched}(\Diamond \Cause)\bigr) \cdot \!\!\!
		\sum_{c\in \Cause} \
		\sum_{\pi \in \Pi_c} p_{\pi}^{\sched}\cdot w_{\pi}^{\sched}
		\tag{GPR-2}  
	\end{align*}
	Pick an MD-scheduler
	$\tsched$ that minimizes the probability to reach $\Effect$
	from every state.
	In particular, $w_c = w_\pi^{\tsched} \leqslant w_\pi^{\sched}$
	for every state $c\in \Cause$ and every path $\pi\in \Pi_c$
	(recall that $w_c =  \Pr^{\min}_{\cM,c}(\Diamond \Effect)$).
	Moreover, the scheduler $\sched$ can be transformed into a
	scheduler $\sched_{\tsched}\in \Ssched_{>0,\min}$
	that is ``equivalent'' to $\sched$ with respect to the
	global probability-raising condition.
	More concretely, let $\sched_\tsched$ denote the scheduler 
	that behaves as $\sched$ as long as $\sched$ has not yet visited a state
	in $\Cause$ and behaves as $\tsched$ as soon as a state in $\Cause$ has been
	reached. Thus,
	$p_{\pi}^{\sched}=p_{\pi}^{\sched_{\tsched}}$
	and
	$\residual{\sched_{\tsched}}{\pi}=\tsched$
	for each $\pi\in \Pi_c$.
	This yields that the probability to reach $c\in \Cause$ from
	$\init$ is the same under
	$\sched$ and $\sched_{\tsched}$, i.e.,
        $\Pr^{\sched}_{\cM}(\Diamond c)=
         \Pr^{\sched_{\tsched}}_{\cM}(\Diamond c)$.
	Therefore $\Pr^{\sched}_{\cM}(\Diamond \Cause)=\Pr^{\sched_{\tsched}}_{\cM}(\Diamond \Cause)$.
	The latter implies that $\sched_{\tsched}\in \Ssched_{>0}$,
	and hence $\sched_{\tsched}\in \Ssched_{>0,\min}$.
	Moreover, $\sched$ and $\sched_{\tsched}$ reach $\Effect$ without visiting $\Cause$ with the same
	probability, i.e., $\Pr^{\sched}_{\cM}(\neg \Cause \Until \Effect)=\Pr^{\sched_{\tsched}}_{\cM}(\neg \Cause \Until \Effect)$.
	
	But this yields: if \eqref{GPR-2} holds for $\sched_{\tsched}$ then
	\eqref{GPR-2} holds for $\sched$.
	As \eqref{GPR-2} holds for $\sched_{\tsched}$ by assumption, this completes the proof.
		\qed
\end{proof}

\subsection{Proofs to Section \ref{sec:SPR_check}}\label{app:proofs_SPR_check}

\soundnessSPRalgo*
\begin{proof} 
	First, we show the soundness of Algorithm \ref{alg:SPR-check}.
	By the virtue of Lemma \ref{lemma:wmin-criterion-PR-causes} it
	suffices to show that Algorithm \ref{alg:SPR-check} returns the correct
	answers ``yes'' or ``no'' when the task is to check whether the
	singleton $\Cause = \{c\}$ is an SPR cause in $\cN=\wminMDP{\cM}{c}$.
	Recall the notation $q_s =\Pr^{\max}_{\wminMDP{\cM}{c},s}(\Diamond \Effect)$.
	We abbreviate $q=q_\init$.
	Note that $(\neg \Cause) \Until c$ is equivalent to $\Diamond c$.

	For every scheduler $\sched$ of $\cN$ we have
	$\Pr^{\sched}_{\cN,c}(\Diamond \Effect)=w_c$.
        Thus, $\Pr^{\sched}_{\cN}( \Diamond \Effect \ | \ \Diamond c)=w_c$
        if $\sched$ is a scheduler of $\cN$ with $\Pr^\sched_\cN(\Diamond c)>0$.

	Algorithm \ref{alg:SPR-check} correctly answers ``no'' (case 2 or 3.1)
	if $w_c=0$.
	Let us now suppose that $w_c>0$. 
	Thus, the SPR condition for $c$ reduces to $\Pr^\sched_\cN(\Diamond \Effect) < w_c$ for all schedulers $\sched$ of $\cN$ with $\Pr^\sched_\cN(\Diamond c)>0$.
	\begin{itemize}
		\item  
		In case 1 of Algorithm \ref{alg:SPR-check} the answer ``yes'' is sound
		as then $\Pr^{\max}_{\cN}(\Diamond \Effect) = q < w_c$.
		\item
		For case 2 (i.e., if $q > w_c$),
		let $\tsched$ be an MD-scheduler 
		with $\Pr^{\tsched}_{\cN,s}(\Diamond \Effect) = q_s$
		for each state $s$ and pick an MD-scheduler $\sched$ with
		$\Pr^{\sched}_{\cN}(\Diamond c)>0$.
		It is no restriction to suppose that
		$\tsched$ and $\sched$ realize the same end components of $\cN$.
		(Note that if state $s$ belongs to an end component
		that is realized by $\tsched$ then $s$ contained in a bottom strongly
		connected component of the Markov chain induced by $\tsched$.
		But then
		$q_s=0$, i.e., no effect state is reachable from $s$ in $\cN$.
		Recall that all effect states are terminal and thus not contained in
		end components.
		But then we can safely assume that $\tsched$ and $\sched$ schedule
		the same action for state $s$.)
		Let $\lambda$ be any real number with $1 > \lambda > \frac{w_c}{q}$
		and let $\cK$ denote the sub-MDP of $\cN$ with state space $S$
		where the enabled
		actions of state $s$ are the actions scheduled for $s$ under one of the
		schedulers $\tsched$ or $\sched$.
		Let now $\usched$ be the MR-scheduler
		$\lambda \tsched \oplus (1{-}\lambda)\sched$
		defined as in Notation \ref{notation:convex-comb-schedulers}
		for the EC-free MDP resulting from $\cK$ when collapsing
		$\cK$'s end components into a single terminal state.
		For the states belonging to an end component of $\cK$,
		$\usched$ schedules the same action as $\tsched$ and $\sched$.
		Then, $\Pr^{\usched}_{\cN}(\Diamond t)=\lambda \Pr^{\tsched}_{\cN}(\Diamond t)+(1{-}\lambda) \Pr^{\sched}_{\cN}(\Diamond t)$
		for all terminal states $t$ of $\cN$ and $t=c$.
		Hence:
		$$
		\Pr_{\cN}^{\usched}(\Diamond c)
		\ \ \geqslant \ \
		(1{-}\lambda)\cdot \Pr^{\sched}_{\cM}(\Diamond c) \ \ > \ \ 0
		$$  
		and
		$$
		\Pr^{\usched}_{\cN}(\Diamond \Effect)
		\ \ \geqslant \ \
		\lambda \cdot \Pr^{\tsched}_{\cM}(\Diamond \Effect)
		\ \ = \ \ \lambda \cdot q \ \ > \ \ w_c
		$$
		Thus, scheduler $\usched$ is a witness
		why \eqref{SPR} does not hold for $c$.
		\item	
		For case 3.1 pick an MD-scheduler $\sched$ of $\wminMDPmax{\cM}{c}$ such that
		$c$ is reachable from $\init$ via a $\sched$-path and
		$\Pr^{\sched}_{\cN,s}(\Diamond \Effect)=q_s$ for all states $s$.
		Hence, \eqref{SPR} does not hold for $c$ and the scheduler $\sched$.
		\item	
		The last case 3.2 has the property that
		$\Pr^{\sched}_{\cN}(\Diamond c)=0$ for all schedulers $\sched$ for $\cN$
		with $\Pr^{\sched}_{\cN}(\Diamond \Effect)=q=w_c$.
		But then $\Pr^{\sched}_{\cN}(\Diamond c)>0$ implies
		$\Pr^{\sched}_{\cN}(\Diamond \Effect) < w_c$ 
                as required in \eqref{SPR}.
	\end{itemize}	
	The polynomial runtime of Algorithm \ref{alg:SPR-check}
	follows from the fact that minimal and maximal reachability probabilities and hence also the MDPs $\cN=\wminMDP{\cM}{c}$ and its sub-MDP $\wminMDPmax{\cM}{c}$ can be computed in polynomial time.	
	\qed
\end{proof}

\begin{lemma}[Criterion for the existence of PR causes (Lemma \ref{lem:existence-of-prob-raising-causes})]
	Let $\cM$ be an MDP and $\Effect$ a nonempty set of states. The following statements are equivalent:
	\begin{itemize}
	\item[(a)]
	$\Effect$ has an SPR cause in $\cM$,
	\item[(b)]
	$\Effect$ has a GPR cause in $\cM$,
	\item[(c)]
	there is a state $c_0\in S \setminus \Effect$ such that
	the singleton $\{c_0\}$ is an SPR cause (and therefore a GRP cause) for $\Effect$ in $\cM$.
	\end{itemize}
    In particular, the existence of SPR/GPR causes can be checked
    with Algorithm \ref{alg:SPR-check} in polynomial time.
\end{lemma}

\begin{proof}
	Obviously, statement (c) implies statements (a) and (b).
	The implication ``(a) $\Longrightarrow$ (b)'' follows from
	Lemma \ref{lemma:strict-implies-global}.
	We now turn to the proof of ``(b) $\Longrightarrow$ (c)''.
	For this, we assume that we are given a
	GPR cause $\Cause$ for $\Effect$ in $\cM$.
	For $c\in \Cause$, let $w_c=\Pr^{\min}_{\cM,c}(\Diamond \Effect)$.
	Pick a state $c_0\in \Cause$ such that
	$w_{c_0} = \max \{ w_c : c \in \Cause\}$.
	For every scheduler $\sched$ for $\cM$ that minimizes the
	effect probability whenever it visits a state in $\Cause$,
	and visits $\Cause$ with positive probability,
	the conditional probability
	$\Pr^{\sched}_{\cM}(\Diamond \Effect |\Diamond \Cause)$ is
	a weighted average of the values $w_c$, $c \in \Cause$,
	and thus bounded by $w_{c_0}$.
	Using Lemma \ref{lemma:wmin-criterion-PR-causes}
	it is now easy to see that
	$\{c_0\}$ is both an SPR and a GPR cause for $\Effect$.
	\qed
\end{proof}

\subsection{Construction justifying assumptions (A1)-(A3)}\label{app:construction_assumptions}

In Section \ref{sec:check}, we are given an MDP $\cM=(S,\Act,P,\init)$ with two disjoint sets of states $\Cause$ and $\Effect$. The states in $\Effect$ are terminal.

Here, we  provide the missing details of the transformation for the assumptions (A1)-(A3) in Section \ref{sec:check} that are listed here again:

\begin{description}
       \item [(A1)]
          $\Effect=\{\effuncov,\effcov\}$ consists of two terminal states.
          
	\item [(A2)] %
          For every state $c\in \Cause$, there is only a single enabled
          action, say $\Act(c)=\{\gamma\}$, and
          there exists $w_c\in [0,1]\cap \Rational$ such that
	  $P(c,\tau,\effcov)=w_c$ and
	  $P(c,\tau,\noeffc)=1-w_c$.
          where $\noeffc$ is a terminal non-effect state
          and $\noeffc$ and $\effcov$ are only accessible via the
          $\gamma$-transition from the states $c\in \Cause$.
          
        \item [(A3)]
          $\cM$ has no end components and 
          there is a further terminal state $\noeffbot$
          and an action $\tau$ such that 
          $\tau \in \Act(s)$ implies $P(s,\tau,\noeffbot)=1$.
\end{description}
The terminal states $\effunc$, $\effcov$, $\noeffc$ and $\noeffbot$ 
are supposed to be pairwise distinct. $\cM$ can have further terminal
states representing true negatives. These could be identified with
$\noeffbot$, but this is irrelevant for our purposes.

Assumptions (A1) and (A2) are established as described in the main body of the paper (Section \ref{sec:check-GPR}) by switching to the MDP $\wminMDP{\cM}{\Cause}$ -- which is justified by Lemma \ref{lemma:wmin-criterion-PR-causes} -- and by renaming and collapsing terminal states resulting in an MDP $\cM^\prime$.

Now, let $\cN$ be the MEC-quotient of $\cM^\prime$ (see Appendix \ref{app:MEC-quotient}). 
Let $\noeffbot$ be the state to which we add a $\tau$-transition with probability $1$ from each MEC that we collapse in the MEC-quotient.
That is, $\noeffbot=\bot$ with the notations of Definition \ref{def:MEC-contraction}. 

\begin{lemma}\label{lem:probabilities_MEC-quotient}
For each scheduler $\sched$ for $\wminMDP{\cM}{\Cause}$, there is a scheduler $\tsched$ for $\cN$, and vice versa, such that 
\begin{itemize}
\item
$\Pr^{\sched}_{\wminMDP{\cM}{\Cause}}(\Diamond \Effect) = \Pr^{\tsched}_{\cN}(\Diamond \Effect)$,
 \item
 $\Pr^{\sched}_{\wminMDP{\cM}{\Cause}}(\Diamond \Cause) = \Pr^{\tsched}_{\cN}(\Diamond \Cause)$, and
  \item 
 $\Pr^{\sched}_{\wminMDP{\cM}{\Cause}}(\Diamond \Cause \land \Diamond \Effect) = \Pr^{\tsched}_{\cN} (\Diamond \effcov)$.
\end{itemize}
\end{lemma}
\begin{proof}
By Lemma \ref{lem:MEC-contraction-terminal-states-final}, there is a scheduler $\tsched$ for $\cN$ for each scheduler $\sched$ for $\cM^{\prime}$ such that each terminal state is reached with the same probability under $\tsched$ in $\cN$ and under $\sched$ in $\cM^{\prime}$.
The state $\effcov$ is also present in $\wminMDP{\cM}{\Cause}$ under the name $\eff$ and reached with the same probability as in $\cM^{\prime}$ when $\sched$ is considered as a scheduler for $\wminMDP{\cM}{\Cause}$. The state $\eff$ is furthermore reached in $\wminMDP{\cM}{\Cause}$ if and only if $\Diamond \Cause \land \Diamond \Effect$ is satisfied along a run.
 The set of terminal states in $\Effect$ is obtained from the set $\Effect$ in $\wminMDP{\cM}{\Cause}$ by collapsing states. As a scheduler $\sched$ can be viewed as a scheduler for both MDPs and these MDPs agree except for the terminal states, the first equality follows as well. 
As the probability to reach $\Cause$ is the sum of the probabilities to reach the terminal states $\effcov$ and $\noeffc$ in $\cN$ and $\cM^{\prime}$ and as these states are only renamed  in the transition from  $\wminMDP{\cM}{\Cause}$ to $\cM^{\prime}$, the claim follows. \qed
\end{proof}

From Lemma \ref{lem:probabilities_MEC-quotient} and Lemma \ref{lemma:wmin-criterion-PR-causes}, we conclude the following corollary that justifies working under assumptions (A1)-(A3) in Section \ref{sec:check}.

\begin{corollary}\label{cor:GPR_MEC}
The set $\Cause$ is a SPR/GPR cause for $\Effect$ in $\cM$ if and only if $\Cause$ is a SPR/GPR cause for $\Effect$ in $\cN$.
\end{corollary}

\begin{proof}
By Lemma \ref{lem:probabilities_MEC-quotient}, for each scheduler $\sched$ for $\wminMDP{\cM}{\Cause}$, there is a scheduler $\tsched$ for $\cN$ such that all relevant probabilities agree, and vice versa. So, $\Cause$ is a GPR cause for $\Effect$ in $\wminMDP{\cM}{\Cause}$ if and only if it is a GPR cause in $\cN$.
By Lemma \ref{lemma:wmin-criterion-PR-causes}, 
$\Cause$ is a GPR cause for $\Effect$ in $\wminMDP{\cM}{\Cause}$ if and only if it is a GPR cause in $\cM$.
\qed
\end{proof}

\subsection{Proofs to Section \ref{sec:check-GPR}}\label{app:GPR_check}

\MRschedulersufficient*

\begin{proof}
	Let $\usched$ be a scheduler with
	$\Pr^{\usched}_{\cM}(\Diamond \Cause) > 0$ violating \eqref{GPR-1}, i.e.:
	\[
	\Pr_{\cM}^{\usched}(\Diamond \Cause) \cdot
	\Pr^{\usched}_{\cM}( \Diamond \effuncov)
	\ < \ 
	\bigl(1{-}\Pr_{\cM}^{\usched}(\Diamond \Cause)\bigr)
	\cdot 
	\!\!\!\!\!\sum_{c\in \Cause} \!\!\!\!\!
	\Pr^{\usched}_{\cM}(\Diamond c) \cdot w_c.
	\]

	We will show  how to transform $\usched$ into an
	MR-scheduler $\tsched$ that schedules the $\tau$-transitions to
	$\noeffbot$ with probability 0 or 1.
	For this, we regard the set $U$ of states
	$u$ that have a $\tau$-transition to $\noeffbot$
	(recall that then $P(u,\tau,\noeffbot)=1$)
	and where $0 < \usched(u)(\tau) < 1$.
	We now process the $U$-states in an arbitrary order,
	say $u_1,\ldots,u_k$,
	and generate a sequence $\tsched_0=\usched,\tsched_1,\ldots,\tsched_k$
	of MR-schedulers such that for $i \in \{1,\ldots,k\}$:
	\begin{itemize}
		\item
		$\tsched_i$ refutes the GPR condition
		(or equivalently condition \eqref{GPR-1} from
		Lemma \ref{lem:GPR-poly-constraint})
		\item 
		$\tsched_i$
		agrees with $\tsched_{i-1}$ for all states but $u_i$,
		\item
		$\tsched_i(u_i)(\tau)\in \{0,1\}$.
	\end{itemize}
	Thus, the final scheduler $\tsched_k$ satisfies the
	desired properties.
	
	We now explain how to derive $\tsched_i$ from $\tsched_{i-1}$.
	Let $i\in \{1,\ldots,k\}$, $\vsched=\tsched_{i-1}$, $u=u_i$
	and $y=1{-}\vsched(u)(\tau)$. Then, $0<y<1$ (as $u\in U$ and
	by definition of $U$) and
	$y= \sum_{\alpha\in \Act(u)\setminus \{\tau\}} \vsched(u)(\alpha)$.
	
	For $x\in [0,1]$, let $\vsched_x$ denote the MR-scheduler that
	agrees with $\vsched$ for all states but $u$, for which
	$\vsched_x$'s decision is as follows:
	$$
	\vsched_x(u)(\tau)=1{-}x,
	\qquad
	\vsched_x(u)(\alpha)=\vsched(u)(\alpha) \cdot \frac{x}{y}
	\quad \text{for $\alpha \in \Act(u)\setminus \{\tau\}$}
	$$ 
	Obviously, $\vsched_y=\vsched$.
	We now show that at least one of the two MR-schedulers
	$\vsched_0$ or $\vsched_1$ also refutes the GPR condition.
	For this, we suppose by contraction that this is not the case,
	which means that the GPR condition holds for both.
	
	Let $f : [0,1]\to [0,1]$ be defined by
	$$
	f(x) \ = \
	\Pr_{\cM}^{\vsched_x}(\Diamond \Cause) \cdot
	\Pr^{\vsched_x}_{\cM}( \Diamond \effuncov)
	- 
	\bigl(1{-}\Pr_{\cM}^{\vsched_x}(\Diamond \Cause)\bigr)
	\cdot 
	\!\!\!\!\!\sum_{c\in \Cause} \!\!\!\!\!\!
	\Pr^{\vsched_x}_{\cM}(\Diamond c) \cdot w_c
	$$
	As $\vsched = \vsched_y$ violates \eqref{GPR-1},
	while $\vsched_0$ and $\vsched_1$
	satisfy \eqref{GPR-1} we obtain:
	$$
	f(0), f(1) < 0  \qquad \text{and} \qquad f(y) \geqslant 0
	$$
	We now split $\Cause$ into the set $C$ of states $c\in \Cause$
	such that there is a $\vsched$-path from $\init$ to $c$ that traverses
	$u$ and $D=\Cause \setminus C$.
	Thus, $\Pr^{\vsched_x}_{\cM}(\Diamond \Cause)= p_x + p$
	where $p_x=\Pr^{\vsched_x}(\Diamond C)$ and
	$p=\Pr^{\vsched}(\Diamond D)$.
	Similarly, $\Pr^{\vsched_x}_{\cM}(\Diamond \effunc)$ has the form
	$q_x + q$ where
	$q_x=\Pr^{\sched_x}_{\cM}(\Diamond (u \wedge \Diamond \effunc))$
	and
	$q= \Pr^{\sched_x}_{\cM}((\neg u) \Until \effunc)$.
	With $p_{x,c}=\Pr^{\vsched_x}_{\cM}(\Diamond c)$
	for $c\in C$
	and $p_d = \Pr^{\vsched}_{\cM}(\Diamond d)$ for $d\in D$,
	let
	$$
	v_x  = \sum_{c\in C} p_{x,c}\cdot w_c
	\qquad \text{and} \qquad 
	v = \sum_{d\in D} p_{d}\cdot w_d
	$$
	As $y$ is fixed, the values
	$p_{y},p_{y,c},q_{y},v_{y}$ can be seen as
	constants.
	Moreover,
	the values $p_x,p_{x,c},q_x,v_x$ differ from
	$p_y,p_{y,c},q_y,v_y$
	only by the factor $\frac{x}{y}$.
	That is:
	\begin{center}
		$p_x=p_y\frac{x}{y}$, \ \
		$p_{x,c}=p_{y,c}\frac{x}{y}$, \ \   
		$q_x=q_y\frac{x}{y}$ \ \ and \ \
		$v_x=v_y\frac{x}{y}$.
	\end{center}
	Thus,
	$f(x)$ has the following form:
	\begin{eqnarray*}
		f(x) & = &
		(p_x{+}p)(q_x{+}q) - \bigl(1{-}(p_x{+}p)\bigr)(v_x{+}v)
		\\[1ex]
		& = &
		\underbrace{p_x q_x {+} p_xv_x}_{\mathfrak{a}x^2} +
		\underbrace{p_x(q+v) + q_xp - v_x}_{\mathfrak{b}x} +
		\underbrace{pq -v +pv}_{\mathfrak{c}}
		\\[1ex]
		& = & \mathfrak{a}x^2 + \mathfrak{b}x + \mathfrak{c}
	\end{eqnarray*}
	For the value $\mathfrak{a}$, we have
	$\mathfrak{a}x^2=p_x q_x {+} p_x v_x $ and hence
	$\mathfrak{a}= \frac{1}{y^2}(p_yq_y + p_yv_y)>0$.
	But then the second derivative $f''(x)=2\mathfrak{a}$ of
	$f$ is positive, which yields
	that $f$ has a global minimum at some point $x_0$ and is strictly
	decreasing for $x < x_0$ and strictly increasing for $x> x_0$.
	As $f(0)$ and $f(1)$ are both negative, we obtain
	$f(x) <0$ for all $x$ in the interval $[0,1]$.
	But this contradicts $f(y) \geqslant 0$.
	
	This yields that at least one of the schedulers $\vsched_0$
	or $\vsched_1$ witnesses the violation of the GPR condition.
	Thus, we can define $\tsched_i\in \{\vsched_0,\vsched_1\}$ accordingly.
	
	The number of states $k$ in $U$ is bounded by the number of states in $S$. In each iteration of the above construction, the function value $f(0)$ is sufficient to determine one of the schedulers $\vsched_0$
	and $\vsched_1$ witnessing the violation of the GPR condition. So, the procedure has to compute the values in condition \eqref{GPR-1} for $k$-many MR-schedulers and update the scheduler afterwards. The update can easily be carried out in polynomial time. Hence, the total run-time of all $k$ iterations is polynomial as well.
	\qed
\end{proof}

\MRschedulerlift*

\begin{proof}
Let $\sched$ be an MR-scheduler for $\MEC{\cM}$ such that $\sched(s_{\cE})(\tau) \in \{0,1\}$ for each MEC $\cE$ of $\cM$.
First, we consider the following extension $\cM^\prime$ of $\cM$: The state space of $\cM$ is extended by a new terminal state $\bot$ and 
a fresh action $\tau$ is enabled in each state $s$ that belongs to a MEC of $\cM$. Action $\tau$ leads to $\bot$ with probability $1$. All remaining transition probabilities are as in $\cM$.
So, $\cM^\prime$ is obtained from $\cM$ by allowing a transition to a new terminal state $\bot$ as in the MEC-quotient from each state that belongs to a MEC.

Now, we first provide a finite-memory scheduler $\tsched$ for $\cM^\prime$ that leaves each MEC $\cE$ for which $\sched(s_{\cE})(\tau)=0$ via the state action pair $(s,\alpha)$ with probability $\sched(s_{\cE})(\alpha)$. Recall that we assume that each action is enabled in at most one state and that the actions enabled in the state $s_{\cE}$ in $\MEC{\cM}$ are precisely the actions that are enabled in some state of $\cE$ and that do not belong to $\cE$ (see Appendix \ref{app:MEC-quotient})

The scheduler $\tsched$ is defined as follows:
In all states that do not belong to a MEC $\cE$ of $\cM$ with $\sched(s_{\cE})(\tau)=0$, the behavior of $\tsched$ is memoryless:
For each state $s$ of $\cM$ (and hence of $\cM^\prime$) that does not belong to a MEC, $\tsched(s)=\sched(s)$. For each state $s$ in an end component 
$\cE$ of $\cM$ with $\sched(s_{\cE})(\tau)=1$, we define $\tsched(s)(\tau)=1$.
If a MEC $\cE$ of $\cM$ with $\sched(s_{\cE})(\tau)=0$ is entered, $\tsched$ makes use of finitely many memory modes as follows:
Enumerate the state action pairs $(s,\alpha)$ where $s$ belongs to $\cE$, but $\alpha$ does not belong to $\cE$, and for which $\sched(s_{\cE})(\alpha)>0$ by 
$(s_1,\alpha_1)$, \dots, $(s_k,\alpha_k)$ for some natural number $k$. Further, let $p_i\eqdef \sched(s_{\cE})(\alpha_i)>0$ for all $1\leq i \leq k$.
By assumption $\sum_{1\leq i \leq k} p_i =1$.

When entering $\cE$, the scheduler works in $k$ memory modes $1$, \dots, $k$ until an action $\alpha$ that does not belong to $\cE$ is scheduled
starting  in memory mode $1$.
In each memory mode $i$,
$\tsched$ follows an MD-scheduler for $\cE$ that reaches $s_i$ with probability $1$ from all states of $\cE$.
Once, $s_i$ is reached, $\tsched$ chooses action $\alpha_i$ with probability 
\[
q_i \eqdef \frac{p_i}{1-\sum_{j<i} p_j}.
\]
Note that this means that $\tsched$ leaves $\cE$ via $(s_k,\alpha_k)$ with probability $1$ if it reaches the last memory mode $k$.
As $\tsched$ behaves in a memoryless deterministic way in each memory mode, it leaves the end component $\cE$ after finitely many steps in expectation.
Furthermore, for each $i\leq k$, it leaves $\cE$ via $(s_i,\alpha_i)$ precisely with probability $(1-\sum_{j<i} p_j)\cdot q_i = p_i$.
As the behavior of $\sched$ in $\MEC{\cM}$ is hence mimicked by $\tsched$ in $\cM^{\prime}$, we conclude that the expected frequency of all actions of $\cM$ that do not belong to an end component is the same in $\cM^\prime$ under $\tsched$ and in $\MEC{\cM}$ under $\sched$.

As each end component of $\cM^\prime$ is either left directly via $\tau$ under $\tsched$ or after finitely many steps in expectation as just described,
the expected frequency of each state-action pair of $\cM^\prime$ under $\tsched$ is finite.
In the terminology of \cite{Kallenberg20}, the scheduler $\tsched$ is \emph{transient}.
By \cite[Theorem 9.16]{Kallenberg20}, this implies that there is a MR-scheduler $\usched$ for $\cM^\prime$ under which the expected frequency of all state-action pairs is the same as under $\tsched$. So, for this scheduler $\usched$,  the expected frequency in $\cM^\prime$ of all actions $\alpha$ of $\cM$ that do not belong to an end component is the same as under $\sched$ in $\MEC{\cM}$.

Finally, we modify $\usched$ such that it becomes a scheduler for $\cM$: For each end component $\cE$ of $\cM$ with $\sched(s_{\cE})(\tau)=1$, we fix a memoryless scheduler $\usched_{\cE}$ that does not leave the end component. Now, whenever a state $s$ in such an end component is visited, the modified scheduler switches to the behavior of $\usched_{\cE}$ instead of choosing action $\tau$ with probability $1$. Clearly, this does not affect the expected frequency of actions of $\cM$ that do not belong to an end component and hence the modified scheduler is as claimed in the theorem. \qed
\end{proof}

\begin{remark}
The proof of Theorem \ref{thm:MRscheduler_lift} above provides an algorithm  how to obtain the scheduler $\tsched$ from $\sched$. The number of memory modes of the intermediately constructed finite-memory scheduler is bounded by the number of state-action pairs of $\cM$. Further, in each memory mode during the traversal of a MEC, the scheduler behaves in a memoryless deterministic way. Hence, the induced Markov chain is of size polynomial in the size of the MDP $\cM$ and the representation of the scheduler $\sched$. Therefore, also the expected frequencies of all state-action pairs under the intermediate finite-memory scheduler and hence under $\tsched$ can be computed in time polynomial in the size of the MDP $\cM$ and the representation of the scheduler $\sched$. 
So, also the scheduler $\tsched$ itself which can be derived from these expected frequencies can be computed in polynomial time from $\sched$.

Together with Lemma \ref{lem:MR-sufficient-global-cause}, this means that $\tsched$ and hence the scheduler with two memory modes whose existence is stated in Theorem \ref{thm:MR-sufficient-GPR} can be computed from a solution to the constraint system (1)-(5) from Section \ref{sec:check-GPR} in time polynomial in the size of the original MDP and the size of the representation of the solution to (1)-(5).
\Ende
\end{remark}

\section{Omitted Proofs and Details of Section \ref{sec:criteria}}\label{app:section_5}

\subsection{Proofs of Section \ref{sec:comp-acc-measures-fixed-cause}}
	\label{app:comp-acc-measures}

The following lemma shows that all three quality measures are
preserved by the switch from $\cM$ to
$\wminMDP{\cM}{\Cause}$.

\begin{lemma}
	\label{lemma:accuracy-measures-M-and-Mcause} 
	If $\Cause$ is an SPR or a GPR cause then:
	\begin{center}
		\begin{tabular}{rcl}
			$\recall_{\cM}(\Cause)$ & \ = \ & $\recall_{\wminMDP{\cM}{\Cause}}(\Cause)$
			\\
			$\covratio_{\cM}(\Cause)$ & = & $\covratio_{\wminMDP{\cM}{\Cause}}(\Cause)$
			\\
			$\fscore_{\cM}(\Cause)$ & = & $\fscore_{\wminMDP{\cM}{\Cause}}(\Cause)$
		\end{tabular}  
	\end{center}  
\end{lemma}

\begin{proof}
	``$\leqslant$'':  
	Each scheduler for $\wminMDP{\cM}{\Cause}$ can be viewed as
	a scheduler $\sched$ for $\cM$ that behaves as an MD-scheduler minimizing the probability
	for reaching an effect state from every state in $\Cause$ and we have:
	\begin{center}
		\begin{tabular}{rcl}
			$\recall_{\cM}^{\sched}(\Cause)$ & \ = \ &
			$\recall^{\sched}_{\wminMDP{\cM}{\Cause}}(\Cause)$ \\
			$\covratio_{\cM}^{\sched}(\Cause)$ & = &
			$\covratio^{\sched}_{\wminMDP{\cM}{\Cause}}(\Cause)$ \\
			$\precision_{\cM}^{\sched}(\Cause)$ & = &
			$\precision^{\sched}_{\wminMDP{\cM}{\Cause}}(\Cause)$ \\
		\end{tabular}
	\end{center}
	and therefore:
	\begin{center}
		\begin{tabular}{rcl}
			$\fscore_{\cM}^{\sched}(\Cause)$ & \ = \ &
			$\fscore^{\sched}_{\wminMDP{\cM}{\Cause}}(\Cause)$
		\end{tabular}  
	\end{center}
	We obtain
	$\recall_{\cM}(\Cause)\leqslant \recall_{\wminMDP{\cM}{\Cause}}(\Cause)$
	and the analogous statements for the coverage ratio and the f-score.

	``$\geqslant$'':  
	Let $\sched$ be a scheduler of
	$\cM$. Let $\tsched = \tsched_{\sched}$
	the scheduler of $\cM$ that behaves as $\sched$ until the first
	visit to a state in $\Cause$. As soon as $\tsched$ has reached $\Cause$, it
	behaves as an MD-scheduler minimizing the probability to reach $\Effect$.
	Recall and coverage
	under $\tsched$ and $\sched$ have the form:
	$$ 
	\begin{array}{rclcrcl}
		\recall_{\cM}^{\sched}(\Cause)
		& \ = \ & \frac{x}{x + q}
		& \ \ \ \ &
		\covratio_{\cM}^{\sched}(\Cause) & \ = \ & \frac{x}{q}
		\\
		\recall_{\cM}^{\tsched}(\Cause)
		& = & \frac{y}{y + q}
		& & 
		\covratio_{\cM}^{\sched}(\Cause) & = & \frac{y}{q}
	\end{array}  
	$$
	where  $x \geqslant y$ (and $q=\mathsf{fn}^{\sched}$).
	Considering $\tsched$ as a scheduler of $\cM$ and of
	$\wminMDP{\cM}{\Cause}$, we get:
	$$
	\begin{array}{rcccl}
		\recall_{\cM}^{\sched}(\Cause)
		& \geqslant &
		\recall_{\cM}^{\tsched}(\Cause)
		& = &
		\recall_{\wminMDP{\cM}{\Cause}}^{\tsched}(\Cause)
		\\
		\covratio_{\cM}^{\sched}(\Cause)
		& \geqslant &
		\covratio_{\cM}^{\tsched}(\Cause)
		& = &
		\covratio_{\wminMDP{\cM}{\Cause}}^{\tsched}(\Cause)
	\end{array}
	$$
	This implies:
	$$
	\begin{array}{rcl}
		\recall_{\cM}^{\sched}(\Cause) & \ \geqslant \ &
		\recall_{\wminMDP{\cM}{\Cause}}(\Cause)
		\\
		\covratio_{\cM}(\Cause) & \geqslant &
		\covratio_{\wminMDP{\cM}{\Cause}}(\Cause)
	\end{array}
	$$     
	With similar arguments we get:
	$$
	\begin{array}{rcccl}
		\precision_{\cM}^{\sched}(\Cause)
		& \geqslant & 
		\precision_{\cM}^{\tsched}(\Cause)
		& = &
		\precision_{\wminMDP{\cM}{\Cause}}^{\tsched}(\Cause)
	\end{array}
	$$ 
	As the harmonic mean viewed as a function
	$f: \Real_{>0}^2 \to \Real$, $f(x,y) = 2 \frac{xy}{x{+}y}$
	is monotonically increasing in both arguments
	(note that $\frac{df}{dx} = \frac{y^2}{x{+}y}>0$ and
	$\frac{df}{dy} = \frac{x^2}{x{+}y}>0$),
	we obtain:
	$$
	\fscore_{\cM}^{\sched}(\Cause)
	\ \geqslant \ 
	\fscore_{\cM}^{\tsched}(\Cause)
	\ = \
	\fscore_{\wminMDP{\cM}{\Cause}}^{\tsched}(\Cause)
	$$
	This yields
	$\fscore_{\cM}(\Cause) \geqslant \fscore_{\wminMDP{\cM}{\Cause}}(\Cause)$.
	\qed
\end{proof}

\begin{lemma}
	\label{lemma:accuracy-measures-M-and-Mcause2} 
	Let $\cN$ be the MEC-quotient of $\wminMDP{\cM}{\Cause}$ for some MDP $\cM$ with a set of terminal states $\Effect$ and 
	 an SPR or a GPR cause $\Cause$. Then:
	\begin{center}
		\begin{tabular}{rcl}
			$\recall_{\cN}(\Cause)$ & \ = \ & $\recall_{\wminMDP{\cM}{\Cause}}(\Cause)$
			\\
			$\covratio_{\cN}(\Cause)$ & = & $\covratio_{\wminMDP{\cM}{\Cause}}(\Cause)$
			\\
			$\fscore_{\cN}(\Cause)$ & = & $\fscore_{\wminMDP{\cM}{\Cause}}(\Cause)$
		\end{tabular}  
	\end{center}  
\end{lemma}

\begin{proof}
Analogously to the proof of Lemma \ref{lem:probabilities_MEC-quotient}. \qed
\end{proof}

This lemma now allows us to work under assumptions (A1)-(A3) when addressing problems concerning the quality measures for a fixed cause set.

\compQ*

\begin{proof}
	$\cM$ has a scheduler $\sched$
	with $\Pr^{\sched}_{\cM}(\Diamond U)>0$ and $\Pr^{\sched}_{\cM}(\Diamond V)=0$
	if and only if
	the transformed MDP $\cN$ in Section \ref{sec:comp-quotient}
	(Max/min ratios of reachability probabilities for disjoint sets of terminal states)
	has an EC containing at least one $U$-state.
	Therefore we then have 
	$$\mathrm{E}^{\max}_{\cN}(\boxplus V) = +\infty.$$
	Otherwise, $$\mathrm{E}^{\max}_{\cN}(\boxplus V)=
	1/\mathrm{E}^{\min}_{\cN}(\boxplus V).$$
	
	For the following we only consider $\ratio{\min}{\cM}{U,V} = \mathrm{E}^{\min}_\cN(\boxplus V)$ since the arguments for the maximum are similar.
	First we show
	$\ratio{\min}{\cM}{U,V}  \leqslant \mathrm{E}^{\min}_{\cN}(\boxplus V)$.
	For this, we consider an arbitrary scheduler $\sched$ for $\cM$.
	Let
	\begin{align*}
		x  &= \Pr^{\sched}_{\cM}(\Diamond U)& 
		p & = \Pr^{\sched}_{\cM}(\Diamond V) &
		q & = 1 - x - p 
	\end{align*}
	For $p>0$ we have
	\[\frac{\Pr^{\sched}_{\cM}(\Diamond U)}
	{\Pr^{\sched}_{\cM}(\Diamond V)}    
	\ \ = \ \
	\frac{x}{p}\]
	Let $\tsched$ be the scheduler that behaves as $\sched$
	in the first round and after each reset.
	Then:
	\begin{equation}
		\label{big-sum} 
		\mathrm{E}^{\tsched}_{\cN}(\boxplus V) \ \ = \ \
		\sum_{n=0}^{\infty}
		\sum_{k=0}^{\infty}
		n \cdot x^n \cdot
		\left( \!\!\!\begin{array}{c}
			n{+}k \\
			k
		\end{array} \!\!\!\right)
		q^k
		\cdot p
		\ \ \stackrel{\text{(*)}}{=} \ \
		\frac{x}{p}
		\tag{$\ddagger$} 
	\end{equation}  
	where (*) relies on some basic calculations
	(see Lemma \ref{second-basic-fact}).
	This yields: 
        $$
          \ratio{\sched}{\cM}{U,V} \ = \ \frac{x}{p} \ = \ 
          \mathrm{E}^{\tsched}_{\cN}(\boxplus V)
          \ \geqslant \ \mathrm{E}^{\min}_{\cN}(\boxplus V)
        $$
        Hence,
        $\ratio{\min}{\cM}{U,V}\geqslant \mathrm{E}^{\min}_{\cN}(\boxplus V)$.

	To see why $\mathrm{E}^{\min}_{\cN}(\boxplus V) \geqslant \ratio{\min}{\cM}{U,V}$, we use the fact that there is an MD-scheduler
	$\tsched$ for $\cN$ such that
	$\mathrm{E}^{\tsched}_{\cN}(\boxplus V)
	= \mathrm{E}^{\min}_{\cN}(\boxplus V)$.
	$\tsched$ can be viewed as an MD-scheduler for the original MDP $\cM$.
	Again we can rely on \eqref{big-sum} to obtain that:
	$$
	\mathrm{E}^{\tsched}_{\cN}(\boxplus V) \ \ = \ \
	\frac{\Pr^{\tsched}_{\cM}\bigl(\Diamond U \bigr)}
	{\Pr^{\tsched}_{\cM}\bigl(\Diamond V \bigr)}
        \ \ = \ \ \ratio{\tsched}{\cM}{U,V}
        \ \ \geqslant \ \ \ratio{\min}{\cM}{U,V}
	$$
	But this yields
	$\mathrm{E}^{\min}_{\cN}(\boxplus V) 
          \ \geqslant  \ \ratio{\min}{\cM}{U,V}$.

	As stated in the main document
	we can now rely on known results \cite{BT91,Alfaro-CONCUR99,LICS18-SSP} to compute
	$\mathrm{E}^{\min}_{\cN}(\boxplus V)$ and $\mathrm{E}^{\max}_{\cN}(\boxplus V)$
	in polynomial time.
	\qed
\end{proof}

\begin{lemma}
	\label{second-basic-fact}
	Let
	$x,y,z \in \Real$ with $x > 0$ and $q,p <1$ such that $x{+}q{+}p=1$.
	Then:
	$$
	\sum_{n=0}^{\infty}
	\sum_{k=0}^{\infty}
	n \cdot x^n \cdot
	\left( \!\!\!\begin{array}{c}
		n{+}k \\
		k
	\end{array} \!\!\!\right)
	q^k
	\cdot p
	\ \ = \ \
	\frac{x}{p}
	$$ 
\end{lemma}

\begin{proof}
	We first show for $0 < q <1$, $n\in \Nat$ and
	\begin{eqnarray*}
		a_n & \eqdef &
		\sum_{k=0}^{\infty}
		\left(\!\!\!\begin{array}{c}
			n{+}k \\
			k
		\end{array} \!\!\!\right)
		q^k,
	\end{eqnarray*}
	we have
	$$
	a_n = \frac{1}{(1{-}q)^{n+1}}
	$$
	This is done by induction on $n$. The claim is clear for $n{=}0$.
	For the step of induction we use:
	$$
	\left( \!\!\!\begin{array}{c}
		n{+}1{+}k \\
		k
	\end{array} \!\!\!\right)
	\ \ = \ \
	\left( \!\!\!\begin{array}{c}
		n{+}k \\
		k
	\end{array} \!\!\!\right)
	\ + \
	\left( \!\!\!\begin{array}{c}
		n{+}k \\
		k{-}1
	\end{array} \!\!\!\right)
	\ \ = \ \
	\left( \!\!\!\begin{array}{c}
		n{+}k \\
		k
	\end{array} \!\!\!\right)
	\ + \
	\left( \!\!\!\begin{array}{c}
		(n{+}1)+(k{-}1) \\
		k{-}1
	\end{array} \!\!\!\right)    
	$$
	But this yields $a_{n+1}= a_n + q\cdot a_{n+1}$.
	Hence:
	$$a_{n+1} \ = \ \frac{a_n}{1{-}q}
	$$
	The claim then follows directly from the induction hypothesis. 
	The statement of Lemma \ref{second-basic-fact} now follows 
        by some basic calculations and
	the preliminary induction.
	\begin{eqnarray*}
		\sum_{n=0}^{\infty}
		\sum_{k=0}^{\infty}
		n \cdot x^n \cdot
		\left( \!\!\!\begin{array}{c}
			n{+}k \\
			k
		\end{array} \!\!\!\right)
		q^k
		\cdot p
		& = &
		\sum_{n=0}^{\infty}
		n \cdot x^n \cdot \frac{1}{(1{-}q)^{n+1}} \cdot p
		\\
		\\[0.5ex]
		& = &
		\frac{p}{1{-}q} \cdot
		\sum_{n=0}^{\infty} n \cdot \left( \frac{x}{1{-}q} \right)^n
		\\
		\\[0.5ex]
		& = &
		\frac{p}{1{-}q} \cdot
		\frac{\displaystyle  \frac{x}{1{-}q} }
		{\displaystyle \ \Bigl(1-\frac{x}{1{-}q}\Bigr)^2 \ }
		\\
		\\[0.5ex]
		& = &
		\frac{px}{(1{-}q{-}x)^2}  
		\ \ \ = \ \ \
		\frac{px}{p^2} \ \ \ = \ \ \ \frac{x}{p}  
	\end{eqnarray*}
	where we use $p=1{-}q{-}x$.
	\qed
\end{proof}  

In the sequel, we will use the following lemma.

\begin{lemma}\label{lem:fscore=0}
  Let $\Cause$ be an SPR or a GPR cause. Then, the following three statements
  are equivalent:
  \begin{enumerate}
  \item [(a)] $\recall(\Cause)=0$
  \item [(b)] $\fscore(\Cause)=0$
  \item [(c)] There is a scheduler $\sched$ such that
    $\Pr^{\sched}_{\cM}(\Diamond \Effect)>0$ and
    $\Pr^{\sched}_{\cM}(\Diamond \Cause)=0$.
  \end{enumerate}
\end{lemma}

\begin{proof}
Let $C=\Cause$.
Using results of \cite{TACAS14-condprob,Maercker-PhD20}, there exist schedulers $\tsched$ and $\usched$
with
\begin{itemize}
\item
  $\Pr^{\tsched}_{\cM}(\Diamond \Effect)>0$ 
  and 
  $\Pr^{\tsched}_{\cM}(\ \Diamond C \ | \Diamond \Effect \ )
   = 
   \inf_{\sched} \Pr^{\sched}_{\cM}(\ \Diamond C \ | \Diamond \Effect \ )$
  where $\sched$ ranges over all schedulers with positive
  effect probability,
\item
  $\Pr^{\usched}_{\cM}(\Diamond C)>0$ 
  and 
  $\Pr^{\usched}_{\cM}(\ \Diamond \Effect \ | \Diamond C \ )
   = 
   \inf_{\sched} \Pr^{\sched}_{\cM}(\ \Diamond \Effect \ | \Diamond C \ )$
  where $\sched$ ranges over all schedulers with
  $\Pr^{\sched}_{\cM}(\Diamond C)>0$.
\end{itemize}
In particular,
$\recall(C)=\Pr^{\tsched}_{\cM}(\ \Diamond C \ | \Diamond \Effect \ )$
and
$\precision(C)=\Pr^{\usched}_{\cM}(\ \Diamond \Effect \ | \Diamond C \ )$.
By the GPR condition applied to $\usched$ and $\tsched$
(recall that each SPR cause is a GPR cause too, 
see Lemma \ref{lemma:strict-implies-global}),
we obtain the following statements (i) and (ii):
\begin{description}
\item [\text{\rm (i)}]
  $p \ \eqdef \ \precision(C) \ > \ 0$
\item [\text{\rm (ii)}]
  If $\Pr^{\tsched}_{\cM}(\Diamond C)>0$
  then $\Pr^{\tsched}_{\cM}(\Diamond C \wedge \Diamond \Effect)>0$
  and therefore
  $\recall(C) >0$.
\end{description}
Obviously,
if there is no scheduler $\sched$ as in statement (c) then
$\Pr^{\tsched}_{\cM}(\Diamond C)>0$.
Hence, as a consequence of (ii) we obtain:
\begin{description}
\item [\text{\rm (iii)}]
  If there is no scheduler $\sched$ as in statement (c)
  then 
  $\recall(C) >0$.
\end{description}

``(a) $\Longrightarrow$ (b)'':
We prove $\fscore(C)>0$ implies $\recall(C)>0$.
If $\fscore(C)>0$ then, by definition of the f-score, there is no scheduler
$\sched$ as in statement (c).
But then $\recall(C)>0$ by statement (iii).

``(b) $\Longrightarrow$ (c)'':
Let $\fscore(C)=0$.
Suppose by contradiction that there is no scheduler as in (c).
Again by (iii) we obtain $\recall(C) >0$.
But then, for each scheduler $\sched$ with $\Pr^{\sched}_{\cM}(\Diamond C)>0$:
$$
  \precision^{\sched}(C) \ \geqslant \ p \ \stackrel{\text{\tiny (i)}}{>} \ 0
$$
and, with $r \eqdef \recall(C)$:
$$
  \recall^{\sched}(C) \ \geqslant \ r \ > \ 0
$$
The harmonic mean as a function 
$]0,1]^2\to \Real$, $(x,y) \mapsto 2 \frac{xy}{x+y}$ is
monotonically increasing in both arguments.
But then:
$$
  \fscore^{\sched}(C) \ \geqslant \ 
  2 \frac{p \cdot r}{p{+}r} \ > \ 0
$$
Hence, $\fscore(C) = \inf_{\sched} \fscore^{\sched}(C) \geqslant 2 \frac{p \cdot r}{p{+}r}  >  0$.
Contradiction.

``(c) $\Longrightarrow$ (a)'':
Let $\sched$ be a scheduler as in statement (c).
Then, 
$$
  \Pr^{\sched}_{\cM}(\ \Diamond C \ | \Diamond \Effect \ )  \ = \ 0.
$$
Hence:
$\recall(C) \ = \ 
  \Pr^{\min}_{\cM}(\ \Diamond C \ | \Diamond \Effect \ ) \ = \ 0$.
\qed
\end{proof}

\computecovratiofscore*

\begin{proof}
	With the simplifying assumptions (A1)-(A3) that can be made due to Lemmata  \ref{lemma:accuracy-measures-M-and-Mcause} and \ref{lemma:accuracy-measures-M-and-Mcause2}, we can express the coverage ratio as:
	$$
	\ratiocov(\Cause) \ \ = \ \
	\inf_{\sched}
	\frac{\Pr^{\sched}_{\cM}(\Diamond \effcov)}
	{\Pr^{\sched}_{\cM}(\Diamond \effuncov)}
	$$   
	where $\sched$ ranges over all schedulers  with $\Pr^{\sched}_{\cM}(\Diamond \effuncov) >0$.
	Now $\ratiocov$ has the form of the infimum in Theorem \ref{thm:comp-Q} and the claim holds.
	
	For the $\fscore(\Cause)$ we get
	after some straight-forward transformations
	$$
	\fscore^{\sched}(\Cause) \ \ = \ \
	2 \cdot
	\frac{\Pr^{\sched}_{\cM}(\Diamond (\Cause \wedge \Diamond \Effect))}
	{\Pr^{\sched}_{\cM}(\Diamond \Effect) + \Pr^{\sched}_{\cM}(\Diamond \Cause)}
	$$
	Since
	$$
	\Pr^{\sched}_{\cM}(\Diamond \Effect) \ = \
	\Pr^{\sched}_{\cM}(\Diamond (\Cause \wedge \Diamond \Effect)) +
	\Pr^{\sched}_{\cM}((\neg \Cause) \Until \Effect)
	$$  
	and
	$$  
	\Pr^{\sched}_{\cM}(\Diamond \Cause) \ = \
	\Pr^{\sched}_{\cM}(\Diamond (\Cause \wedge \Diamond \Effect)) +
	\Pr^{\sched}_{\cM}(\Diamond (\Cause \wedge \Box \neg \Effect))
	$$  
	we get
	\begin{eqnarray*}
		\frac{2}{\fscore^{\sched}(\Cause)} & = &
		\frac{\Pr^{\sched}_{\cM}(\Diamond \Effect) + \Pr^{\sched}_{\cM}(\Diamond \Cause)}
		{\Pr^{\sched}_{\cM}(\Diamond (\Cause \wedge \Diamond \Effect))}
		\\
		\\[0ex]
		& = &
		2 + \frac{\Pr^{\sched}_{\cM}(\Diamond (\Cause \wedge \Box \neg \Effect))
			+
			\Pr^{\sched}_{\cM}((\neg \Cause) \Until \Effect)}
		{\Pr^{\sched}_{\cM}(\Diamond (\Cause \wedge \Diamond \Effect))}
	\end{eqnarray*}
	$\Cause$ is fixed and thus we can also assume (A1)-(A3), since the corresponding transformation does not affect the f-score.
	Therefore
	\begin{eqnarray*}
		\Pr^{\sched}_{\cM}(\Diamond (\Cause \wedge \Diamond \Effect))
		& = & \Pr^{\sched}_{\cM}(\Diamond \effcov)
		\\[1ex]
		\Pr^{\sched}_{\cM}(\Diamond (\Cause \wedge \Box \neg \Effect))
		& = & \Pr^{\sched}_{\cM}(\Diamond \noeff_{\mathsf{fp}})
		\\[1ex]
		\Pr^{\sched}_{\cM}((\neg \Cause) \Until \Effect)
		& = & \Pr^{\sched}_{\cM}(\Diamond \effuncov).
	\end{eqnarray*}
	Thus
	\begin{align*}
		\frac{2}{\fscore^{\sched}(\Cause)} -2 & = 
		\frac{\Pr^{\sched}_{\cM}(\Diamond \noeff_{\mathsf{fp}})	+ \Pr^{\sched}_{\cM}(\Diamond \effuncov)}
		{\Pr^{\sched}_{\cM}(\Diamond \effcov)}
	\end{align*}
	The task is to compute
	\[X = \sup_\sched \frac{2}{\fscore^\sched(\Cause)}-2 = \sup_\sched \frac{\Pr^{\sched}_{\cM}(\Diamond \noeff_{\mathsf{fp}})	+ \Pr^{\sched}_{\cM}(\Diamond \effuncov)}
	{\Pr^{\sched}_{\cM}(\Diamond \effcov)},\]
	where $\sched$ ranges over all schedulers with $\Pr_{\cM}^\sched (\Diamond \effcov)>0$.
	We have \[\fscore(\Cause) = \frac{2}{X+2}.\]
	But $X$ can be expressed as a supremum in the form of Theorem \ref{thm:comp-Q}. This yields the claim that the optimal value is computable in polynomial time. 
	
	In case $\fscore(\Cause) = 0$, we do not obtain an optimal scheduler via Theorem \ref{thm:comp-Q}.
	Lemma \ref{lem:fscore=0}, however, shows that there is a scheduler $\sched$ with $\Pr^{\sched}_{\cM}(\Diamond \Effect ) > 0$ and $\Pr^{\sched}_{\cM}(\Diamond \Cause ) = 0$. Such a scheduler can be computed in polynomial time as any (memoryless) scheduler in the largest sub-MDP of $\cM$ that does not contain states in $\Cause$. (This sub-MDP can be constructed by successively removing states and state-action pairs.)
	\qed
\end{proof}

\subsection{Proofs of Section \ref{sec:opt-PR-causes}}
\label{app:opt-PR-causes}

\recalloptimalityequalsratiooptimality*

\begin{proof}
	For each scheduler $\sched$ and each set $C$ of states we have:
	$$
	\Pr^{\sched}_{\cM}(\Diamond \Effect) \ = \ p^{\sched}_C+q^{\sched}_C
	$$
	where  
	$p^{\sched}_C=\Pr^{\sched}_{\cM}\bigl((\neg C) \until \Effect \bigr)$
	and
	$q^{\sched}_C =
	\Pr^{\sched}_{\cM}\bigl(\Diamond (C \wedge \Diamond \Effect) \bigr)$.
	If $C$ is a cause where $q^{\sched}_C$ is positive then
	$$
	\ratiocov^{\sched}(C) \ = \ \frac{q^{\sched}_C}{p^{\sched}_C}
	\quad \text{and} \quad
	\relcov^{\sched}(C) \ = \ \frac{q^{\sched}_C}{p^{\sched}_C+q^{\sched}_C}
	$$
	For all non-negative reals $p,q,p',q'$ where $q,q'> 0$ we have:
	$$
	\frac{q}{p} < \frac{q'}{p'}
	\qquad \text{iff} \qquad
	\frac{q}{p+q} < \frac{q'}{p'+q'}
	$$
	Hence, if $C$ is fixed and  $\sched$ ranges over all schedulers with
	$q_C^{\sched}>0$:
	\begin{center}
		$\frac{q^{\sched}_C}{p^{\sched}_C}$ is minimal \ iff \
		$\frac{q^{\sched}_C}{p^{\sched}_C+q^{\sched}_C}$ is minimal
	\end{center}  
	Thus, if $C$ is fixed and $\sched=\sched_C$ is a scheduler achieving the
	worst-case (i.e., minimal) coverage ratio for $C$ then
	$\sched$ achieves the minimal recall for $C$, and vice versa.
	
	Let now $p_C=p_C^{\sched_C}$, $q_c=q_C^{\sched_C}$ where
	$\sched_C$ is a scheduler that minimizes the coverage ratio and
	minimizes the recall for cause set $C$.
	Then:
	\begin{center}
		$\ratiocov(C)= \frac{q_C}{p_C}$ is maximal \ iff \
		$\frac{q_C}{p_C+q_C}$ is maximal \ iff \
		$\relcov(C)$ is maximal
	\end{center}  
	where the extrema range over all SPR resp. GPR causes $C$.
	This yields the claim.
	\qed
\end{proof}

Recall that $\cC$ denotes the set of states that constitute a singleton SPR cause. The following lemma is a direct consequence of the definition of SPR causes.

\begin{lemma}[Characterization of SPR causes]
  \label{lemma:charac-strict-SPR-causes}
   For each subset $\Cause$ of $S \setminus \Effect$,
  $\Cause$ is an SPR cause if and only if
  $\Cause \subseteq \cC$ and $\Cause$ fulfills (M).
\end{lemma}

Recall that the canonical cause $\CanCause$ has been defined as the set of states $c\in \cC$ such that there is a scheduler $\sched$ with $\Pr_{\cM}^\sched((\neg \cC) \Until c)>0$.

\ThmCanCau*

\begin{proof}
  Lemma \ref{lemma:charac-strict-SPR-causes}
  yields that $\CanCause$ is a SPR cause.
  Optimality is a consequence as $\CanCause$ even yields
  path-wise optimal coverage in the following sense.
  If $C$ is a SPR cause then $C \subseteq \cC$
  (by Lemma \ref{lemma:charac-strict-SPR-causes})
  and for each path $\pi$ in $\cM$:
  \begin{itemize}
  \item
    If $\pi \models (\neg \CanCause) \Until \Effect$ then $\pi \models (\neg C) \Until \Effect$.
  \item
    If $\pi \models \Diamond (C \wedge \Diamond \Effect)$ then $\pi \models \Diamond (\CanCause \wedge \Diamond \Effect)$.
  \end{itemize}
  But then
  \begin{itemize}
  \item
  $\Pr^{\sched}_{\cM}(\Diamond (C \wedge \Diamond \Effect)) \leqslant
   \Pr^{\sched}_{\cM}(\Diamond (\CanCause \wedge \Diamond \Effect))$,
  \item
  $\Pr^{\sched}_{\cM}((\neg C) \Until \Effect)) \geqslant
   \Pr^{\sched}_{\cM}((\neg \CanCause) \Until \Effect)$
  \end{itemize}
  for every scheduler $\sched$.
  This yields the claim.
  \qed
\end{proof}

\begin{lemma}\label{lem:reformulating_fscore}
Let $\cM=(S,\Act,P,\init)$ be an MDP with a set of terminal states $\Effect$, let $C$ be an SPR cause for $\Effect$ in $\cM$, and let $\vartheta$ be a rational.
Then, $\fscore(C)>\vartheta$ iff
 \[
2(1{-}\vartheta)\Pr^{\sched}_{\cM}(\Diamond C \land \Diamond \Effect) - \vartheta \Pr^{\sched}_{\cM}(\neg \Diamond C \land \Diamond \Effect)
	-\vartheta \Pr^{\sched}_{\cM}(\Diamond C \land \neg \Diamond \Effect)  >  0 \tag{$\times$}
 \]
 for all schedulers $\sched$ for $\cM$ with $\Pr^{\sched}_{\cM}(\Diamond \Effect)>0$.
\end{lemma}

\begin{proof}
First, assume that $\fscore(C)>\vartheta$ and let $\sched$ be a scheduler with $\Pr^{\sched}_{\cM}(\Diamond \Effect) > 0$.
If $\Pr^{\sched}_{\cM}(\Diamond C)=0$, then $\fscore(C)$ would be $0$. So, $\Pr^{\sched}_{\cM}(\Diamond C) > 0$.
Then,
 \begin{align*}
  \fscore^{\sched}(C) \ \ 
 & = \ \
	2 \cdot
	\frac{\Pr^{\sched}_{\cM}(\Diamond (C \wedge \Diamond \Effect))}
	{\Pr^{\sched}_{\cM}(\Diamond \Effect) + \Pr^{\sched}_{\cM}(\Diamond C)} > \vartheta.
 \end{align*}
So,
\begin{align*}
& 2 \cdot \Pr^{\sched}_{\cM}(\Diamond (C \wedge \Diamond \Effect)) \\
> \,\, & \vartheta \cdot (Pr^{\sched}_{\cM}(\neg \Diamond C \land \Diamond \Effect) + 2\cdot  \Pr^{\sched}_{\cM}(\Diamond (C \wedge \Diamond \Effect)) + \Pr^{\sched}_{\cM}(\Diamond C \land \neg \Diamond \Effect)  )
\end{align*}
 from which we can conclude ($\times$) for $\sched$.
 
 Now, suppose that ($\times$) holds for a schedulers $\sched$ with $\Pr^{\sched}_{\cM}(\Diamond \Effect) > 0$. Let $\sched$ be a scheduler that minimizes 
 $\fscore^{\sched}(C)$. Such a scheduler exists by Theorem \ref{ratiocov-fscore-in-PTIME}.
 From ($\times$), we conclude 
 \begin{align*}
& 2 \cdot \Pr^{\sched}_{\cM}(\Diamond (C \wedge \Diamond \Effect)) \\
> \,\, & \vartheta \cdot (Pr^{\sched}_{\cM}(\neg \Diamond C \land \Diamond \Effect) + 2\cdot  \Pr^{\sched}_{\cM}(\Diamond (C \wedge \Diamond \Effect)) + \Pr^{\sched}_{\cM}(\Diamond C \land \neg \Diamond \Effect)  )
\end{align*}
and hence that $\fscore^{\sched}(C)>\vartheta$ as above. 
\qed

\end{proof}

\fscoreNPcoNP*

\begin{proof}
Let $\cM=(S,\Act,P,\init)$ be an MDP, $\Effect\subseteq S$ a set of terminal states, and $\vartheta$ a rational.
As before, let $\cC$ be the set of states $c\in S \setminus \Effect$ where $\{c\}$
	is an SPR cause.
	If $\cC$ is empty then the threshold problem is trivially
	solvable as there is no SPR cause at all.
	Suppose now that $\cC$ is nonempty.

Note that $\Pr^{\min}_{\cM,c}(\Diamond \Effect)>0$ for all $c\in \cC$. 
As the terminal states in $\Effect$ are not part of any end component of $\cM$,  no state $c\in \cC$ is contained in an end component of $\cM$ either.
Let $\cN=(S_{\cN},\Act_{\cN},P_{\cN},\init_{\cN})$ be the MEC-quotient of $\cM$ with the new additional terminal state $\bot$. The MEC-quotient $\cN$ contains the states from $\Effect$ and $\cC$.

\paragraph*{Claim 1: }
There is an SPR cause $C$ for $\Effect$ in $\cM$ with $\fscore(C)>\vartheta$ if and only if there is an SPR cause $C^\prime$ for $\Effect$ in $\cN$ with $\fscore(C^\prime)>\vartheta$.

\noindent
{\it Proof of Claim 1.}
We first observe that all reachability probabilities involved in the claim do not depend on the behavior during the traversal of MECs. Furthermore, staying inside a MEC in $\cM$ can be mimicked in $\cN$ by moving to $\bot$, and vice versa. 
More precisely, let $C\subseteq \cC$. Then,
analogously to Lemma \ref{lem:probabilities_MEC-quotient},  for each scheduler $\sched$ for $\cM$, there is a scheduler $\tsched$ for $\cN$, and vice versa,
such that 
\begin{itemize}
\item
$\Pr^{\sched}_{\cM}(\Diamond \Effect \mid (\neg C) \Until c) = \Pr^{\tsched}_{\cN}(\Diamond \Effect \mid (\neg C) \Until c)$ for all $c\in C$ for which the values are defined, 
\item
$\Pr^{\sched}_{\cM}(\Diamond \Effect) = \Pr^{\tsched}_{\cN}(\Diamond \Effect)$,
\item
$\Pr^{\sched}_{\cM}(\Diamond \Effect\mid \Diamond C) = \Pr^{\tsched}_{\cN}(\Diamond \Effect \mid \Diamond C)$ if the values are defined, and 
\item
$\Pr^{\sched}_{\cM}(\Diamond C\mid \Diamond \Effect) = \Pr^{\tsched}_{\cN}(\Diamond C \mid \Diamond \Effect)$ if the values are defined.
\end{itemize}
Hence, $C$ is an SPR cause for $\Effect$ in $\cM$ if and only if it is in $\cN$ and furthermore, if it is an SPR cause, the f-score of $C$ in $\cM$ and in $\cN$ agree. This finishes the proof of Claim~1.

\paragraph*{\it Model transformation 
      for ensuring positive effect probabilities.}
Recall that the f-score is only defined for  schedulers reaching $\Effect$ with positive probability.
Now, we will provide a further model transformation that will ensure that $\Effect$ is reached with positive probability under all schedulers. 
If this is already the case, there is nothing to do. So, we assume now that $\Pr^{\min}_{\cN,\init_{\cN}}(\Diamond  \Effect)=0$.

We define the subset $D\subseteq S_{\cN}$ by
\[
D\eqdef \{s\in S_{\cN} \mid \Pr^{\min}_{\cN,s}(\Diamond \Effect) =0\}.
\]
Note that $\init_{\cN} \in D$. For each $s\in D$, we further define 
\[
\Act^{\min}(s)=\{\alpha \in \Act_{\cN}(s) \mid P_{\cN}(s,\alpha,D)=1\}.
\]
Finally, let $E\subseteq D$ be the set of states that are reachable from $\init_{\cN}$ when only choosing actions from $\Act^{\min}(\cdot)$. Note that $E$ does not contain any states from $\cC$.

All schedulers that reach $\Effect$ with positive probability in $\cN$ have to leave the sub-MDP consisting of $E$ and the actions in $\Act^{\min}(\cdot)$ at some point. Let us call this sub-MDP $\cN^{\min}_E$.
We define the set of state-action pairs $\Pi$ that leave the sub-MDP $\cN^{\min}_E$:
\[
\Pi\eqdef \{ (s,\alpha) \mid s\in E\text{ and } \alpha\in \Act_{\cN}(s) \setminus \Act^{\min}(s)\}.
\]
We now construct a further MDP $\cK$. The idea is that $\cK$ behaves like $\cN$ after initially a scheduler is forced to  choose  a probability distribution over 
state-action pairs from $\Pi$. In this way, $\Effect$ is reached with positive probability under all schedulers.
Given an SPR cause, we will observe that for the f-score of this cause under a scheduler, it is only important how large the probabilities with which state action pairs from $\Pi$ are chosen are relative to each other while the absolute values are not important. Due to this observation, for each SPR cause $C$ and for each scheduler $\sched$ for $\cN$ that reaches $\Effect$ with positive probability, we can then construct a scheduler for $\cK$ that leads to the same recall and precision of $C$.

Formally, $\cK$ is defined as follows: The state space is $S_{\cN}\cup \{\init_{\cK}\}$ where $\init_{\cK}$ is a fresh initial state.
For all states in $S_{\cN}$, the same actions as in $\cN$ are available with the same transition probabilities. I.e., for all $s,t\in S_{\cN}$,
\[
\Act_{\cK}(s)\eqdef \Act_{\cN}(s) \text{ and } P_{\cK}(s,\alpha,t)\eqdef P_{\cN}(s,\alpha,t) \text{ for all }\alpha \in \Act_{\cK}(s).
\]
For each state-action pair $(s,\alpha)$ from $\Pi$, we now add a new action $\beta_{(s,\alpha)}$ that is enabled only in $\init_{\cK}$. 
These are all actions enabled in $\init_{\cK}$, i.e., 
\[
\Act_{\cK}(\init_{\cK}) \eqdef \{\beta_{(s,\alpha)}\mid (s,\alpha)\in \Pi\}.
\]
For each state $t\in S_{\cN}$, we define the transition probabilities under $\beta_{(s,\alpha)}$ by
\[
P_{\cK}(\init_{\cK},\beta_{(s,\alpha)}, t) \eqdef P_{\cN}(s,\alpha,t).
\]

\paragraph*{Claim 2: }
A subset $C\subseteq \cC$ that satisfies (M) is an SPR cause  for $\Effect$ in $\cN$ with $\fscore(C)>\vartheta$
 if and only if 
 for all schedulers $\tsched$ for $\cK$, we have
 \[
2(1{-}\vartheta)\Pr^{\tsched}_{\cK}(\Diamond C \land \Diamond \Effect) - \vartheta \Pr^{\tsched}_{\cK}(\neg \Diamond C \land \Diamond \Effect)
	-\vartheta \Pr^{\tsched}_{\cK}(\Diamond C \land \neg \Diamond \Effect)  >  0. \tag{$\ast$}
 \]

\noindent
{\it Proof of Claim 2.}
We first prove the direction ``$\Rightarrow$''. So, let $C$ be an SPR cause for $\Effect$ in $\cN$ with $\fscore(C)>\vartheta$.

We first observe that in order to prove ($\ast$) for all schedulers $\tsched$ for $\cK$, it suffices to consider schedulers $\tsched$ that start with a deterministic choice for state $\init_{\cK}$ and then behave in an arbitrary way. 
\begin{enumerate}
\item []
To see this, we consider the MDP $\cK_C$ that consists of two copies of $\cK$: ``before $C$'' and ``after $C$''. That is, when $\cK_C$ enters a $C$-state in the first copy (``before $C$''), it switches to the second copy (``after $C$'') and stays there forever. Let us write $(s,1)$ for state $s$ in the first copy and $(s,2)$ for the copy of state $s$ in the second copy. Thus, in $\cK_C$ the event $\Diamond C \wedge \Diamond \Effect$ is equivalent to reaching a state $(\eff,2)$ where $\eff\in \Effect$, while $\Diamond C \wedge \neg \Diamond \Effect$ is equivalent to reaching a non-terminal state in the second copy, while $\neg \Diamond C \wedge \Diamond \Effect$ corresponds to the event reaching an effect state in the first copy.

Obviously, there is a one-to-one-correspondendence of the schedulers of $\cK$ and $\cK_C$. With $\cK$ also $\cK_C$ has no end components, i.e., a terminal state will be reached almost surely under every scheduler. Furthermore, we equip $\cK_C$ with a weight function for the states that assigns 
\begin{itemize}
\item
  weight $2(1{-}\vartheta)$ to the states $(\eff,2)$ where $\eff \in \Effect$,
\item
  weight $-\vartheta$ to the states $(\eff,1)$ where $\eff \in \Effect$ and
  to the states $(s,2)$ where $s$ is a terminal non-effect state in $\cK$ 
  (and $\cK_C$), and
\item
  weight 0 to all other states.
\end{itemize}
Let $V$ denote the set of all terminal states in $\cK_C$.
Then, the expression on the left hand side of ($\ast$) equals $\mathrm{E}^{\tsched}_{\cK_C}(\boxplus V)$,
 the expected accumulated weight until reaching a terminal state under scheduler $\tsched$. Hence,  ($\ast$) holds for all schedulers $\tsched$ in $\cK$ if and only if $\mathrm{E}^{\min}_{\cK_C}(\boxplus V) >0$. 

It is well-known that the minimal expected accumulated weight in EC-free MDPs is achieved by an MD-scheduler. That is, there is an MD-scheduler $\tsched$ of $\cK_C$ such that $\mathrm{E}^{\min}_{\cK_C}(\boxplus V) = \mathrm{E}^{\tsched}_{\cK_C}(\boxplus V)$. When viewed as a scheduler of $\cK$, $\tsched$ behaves memeoryless deterministic before reaching $C$. In particular, $\tsched$'s initial choice in $\init_{\cK}$ is deterministic.
\end{enumerate}
So, let now $\tsched$ be a scheduler for $\cK$ with a deterministic initial choice in $\init_{\cK}$. Say $\tsched(\init_{\cK})(\beta_{(s,\alpha)})=1$ where $(s,\alpha)\in \Pi$. 

To construct an analogous scheduler $\sched$ of $\cN$, we pick an MD-scheduler $\usched$ of the sub-MDP $\cN^{\min}_E$ of 
$\cN$ induced by the state-action pairs $(u,\beta)$ where $u\in E$ and $\beta \in \Act^{\min}(u)$ such that there is a $\usched$-path from $\init_{\cN}$ to state $s$.

Scheduler $\sched$ of $\cN$ operates with the mode $\mathfrak{m}_1$ and
the modes $\mathfrak{m}_{2,t}$ for $t\in S_{\cN}$. In its initial mode $\mathfrak{m}_1$, scheduler $\sched$ behaves as $\usched$ as long as state $s$ has not been visited. When having reached state $s$ in mode $\mathfrak{m}_1$, then $\sched$ schedules the action $\alpha$ with probability 1. Let $t \in S_{\cN}$ be the state that $\sched$ reaches via the $\alpha$-transition from $s$. Then, $\sched$ switches to mode $\mathfrak{m}_{2,t}$ and behaves from then on as the residual scheduler $\residual{\tsched}{\varpi}$ of $\tsched$ for the $\tsched$-path $\varpi = \init_{\cK} \, \beta_{(s,\alpha)} \, t$ in $\cK$. That is, after having scheduled the action $\beta_{(s,\alpha)}$, scheduler $\sched$ behaves exactly as $\tsched$.

Let $\lambda$ denote $\sched$'s probability to leave mode $\mathfrak{m}_1$, which equals $\usched$'s probability to reach $s$ from $\init_{\cN}$. That is, $\lambda = \Pr_{\cN}^{\usched}(\Diamond s)$ when $\usched$ is viewed as a scheduler of $\cN$. 
As $E$ is disjoint from $C$ and $\Effect$,
scheduler $\sched$ stays forever in mode $\mathfrak{m}_1$ and never reaches a state in $C \cup \Effect$ with probability $1{-}\lambda$. 

As $\sched$ and $\tsched$ behave identically after choosing the state-action pair $(s,\alpha) \in \Pi$ or the corresponding action $\beta_{(s,\alpha)}$, respectively, this implies that 
\begin{itemize}
\item 
$\Pr^{\sched}_{\cN}(\Diamond C \land \Diamond \Effect) = \lambda \cdot \Pr^{\tsched}_{\cK}(\Diamond C \land \Diamond \Effect)$,
\item
$\Pr^{\sched}_{\cN}(\Diamond \Effect) = \lambda \cdot \Pr^{\tsched}_{\cK}(\Diamond \Effect)$, and 
\item
$\Pr^{\sched}_{\cN}(\Diamond C \land \neg \Diamond \Effect)  = \lambda \cdot \Pr^{\tsched}_{\cK}(\Diamond C \land \neg \Diamond \Effect)$.
\end{itemize}

As $\sched$ leaves the sub-MDP $\cN^{\min}_E$ with probability $\lambda >0$, we have $\Pr^{\sched}_{\cN}(\Diamond \Effect)>0$.
By Lemma \ref{lem:reformulating_fscore}, we can conclude that 
 \[
2(1{-}\vartheta)\Pr^{\sched}_{\cN}(\Diamond C \land \Diamond \Effect) - \vartheta \Pr^{\sched}_{\cN}(\neg \Diamond C \land \Diamond \Effect)
	-\vartheta \Pr^{\sched}_{\cN}(\Diamond C \land \neg \Diamond \Effect)  >  0.
 \]
 By the equations above, this in turn implies that 
  \[
2(1{-}\vartheta)\Pr^{\tsched}_{\cK}(\Diamond C \land \Diamond \Effect) - \vartheta \Pr^{\tsched}_{\cK}(\neg \Diamond C \land \Diamond \Effect)
	-\vartheta \Pr^{\tsched}_{\cK}(\Diamond C \land \neg \Diamond \Effect)  >  0.
 \]

For the direction ``$\Leftarrow$'', first recall that any subset of $\cC$ satisfying (M) is an SPR cause for $\Effect$ in $\cN$ (see Lemma \ref{lemma:charac-strict-SPR-causes}). Now, let $\sched$ be a scheduler for $\cN$ with $\Pr^{\sched}_{\cN}(\Diamond \Effect)>0$.
Let $\Gamma$ be the set of finite $\sched$-paths $\gamma$ in the sub-MDP $\cN^{\min}_E$  such that 
$\sched$ chooses an action in $\Act_{\cN}(\mathit{last}(\gamma))\setminus \Act^{\min}(\mathit{last}(\gamma))$ with positive probability after $\gamma$ where $\mathit{last}(\gamma)$ denotes the last state of $\gamma$.
Let
\[
q\eqdef \sum_{\gamma \in \Gamma} \qquad \sum_{\alpha\in \Act_{\cN}(\mathit{last}(\gamma))\setminus \Act^{\min}(\mathit{last}(\gamma))} P_{\cN}(\gamma)\cdot \sched(\gamma)(\alpha).
\]
 So, $q$ is the overall probability that a state-action pair from $\Pi$ is chosen under $\sched$.
 We now define a scheduler $\tsched$ for $\cK$: For each $\gamma\in\Gamma$ ending in a state $s$ and each $\alpha\in \Act_{\cN}(s)\setminus \Act^{\min}(s)$, the scheduler $\tsched$ chooses action $\beta_{(s,\alpha)}$ in $\init_{\cK}$ with probability $P_{\cN}(\gamma)\cdot \sched(\gamma)(\alpha) / q$.
 When reaching a state $t$ afterwards, $\tsched$ behaves like $\residual{\sched}{\gamma\, \alpha \, t}$ afterwards.
 Note that by definition this indeed defines a probability distribution over the actions in the initial state $\init_{\cK}$.
 
 By assumption, we know that now
   \[
2(1{-}\vartheta)\Pr^{\tsched}_{\cK}(\Diamond C \land \Diamond \Effect) - \vartheta \Pr^{\tsched}_{\cK}(\neg \Diamond C \land \Diamond \Effect)
	-\vartheta \Pr^{\tsched}_{\cK}(\Diamond C \land \neg \Diamond \Effect)  >  0.
 \]
 As the probability with which an action $\beta_{(s,\alpha)}$ is chosen by $\tsched$ for a $(s,\alpha) \in \Pi$ is $1/q$ times the probability that $\alpha$ is chosen in $s$ to leave the sub-MDP $\cN^{\min}_E$  under $\sched$ in $\cN$ and as the residual behavior is identical, we conclude that 
 \begin{align*}
& 2(1{-}\vartheta)\Pr^{\sched}_{\cN}(\Diamond C \land \Diamond \Effect) - \vartheta \Pr^{\sched}_{\cN}(\neg \Diamond C \land \Diamond \Effect)
	-\vartheta \Pr^{\sched}_{\cN}(\Diamond C \land \neg \Diamond \Effect)\\
= \,\,& 	q\cdot ( 2(1{-}\vartheta)\Pr^{\tsched}_{\cK}(\Diamond C \land \Diamond \Effect) - \vartheta \Pr^{\tsched}_{\cK}(\neg \Diamond C \land \Diamond \Effect)
	-\vartheta \Pr^{\tsched}_{\cK}(\Diamond C \land \neg \Diamond \Effect))  >  0.
 \end{align*}
By Lemma \ref{lem:reformulating_fscore}, this shows that $\fscore(C)>\vartheta$ in $\cN$ and finishes the proof of Claim 2.

\paragraph*{\it Construction of a game structure.}
We now construct a \emph{stochastic shortest path game} (see \cite{patek1999stochastic}) to check whether there is a subset $C\subseteq \cC$ in $\cK$ such that ($\ast$) holds. Such a game is played on an MDP-like structure with the only difference that the set of states is partitioned into two sets indicating which player controls which states.

The game $\cG$ has states $(S_{\cK}\times \{\yes,\no\}) \cup \cC\times\{\choice\}$.
On the subset $S_{\cK}\times\{\yes\}$, all available actions and transition probabilities are just as in $\cK$ and this copy of $\cK$ cannot be left.
More formally, for all $s,t\in S_{\cK}$ and $\alpha \in \Act_{\cK}(s)$, we have $\Act_{\cG}((s,\yes))=\Act_{\cK}(s)$ and $P_{\cG}((s,\yes),\alpha,(t,\yes ))=P_{\cK}(s,\alpha,t)$.

In the ``$\no$''-copy, the game also behaves like $\cG$ but when a state in $\cC$ would be entered, the game moves to a state in $\cC\times \{\choice\}$ instead.
In a state of the form $(c,\choice)$ with $c\in \cC$, two action $\alpha$ and $\beta$ are available. Choosing $\alpha$ leads to the state $(c,\yes)$ while choosing $\beta$ leads to $(c,\no)$ with probability $1$.

Formally, this means that for all state $s\in S_{\cK}$, we define $\Act_{\cG}((s,\no))=\Act_{\cK}(s)$ and for all actions $\alpha\in \Act_{\cK}(s)$:
\begin{itemize}
\item
$P_{\cG}((s,\no),\alpha,(t,\no))=P_{\cK}(s,\alpha,t)$
for all states $t\in S_{\cK}\setminus \cC$ 
\item
$P_{\cG}((s,\no),\alpha,(c,\choice))=P_{\cK}(s,\alpha,c)$
  for all states $c\in \cC$
\end{itemize}
For states $s\in S_{\cK}$, $c\in \cC$, and $\alpha\in \Act_{\cK}(s)$,
we furthermore define: 
\begin{center}
  $P_{\cG}((c,\choice),\alpha,(c,\yes))=P_{\cG}((c,\choice),\beta,(c,\no))=1$.
\end{center}
Intuitively speaking, whether a state $c\in \cC$ should belong to the cause set can be decided in the state $(c,\choice)$. The 
``$\yes$''-copy encodes that an effect state has been selected. 
More concretely, the ``$\yes$-copy'' is entered as soon as $\alpha$ has been
chosen in some state $(c,\choice)$ and will never be left from then on.
The ``$\no$''-copy of $\cK$ then encodes that no state $c\in \cC$ which has been selected to become a cause state has been visited so far.  
That is, if the current state of a play in $\cG$ belongs to the $\no$-copy then
in all previous decisions in the states $(c,\choice)$, action $\beta$ has been
chosen.

Finally, we equip the game with a weight structure.
All states in $\Effect\times\{\yes\}$ get weight $2(1-\vartheta)$. All remaining terminal states in $S_{\cK}\times \{\yes\}$ get weight $-\vartheta$. Further, 
all  states in $\Effect \times \{\no\}$ get weight $-\vartheta$. All remaining states have weight $0$.

 The game is played between two players $0$ and $1$. Player $0$ controls all states in $\cC\times \{\choice\}$ while player $1$ controls the remaining states.
 The goal of player $0$ is to ensure that the expected accumulated weight is $>0$.

 \paragraph*{Claim 3: }
 Player $0$ has a winning strategy ensuring that the expected accumulated weight is $>0$ in the game $\cG$ if and only if there is a subset $C\subseteq \cC$ in $\cK$ that satisfies (M) and 
  for all schedulers $\tsched$ for $\cK$, 
 \[
2(1{-}\vartheta)\Pr^{\tsched}_{\cK}(\Diamond C \land \Diamond \Effect) - \vartheta \Pr^{\tsched}_{\cK}(\neg \Diamond C \land \Diamond \Effect)
	-\vartheta \Pr^{\tsched}_{\cK}(\Diamond C \land \neg \Diamond \Effect)  >  0. \tag{$\ast$}
 \]

\noindent
{\it Proof of Claim 3.}
As $\cK$ has no end components, also in the game $\cG$ a terminal state is reached almost surely under any pair of strategies. Hence, we can rely on the results of \cite{patek1999stochastic} that state that both players have an optimal memoryless deterministic strategy.

We start by proving direction ``$\Rightarrow$'' of Claim 3. Let $\zeta$ be a memoryless deterministic winning strategy for player $0$. I.e., $\zeta$ assigns to each state in $\cC\times \{\choice\}$ an action from $\{\alpha, \beta\}$.
We define 
\[
  \cC_{\alpha} \eqdef \{c\in \cC \mid \zeta((c,\choice))=\alpha \}.
\]
Note that $\cC_{\alpha}$ is not empty as otherwise a positive expected accumulated weight in the game is not possible. (Here we use the fact that only the effect states
in the $\yes$-copy have positive weight and that the $\yes$-copy
can only be entered by taking $\alpha$ in one of the states $(c,\choice)$.)

To ensure that (M) is satisfied, we remove states that cannot be visited as the first state of this set:
\[
C\eqdef \{c\in \cC_{\alpha} \mid \cK,c \models \exists (\neg \cC_{\alpha})\Until c\}.
\]
 Note that the strategies for player $0$ in $\cG$ that correspond to the sets $\cC_{\alpha}$ and $C$ lead to exactly the same plays.

 Let $\tsched$ be a scheduler for $\cK$. This scheduler can be used as a strategy for player $1$ in $\cG$.
 Let us denote the expected accumulated weight when player $0$ plays according to $\zeta$ and player $1$ plays according to $\tsched$ by $w(\zeta,\tsched)$.
As $\zeta$ is winning for player 0 we have
$$
  w(\zeta,\tsched) > 0
$$
 By the construction of the game, it follows directly that 
 \[
 w(\zeta,\tsched) = 2(1{-}\vartheta)\Pr^{\tsched}_{\cK}(\Diamond C \land \Diamond \Effect) - \vartheta \Pr^{\tsched}_{\cK}(\neg \Diamond C \land \Diamond \Effect)
	-\vartheta \Pr^{\tsched}_{\cK}(\Diamond C \land \neg \Diamond \Effect). 
 \]
Putting things together yields:
 \[
 2(1{-}\vartheta)\Pr^{\tsched}_{\cK}(\Diamond C \land \Diamond \Effect) - \vartheta \Pr^{\tsched}_{\cK}(\neg \Diamond C \land \Diamond \Effect)
	-\vartheta \Pr^{\tsched}_{\cK}(\Diamond C \land \neg \Diamond \Effect) 
\ > \ 0 \tag{$\dagger$}
 \]
 For the other direction, suppose there is a set $C\subseteq \cC$ that satisfies (M) and ($\ast$) for all schedulers $\tsched $ for $\cK$.
 We define the MD-strategy $\zeta$ from $C$ by letting $\zeta((c,\choice))=\alpha $ if and only if $c\in C$. For any strategy $\tsched$ for player $1$, we can again view $\tsched$ also as a scheduler for $\cK$. Equation ($\dagger$) holds again and shows that the expected accumulated weight in $\cG$ is positive if player $0$ plays according to $\zeta$ against any strategy for player $1$. This finishes the proof of Claim 3.
 
\paragraph*{Putting together Claims 1-3.} 
We conclude that there is an SPR cause $C$ in the original MDP $\cM$ with $\fscore(C)>\vartheta$ if and only if player $1$ has a winning strategy in the constructed game $\cG$.
 As both players have optimal MD-strategies in $\cG$ \cite{patek1999stochastic}, the decision problem is in $\NP \cap \coNP$: We can guess the MD-strategy for player $0$ and solve the resulting stochastic shortest path problem in polynomial time \cite{BT91} to obtain an NP-upper bound. Likewise, we can guess the MD-strategy for player $1$ and solve the resulting stochastic shortest path problem to obtain the coNP-upper bound. \qed
 
\end{proof}

\measureNPhardness*

\begin{proof}
	\textit{$\PSPACE$-membership.}
	As $\NPSPACE = \PSPACE$,
	it suffices to provide a non-deterministic polynomially space-bounded
	algorithm for GPR-covratio, GPR-recall and GPR-f-score.
	The algorithms rely on the guess-and-check principle: they start
	by non-deterministically guessing a set $\Cause \subseteq S$,
	then check in polynomial space whether
	$\Cause$ constitutes a GPR cause
	(see Section \ref{sec:check}) and finally check $\relcov(\Cause) \leq \vartheta$ (with standard techniques),
	resp. $\ratiocov(\Cause) \leq \vartheta$,
	resp. $\fscore(\Cause) \leq \vartheta$
	(Theorem \ref{ratiocov-fscore-in-PTIME})
	in polynomial time.
	
	\paragraph*{\it $\NP$-membership for Markov chains.}
	$\NP$-membership for all three problems within Markov chains is straightforward as we
	may non-deterministically guess a cause and check in
	polynomial time whether it constitutes a GPR cause and satisfies the threshold
	condition for the recall, coverage ratio or f-score.
	
	\paragraph*{\it $\NP$-hardness of GPR-recall and GPR-covratio.}
	With arguments as in the proof of Lemma \ref{lemma:recall-opt=ratio-opt},
	the problems GPR-recall and GPR-covratio are polynomially interreducible
	for Markov chains.
	Thus, it suffices to prove NP-hardness of GPR-recall.
	For this, we provide a polynomial reduction from the knapsack problem.
	The input of the latter are sequences
	$A_1,\ldots,A_n,A$ and $B_1,\ldots,B_n,B$ of positive natural numbers
	and the task is to decide whether there exists a subset $I$ of $\{1,\ldots,n\}$ such that
	\begin{equation}
		\label{knapsack1}
		\sum_{i\in I} A_i \ < \ A \qquad \text{and} \qquad
		\sum_{i\in I} B_i \ \geqslant \ B
		\tag{*}
	\end{equation}
	Let $K$ be the maximum of the values $A,A_1,\ldots,A_n,B,B_1,\ldots,B_n$ and
	$N = 8(n{+}1)\cdot (K{+}1)$.
	We then define
	\begin{center}
		$a_i=\frac{A_i}{N}$, \ \ \
		$a=\frac{A}{N}$, \ \ \
		$b_i=\frac{B_i}{N}$, \ \ \
		$b=\frac{B}{N}$.
	\end{center}
	Then, $a, a_1,\ldots,a_n,b,b_1,\ldots,b_n$
	are positive rational numbers strictly smaller than $\frac{1}{8(n{+}1)}$,
	and \eqref{knapsack1} can be rewritten as:
	\begin{equation}
		\label{knapsack2}
		\sum_{i\in I} a_i \ < \ a \qquad \text{and} \qquad
		\sum_{i\in I} b_i \ \geqslant \ b
		\tag{**}
	\end{equation}
	For $i\in \{1,\ldots,n\}$, let
	\begin{center}
		$p_i=2(a_i+b_i)$ \ \ \ and \ \ \
		$w_i=\frac{b_i}{p_i}=\frac{1}{2}\cdot \frac{b_i}{a_i+b_i}$.
	\end{center}  
	Then,
	$0 < p_i < \frac{1}{2(n{+}1)}$ and $0< w_i < \frac{1}{2}$. Moreover:
	\begin{center}
		$p_i \bigl(\frac{1}{2}-w_i \bigr) = a_i$ \ \ \
		\text{and} \ \ \ $p_i \cdot w_i =  b_i$
	\end{center}
	Hence, \eqref{knapsack2} can be rewritten as:  
	\begin{center}
		$\sum\limits_{i\in I} p_i \bigl(\frac{1}{2}-w_i\bigr) \ < \ a$
		\qquad \text{and} \qquad $\sum\limits_{i\in I} p_i w_i \ \geqslant \ b$
	\end{center}
	which again is equivalent to:
	\begin{equation}
		\label{knapsack3}
		\frac{\sum\limits_{i\in I_0} p_iw_i}{\sum\limits_{i\in I_0} p_i} \ > \
		\frac{1}{2}
		\qquad \text{and} \qquad
		\sum\limits_{i\in I_0} p_i w_i \ \geqslant \ \ p_0+b
		\tag{***}
	\end{equation}
	where $p_0 = 2a$, $w_0=1$ and $I_0=I\cup\{0\}$.
	Note that $a<\frac{1}{8(n{+}1)}$ and hence $p_0 <\frac{1}{4(n{+}1)}$.
	
	Define a tree-shape Markov chain $\cM$ with non-terminal states
	$\init$, $s_0,s_1,\ldots,s_n$, 
	and terminal states
	$\eff_0,\ldots,\eff_n$, $\effuncov$ and
	$\noeff,\noeff_1,\ldots,\noeff_n$.
	Transition probabilities are as follows:
	\begin{itemize}
		\item
		$P(\init,s_i)=p_i$ for $i=0,\ldots,n$
		\item
		$P(\init,\effuncov) \ = \
		\frac{1}{2}-\sum\limits_{i=0}^n p_iw_i$
		
		\item   
		$P(\init,\noeff)=1-\sum\limits_{i=0}^n p_i - P(\init,\effuncov)$,
		\item
		$P(s_i,\eff_i)=w_i$, $P(s_i,\noeff_i)=1{-}w_i$ for $i=1,\ldots,n$
		\item
		$P(s_0,\eff_0)=1=w_0$.
	\end{itemize}
	
	Note that $p_0+p_1+\ldots+p_n  < \frac{1}{2}$ as all $p_i$'s are
	strictly smaller
	than $\frac{1}{2(n{+}1)}$. As the $w_i$'s are bounded by 1,
	this yields $0 < P(\init,\effuncov) <\frac{1}{2}$ and
	$0< P(\init,\noeff) < 1$.
	
	The graph structure of $\cM$ is indeed a tree and $\cM$ can be constructed from the values
	$A,A_1,\ldots,A_n,B,B_1,\ldots,B_n$ in polynomial time.
	Moreover, for $\Effect = \{\effuncov\}\cup\{\eff_i : i=0,1,\ldots,n\}$
	we have:
	$$
	\Pr_{\cM}(\Diamond \Effect)
	\ \ = \ \
	\sum_{i=0}^n p_iw_i + P(\init,\effuncov) \ \ = \ \ \frac{1}{2}
	$$
	As the values $w_1,\ldots,w_n$ are strictly smaller than $\frac{1}{2}$,
	we have
	$\Pr_{\cM}(\ \Diamond \Effect \ | \ \Diamond C \ ) < \frac{1}{2}$
	for each nonempty subset $C$ of $\{s_1,\ldots,s_n\}$.
	Thus, the only candidates for GPR causes are the sets
	$C_I =\{s_i : i\in I_0\}$ where $I \subseteq \{1,\ldots,n\}$
	where as before $I_0=I\cup\{0\}$.
	Note that for all states $s\in C_I$ there is a path satisfying
	$(\neg C_I)\Until s$. Thus, $C_I$ is a GPR cause if and only
	if $C_I$ satisfies the GPR condition.
	We have:
	$$
	\Pr_{\cM}(\ \Diamond \Effect \ | \ \Diamond C_I \ ) \ \ = \ \
	\frac{\sum\limits_{i\in I_0} p_i w_i}{\sum\limits_{i\in I_0} p_i}
	$$
	and
	$$
	\recall(C_I) \ \ = \ \
	\Pr_{\cM}
	(\ \Diamond (C_I \wedge \Diamond \Effect) \ | \ \Diamond \Effect \ )
	\ \ = \ \
	2 \cdot \sum_{i\in I_0} p_iw_i
	$$
	Thus, $C_I$ is a GPR cause with recall at least $2(p_0+b)$
	if and only if the two conditions in \eqref{knapsack3} hold,
	which again is equivalent to the satisfaction of the conditions
	in \eqref{knapsack1}.
	But this yields that $\cM$ has a GPR cause with recall at least
	$2(p_0+b)$ if and only if the knapsack problem is solvable
	for the input $A,A_1,\ldots,A_n,B,B_1,\ldots,B_n$.
	
	\paragraph*{\it $\NP$-hardness of GPR-f-score.}
	Using similar ideas, we also provide a polynomial reduction
	from the knapsack problem.
	Let $A,A_1,\ldots,A_n,B,B_1,\ldots,B_n$ be an input for the knapsack
	problem.
	We replace the $A$-sequence with $a,a_1,\ldots,a_n$ where
	$a=\frac{A}{N}$ and $a_i=\frac{A_i}{N}$
	where $N$ is as before.
	The topological structure of the Markov chain that we are going to construct
	is the same as in the NP-hardness proof for GPR-recall.

	We will define polynomial-time computable values
	$p_0,p_1,\ldots,p_n \in \ ]0,1[$ (where $p_i=P(\init,s_i)$),
	$w_1,\ldots,w_n \in \ ]0,1[$ (where $w_i=P(s_i,\eff_i)$)
	and auxiliary variables $\delta \in \ ]0,1[$ and $\lambda > 1$    
	such that:
	\begin{enumerate}
		\item [(1)] $p_0+p_1 + \ldots + p_n < \frac{1}{2}$
		\item [(2)] $\lambda = \frac{p_0 + \frac{1}{2} - \delta}{p_0}$ 
		\item [(3)] for all $i\in \{1,\ldots,n\}$:
		\begin{enumerate}
			\item [(3.1)]
			$a_i \ = \ p_i \bigl(\frac{1}{2}-w_i)$
			(in particular $w_i < \frac{1}{2}$)
			\item [(3.2)]
			$B_i \ = \ \frac{1}{\delta} B p_i \bigl( \lambda w_i -1)$
			(in particular $w_i > \frac{1}{\lambda}$)
		\end{enumerate}
	\end{enumerate} 
	Assuming such values have been defined, we obtain:
	\begin{eqnarray*}
		\sum_{i\in I} B_i \ \geqslant \ B
		& \ \ \text{iff} \ \ & 
		\frac{1}{\delta} B \sum_{i \in I} p_i (\lambda  w_i -1) \ \geqslant \ B
		\\
		\\[0ex]
		& \text{iff} &
		\sum_{i \in I} p_i (\lambda  w_i -1) \ \geqslant \ \delta
		\\
		\\[0ex]
		& \text{iff} &
		\lambda \sum_{i\in I} p_i w_i
		\ \geqslant \ \delta + \sum_{i\in I} p_i
	\end{eqnarray*}
	Hence:
	\begin{eqnarray*}
		\sum_{i\in I} B_i \ \geqslant \ B
		& \ \ \text{iff} \ \ & 
		\frac{\displaystyle \sum\limits_{i\in I} p_i w_i}
		{\displaystyle \delta + \sum\limits_{i\in I} p_i}
		\ \geqslant \ 
		\frac{1}{\lambda}
	\end{eqnarray*}
	For all positive real numbers $x,y,u,v$ with $\frac{x}{y}=\frac{1}{\lambda}$
	we have:
	$$
	\frac{x+u}{y+v} \geqslant \frac{1}{\lambda}
	\ \ \ \ \ \text{iff} \ \ \ \ \
	\frac{u}{v}\geqslant \frac{1}{\lambda}
	$$
	By the constraints for $\lambda$ (see (2)), we have
	$\frac{p_0}{p_0 + \frac{1}{2} - \delta}=\frac{1}{\lambda}$.
	Therefore:
	\begin{eqnarray*}
		\frac{\displaystyle \sum\limits_{i\in I} p_i w_i}
		{\displaystyle \delta + \sum\limits_{i\in I} p_i}
		\ \geqslant \ 
		\frac{1}{\lambda}
		& \ \text{iff} \ &
		\frac{\displaystyle p_0 + \sum\limits_{i\in I} p_i w_i}
		{\displaystyle (p_0 + \frac{1}{2} - \delta)
			+ \delta + \sum\limits_{i\in I} p_i}
		\ \ = \ \   
		\frac{\displaystyle p_0 + \sum\limits_{i\in I} p_i w_i}
		{\displaystyle p_0 + \frac{1}{2} + \sum\limits_{i\in I} p_i}
		\ \geqslant \ 
		\frac{1}{\lambda}    
	\end{eqnarray*}  
	As before let $w_0=1$ and  $I_0=I\cup \{0\}$. Then,
	the above yields:
	\begin{eqnarray*}
		\sum_{i\in I} B_i \ \geqslant \ B
		& \ \ \ \text{iff} \ \ \ & 
		\frac{\displaystyle  \sum\limits_{i\in I_0} p_i w_i}
		{\displaystyle   \frac{1}{2} + \sum\limits_{i\in I_0} p_i}
		\ \geqslant \ 
		\frac{1}{\lambda}    
	\end{eqnarray*}
	As in the NP-hardness proof for GPR-recall and using (3.1):
	$$
	\Pr_{\cM}(\Diamond \Effect) \ = \ \frac{1}{2} \ > \ w_i
	\qquad \text{for $i=1,\ldots,n$}
	$$
	Thus,
	each GPR cause must have the form $C_I=\{s_i : i \in I_0\}$
	for some subset $I$ of $\{1,\ldots,n\}$.
	Moreover:
	$$
	\Pr_{\cM}(\Diamond C_I) \ = \ \sum_{i\in I_0} p_i
	\qquad \text{and} \qquad
	\Pr_{\cM}(\Diamond (C_I \wedge \Diamond \Effect))
	\ = \
	\sum_{i\in I_0} p_i w_i
	$$
	So, the f-score of $C_I$ is:
	$$
	\fscore(C_I) \ \ = \ \
	2 \cdot \frac{\Pr_{\cM}(\Diamond (C_I \wedge \Diamond \Effect))}
	{\Pr_{\cM}(\Diamond \Effect) + \Pr_{\cM}(\Diamond C_I)}
	\ \ = \ \
	2 \cdot \frac{\sum\limits_{i\in I_0} p_i w_i}
	{\frac{1}{2}+\sum\limits_{i\in I_0} p_i}
	$$
	This implies:               
	\begin{eqnarray*}
		\sum_{i\in I} B_i \ \geqslant \ B
		& \ \text{iff} \ &
		\fscore(C_I) \ \geqslant \
		\frac{2}{\lambda}    
	\end{eqnarray*}
	With $p_0=2a$ and using (3.1) and arguments as in the NP-hardness proof
	for GPR-recall, we obtain:
	\begin{eqnarray*}
		\sum_{i\in I} A_i \ < \ A
		& \ \ \text{iff} \ \ &
		\text{$C_I$ is a GPR cause}
	\end{eqnarray*}
	Thus, the constructed Markov chain has a GPR cause with f-score at least
	$\frac{2}{\lambda}$ if and only if the knapsack problem is solvable
	for the input $A,A_1,\ldots,A_n,B,B_1,\ldots,B_n$.
	
	It remains to define the values $p_1,\ldots,p_n,w_1,\ldots,w_n$ and
	$\delta$. (The value of $\lambda$ is then obtained by (2).)
	(3.1) and (3.2) can be rephrased as equations for $w_i$:
	\begin{description}
		\item [\text{\rm (3.1')}] $w_i=\frac{1}{2}-\frac{a_i}{p_i}$
		\item [\text{\rm (3.2')}]
		$w_i=\frac{1}{\lambda}\bigl( \delta \frac{B_i}{Bp_i}+1 \bigr)$
	\end{description}
	This yields an equation for $p_i$:
	$$
	\frac{1}{2}-\frac{a_i}{p_i} \ \ = \ \ 
	\frac{1}{\lambda}\Bigl( \delta \frac{B_i}{Bp_i}+1 \Bigr)
	$$
	and leads to:
	\begin{equation}
         \label{equation-for-pi}
	p_i \ \ = \ \
	\frac{2\lambda}{\lambda-2} a_i \ + \
	\frac{2\delta}{\lambda-2} \cdot \frac{B_i}{B}
        \tag{****}
	\end{equation}
	We now substitute $\lambda$ by (2) and arrive at
	$$
	p_i \ \ = \ \ 
           \frac{p_0}{\frac{1}{2}-\delta}a_i \ + \ a_i + \ \frac{\delta p_0}{\frac{1}{2}-\delta}\frac{B_i}{B}.
	$$
	By choice of $N$, all $a_i$'s and $a$ are smaller than $\frac{1}{8(n{+}1)}$.
	Using this together with $p_0 = 2a$, we get:
	\begin{equation}
        \label{five-star}
	p_i \ < \ \frac{1}{4(n{+}1)(\frac{1}{2}-\delta)}\frac{1}{8(n{+}1)} \ + \ \frac{1}{8(n{+}1)} \ + \ \frac{\delta}{4(n{+}1)(\frac{1}{2}-\delta)}\frac{B_i}{B}
        \tag{*****}
	\end{equation}
	Let now $\delta =\frac{1}{8K}$ (where $K$ is as above, i.e.,
	the maximum of the values $A,A_1,\ldots,A_n$, $B$, $B_1,\ldots,B_n$).
        Then, $p_1,\ldots,p_n$  are computable
        in polynomial time, and so are the values $w_1,\ldots,w_n$ 
        (by (3.1')).  
	As $\frac{2\lambda}{\lambda-2}>2$ and using \eqref{equation-for-pi},
        we obtain 
        $p_i > 2a_i$.
        So, by (3.1') we get $0 < w_i < \frac{1}{2}$.

        It remains to prove (1). 
        Using $\delta=\frac{1}{8K}$, we obtain from \eqref{five-star}:
	$$
	p_i \ < \ \frac{1}{4(n{+}1)(\frac{1}{2}-\frac{1}{8K})}\frac{1}{8(n{+}1)} + \frac{1}{8(n{+}1)}+\frac{1}{32(n{+}1)(\frac{1}{2}-\frac{1}{8K})K}\frac{B_i}{B}  \ \eqdef \ x
	$$  
	As $\frac{1}{2}-\frac{1}{8K}\geq \frac{1}{4}$ and $\frac{B_i}{B} < K$, this yields:
	$$
	p_i \ < \ x \ < \ \frac{1}{8(n{+}1)^2}+ \frac{1}{8(n{+}1)}+\frac{1}{8(n{+}1)} \ < \ \frac{1}{2(n{+}1)}.
	$$
	But then condition (1) holds.
	\qed
\end{proof}

\end{appendix}

\end{document}